\documentclass[10pt,reqno]{amsart}
\makeatletter
\def\l@subsection{\@tocline{2}{0pt}{2.5pc}{5pc}{}}
\def\l@subsubsection{\@tocline{2}{0pt}{5pc}{7.5pc}{}}
\makeatother
\usepackage{lmodern}
\usepackage{amsbsy}
\usepackage{wasysym}
\usepackage{amsmath}
\usepackage{amsthm}
\usepackage{amssymb}
\usepackage{mathrsfs}
\usepackage{amsfonts}
\usepackage{newlfont}
\usepackage[sans]{dsfont}
\usepackage{dcolumn}
\usepackage{bm}
\usepackage{stmaryrd}
\usepackage{bbm}
\usepackage{relsize}
\usepackage{euscript}
\usepackage{eufrak}
\usepackage{palatino}
\usepackage[margin=1.0in]{geometry}
\numberwithin{equation}{section}
\newtheorem{thm}{Theorem}[section]

\newtheorem{cor}[thm]{Corollary}
\newtheorem{lem}[thm]{Lemma}
\newtheorem{prop}[thm]{Proposition}
\newtheorem{defn}[thm]{Definition}
\newtheorem{rem}[thm]{Remark}

\begin{document}
\allowdisplaybreaks{
\title[]{The fundamental equilibrium equation for gaseous stars and the Tolman-Oppenheimer-Volkoff equation--derivations and applications with emphasis on optimizational-variational methods.}
\author{Steven D. Miller}
\address{Rytalix Analytics, Strathclyde, Scotland}
\email{stevendM@ed-alumnus.net}
\maketitle
\begin{abstract}
Stars are essentially gravitationally stabilised thermonuclear reactors in hydrostatic equilibrium. The fundamental differential equation for all Newtonian gaseous stars in equilibrium is
\begin{align}
\frac{dp(r)}{dr}=-\frac{\mathscr{G}\mathcal{M}(r)\rho(r)}{r^{2}}\nonumber
\end{align}
where $p(r),\rho(r)$ are the pressure, density at radius $r$ and $\mathcal{M}(r)$ is
the mass contained within a shell of radius $r$ given by $\mathcal{M}(r)
=\int_{0}^{r}4\pi \overline{r}^{2} \rho(\overline{r})d\overline{r}$, and $\mathscr{G}$ is Newton's constant. This simple but crucial differential equation for the pressure gradient within any star, underpins much of astrophysical theory and it can derived by various methods:via a simple heuristic argument; via the Euler-Poisson equations for a self-gravitating fluid/gas; via a variational method by taking the 1st variation of the sum of the thermal and gravitational energies of the star; via the 2nd variation of the Massiue thermodynamic functional for a self-gravitating isothermal perfect-gas sphere; from conservation of the virial tensor; as the non-relativistic limit of the Tolman-Oppenheimer-Volkoff equation (TOVE). The TOVE for equilibrium of relativistic stars in general relativity can in turn be derived by various methods: from the energy-momentum conservation constraint on the Einstein equations applied to a spherically symmetric perfect fluid/gas; via a constrained optimization method on the mass and nucleon number; via a maximum entropy variational method for a sphere of self-gravitating perfect fluid/gas or radiation. An overview is given of all derivations with emphasis on variational methods. Many important applications and astrophysical consequences of the equilibrium equation are also reviewed.
\end{abstract}
\tableofcontents
\section{Introduction}
Stars are essentially gravitationally stabilized thermonuclear reactors in hydrostatic equilibrium. If stars where not in near-perfect hydrostatic equilibrium for billions of years, stable conditions permitting life to evolve on Earth (and probably elsewhere in the Universe) would simply not have existed. Hydrostatic equilibrium is an absolutely necessary condition for stable luminosity, temperature and thermonuclear fusion over vast periods of cosmic time. It is a fundamental and crucial property of stars, providing an exact balance between the gravitational force which attracts the matter toward the center, and the force due to the gas, the thermal and radiative pressure pushing outwards, which perfectly counterbalances gravity. Any slight deviation from this equilibrium will immediately lead to a rapid reaction to induce restoration of the equilibrium.

Let us suppose, for example, that a gaseous star  is arbitrarily compressed to a smaller radius. The gas becomes hotter rising the internal pressure. The higher pressure provokes an expansion re-settling the star back to its equilibrium state. Conversely, an arbitrary extension of the radius would decrease the internal temperature and pressure, the star would then contract again due to its self gravitation.  Such re-adjustments are very fast, occurring at the dynamical timescale of the order of half an hour for a typical main sequence star like the Sun, which is effectively instantaneous with respect to typical stellar lifetimes. The hydrostatic equilibrium of a any star therefore governs all its properties. It is always satisfied except in very short phases, such as the initial collapse of interstellar clouds or when it has exhausted its nuclear fuel and either collapses or supernovas.

In this article, the fundamental equation of hydrostatic equilibrium (FEHE) is derived and its consequences are studied in detail. This simple but crucial equation underpins most of astrophysical theory and in the literature is usually only derived heuristically by very simple arguments. However, here the various methods by which the equation can be derived are given in much greater and more rigorous mathematical detail, with an emphasis on variational methods whereby HE corresponds to an extremum or critical point of the total energy of the star.  The relativistic extension of the FEHE--the Tolman-Oppenheimer-Volkoff equation--is also derived by the standard method via the Einstein equations coupled to a perfect fluid/gas, but also from various constrained optimisational-variational methods, including maximum entropy methods. The motivation is due to a dissatisfaction from a lack of mathematical rigour and detail in most textbooks, papers and online lecture notes on astrophysics. (In the author's opinion.) Often derivations and calculations are only roughly outlined and one is obliged to do the calculation or proof oneself filling in all the details. 

Although this is predominantly an autodidactic-type review article, it does also contain a few original perspectives and derivations and fills in many mathematical details omitted within the literature. As such, a great deal of fundamental astrophysical theory is covered including the theory of polytropic gaseous stars and the fundamentals of Lagrangian perturbation theory. In particular, the Euler-Poisson equations for a self gravitating fluid/gas are considered in detail purely for the equilibrium case. Key applications and astrophyical consequences of the hydrostatic equilibrium equation are reviewed: the virial theorem, incompressible stars, stability and perturbations, lower-bound temperature and pressure estimates in stellar cores, polytropic gas spheres and the Lane-Emden equation, isothermal gas spheres and red giants, and the sharp Buchdal bound for incompressible relativistic stars.

\section{Heuristic Derivation Of The FEHE--Its Applications And Consequences}
In this section a (standard) heuristic derivation of the equations of hydrostatic equilibrium for self-gravitating gaseous stars is presented and its many important and crucial applications and consequences in astrophysical theory are briefly reviewed. First the relevant variables are defined.
\begin{defn}
The scenario describing the self-gravitating gas/fluid is defined as follows.
\begin{enumerate}
\item Generally, we use the functions
    $\rho:[0,T)\times\mathbf{R}^{3}\rightarrow\mathbf{R}_{\ge 0},p:[0,T)\times\mathbf{R}^{3}\rightarrow\mathbf{R}_{\ge 0}$ and $\Phi:[0,T)\times\mathbf{R}^{3}\rightarrow\mathbf{R}$ for some temporal interval $[0,T)$, to describe the spatio-temporal distributions of mass density, pressure and Newtonian potential of the gas/fluid. Then  $\rho=\rho(\mathbf{x},t),p=p(\mathbf{x},t),\Phi=\Phi(\mathbf{x},t)$.
\item The self-gravitating gas is supported within a spherical region $\mathbf{D}(\lambda(t))\equiv\mathbf{D}(t)$ defined by
\begin{align}
\mathbf{D}(\lambda(t))=\mathbf{D}(t)=supp~\rho(t,\bullet)=\lbrace(\mathbf{x}\in\mathbf{R}^{3}|
\rho(\mathbf{x},t)>0,\lambda(t)\gtrless \lambda(0)=R\rbrace\subset\mathbf{R}^{3}\nonumber
\end{align}
with varying radius $\lambda(t)\gtrless \lambda(0)=R$ depending on whether the self-gravitating gaseous sphere is collapsing or expanding, or $\lambda(t)= \lambda(0)=R$ in hydrostatic equilibrium. Then $\mathrm{V}(\mathbf{D}(t))=\tfrac{4}{3}\pi|\lambda(t)|^{3}$ is the volume at time t.
\item The (moving) boundary is $\partial\mathbf{D}(t)\subset\mathbf{D}$ with surface area $\mathrm{A}(\mathbf{D}(t))=4\pi|\lambda(t)|^{2}$. The boundary conditions are that the pressure and density vanish here so that $\rho(\mathbf{x},t)=p(\mathbf{x},t)=0$ for all $x\in\partial\mathbf{D}(t)$. If $\mathbf{D}(t)\bigcup\mathbf{D}^{c}(t)=\mathbf{R}^{3}$ then the vacuum region is defined by $\rho(\mathbf{x},t)=0$ in $\mathbf{D}^{c}(t)$ and the matter or hydrodynamic region by $\rho(\mathbf{x},t)>0,x\in\mathbf{D}(t)$. The fluid velocity $\mathbf{u}$ is defined only on $\mathbf{D}(t)$ by $\mathbf{u}\equiv u^{i}:[0,T)\times\mathbf{D}(t)\rightarrow \mathbf{R}^{3}$ and there is no definition in $\mathbf{D}(t)$. The space $\mathbf{D}(0)$ is taken to be precompact. A diffuse boundary condition can also be defined as (ref)
\begin{align}
\lim_{(\mathbf{x},t)\rightarrow(\overline{\mathbf{x}},\overline{t})}\rho(\mathbf{x},t)\nabla_{i}p(\mathbf{x},t))=0
\end{align}
where $(x,t)\in[0,T)\times\mathbf{D}(t)$ and $(\overline{x},\overline{t})\in[0,T)\times\partial\mathbf{D}(t)$, where again
$\partial\mathbf{D}(t)$ is the boundary of the volume $\mathbf{D}(t)$, that is, the moving interface between the fluid/gas matter and the exterior vacuum. This boundary condition characterizes the gradually vanishing fluid pressure per mass near the boundary.
\item
The initial data are
\begin{align}
&\rho(0,\mathbf{x})=\rho_{o}(\mathbf{x}),~~p(0,\mathbf{x})=p_{o}(\mathbf{x})~~\mathbf{x}\in\mathbf{R}^{3}\\&
\mathbf{u}(0,\mathbf{x})=\mathbf{u}_{o}(\mathbf{x}),~~\mathbf{x}\in\mathbf{D}(0)
\end{align}
\item The Newtonian potential $\Phi$ exists on $\mathbf{R}^{3}=\mathbf{D}(t)\bigcup\mathbf{D}^{c}(t)$, in both the matter and vacuum regions and obeys the Poisson equation
\begin{align}
\Delta\Phi=4\pi\mathscr{G}\rho,~~~in~~\mathbf{R}^{3}
\end{align}
with the general solution
\begin{align}
\Phi(\mathbf{x},t)=-\mathscr{G}{\int}_{\mathbf{D}(t)}\frac{\rho(\mathbf{y},t)d^{3}\mathbf{y}}{|\mathbf{x}-\mathbf{y}|},~~~~
for~(t,\mathbf{x},\mathbf{y})\in [0,T)\times\mathbf{R}^{3}
\end{align}
\item It is generally assumed that the gas is \emph{isentropic} having constant entropy or entropy per nucleon, and obeys a barytropic pressure law of the form $p=p(\rho)$, typically a polytropic gas law
\begin{align}
p=\mathsf{K}|\rho|^{\gamma},~~\mathsf{x}\in\mathbf{D}(t)\in \mathbf{R}^{3}
\end{align}
(Polytropic stars will be considered in detail in Section (5) The set $\lbrace\rho(x,t),p(x,t),\Phi(x,t)\rbrace$ is also a solution to the Euler-Poisson system of equations but this will be considered formally in Section (6). For a gas in hydrostatic equilibrium $\mathbf{u}$=0 and $\rho:\mathbf{R}^{3}\rightarrow\mathbf{R}_{\ge 0},p:\mathbf{R}^{3}\rightarrow\mathbf{R}_{\ge 0}$ and $\Phi:\mathbf{R}^{3}\rightarrow\mathbf{R}$ so that
$\rho=\rho(\mathbf{x}),p=p(\mathbf{x}),\Phi=\Phi(\mathbf{x})$. The volume of the support is now fixed so that $\mathrm{V}(\mathbf{D}(t))\mathrm{V}(\mathbf{D}(0))=\tfrac{4}{3}\pi R^{3}\equiv\tfrac{4}{3}\pi |\lambda(0)|^{3}$.
\item For self-gravitating fluid-gas systems like stars in equilibrium at radius $R$, it is also natural to assume spherical symmetry and coordinates so that $\rho=\rho(r), p=p(r),\Phi=\Phi(r)$ with central values $\rho(0),p(0)$ and boundary values $\rho(R)=p(R)=0$. The domain $\mathbf{D}(\lambda(t))=\mathbf{D}(\lambda(0))\equiv\mathbf{D}(0)=\mathbf{B}(\mathbf{x}_{o},R) $ is now a static ball with centre $\mathbf{x}_{o}$ an radius $R$. Then
    \begin{align}
    \mathbf{\mathbf{D}}(\lambda(0)=R)=\mathbf{B}(0,R)=
    \lbrace \mathbf{x},\mathbf{y}\in \mathbf{R}^{3}|d(x,y)\le R\rbrace
    \end{align}
    The mass contained within a radius $r$ is or ball $\mathbf{B}(0,r)\subset
    \mathbf{B}(0,R)$ is then $\mathcal{M}(r)$ and $\mathcal{M}(R)=M$, the total mass contained within $\mathbf{B}(0,R)$  It is defined as
\begin{align}
\mathcal{M}(r)=\int_{\mathbf{B}(0,r)}4\pi\widehat{r}^{2}\rho(\widehat{r})d{r}\equiv
\int_{0}^{r}4\pi \widehat{r}^{2}\rho(\widehat{r})d\overline{r}
\end{align}
\end{enumerate}
\end{defn}
Within this scenario, the fundamental equation(s) of hydrostatic equilibrium (FEHE) can be derived heuristically.$\bm{[1,2,3]}$ The nonlinear Euler-Poisson PDEs describing this self-gravitating system can also be formally defined but this will be done in Section (6).
\subsection{Heuristic Derivation of the FEHE}
The total mass of the gas is $M$ and the mass of fluid/gas contained within a spherical surface of radius $r$ is $\mathcal{M}(r)$. The mass of gas contained within a thin shell of thickness $dr=(r+dr)-r$ is $d\mathcal{M}(r)=4\pi r^{2}\rho(r)$. The mass within the shell experiences a gravitational force $\mathscr{F}_{grav}$ inward due to the mass of gas $\mathcal{M}(r)$. Applying Newton's inverse square law to the shell
\begin{align}
\mathscr{F}_{grav}=-\frac{\mathscr{G}\mathcal{M}(r)d\mathcal{M}(r)}{r^{2}}
=-\frac{\mathscr{G}(4\pi r^{2}\rho(r)dr)\mathcal{M}(r)}{r^{2}}
\end{align}
Due to the intrinsic pressure of the gas, the shell also experiences a buoyancy force pushing radially outwards. The outward pressure on the shell is the difference in pressures at the interior and exterior boundaries of the shell at r and r+dr so that
\begin{align}
\mathscr{F}_{buoy}=4\pi r^{2}(p(r)=p(r+dr))-4\pi r^{2}\frac{dp(r)}{dr}dr
\end{align}
and $dp(r)/dr=(p(r+dr)-p(r))dr$ for small $dr$. The equilibrium of the shell requires that
these two forces balance so that $\mathscr{F}_{grav}+\mathscr{F}_{buoy}=0$.
\begin{align}
=-\frac{\mathscr{G}(4\pi r^{2}\rho(r)dr)\mathcal{M}(r)}{r^{2}}+4\pi r^{2}\frac{dp(r)}{dr}dr=0
\end{align}
This then gives the fundamental equation of hydrostatic equilibrium for all Newtonian gaseous stars
\begin{align}
\underbrace{\frac{dp(r)}{dr}=-\frac{\mathscr{G}\mathcal{M}(r)\rho(r)}{r^{2}}}
\end{align}
\begin{rem}
Throughout the paper, whenever the FEHE (2.12) is derived for the first time from a new method or technique, it will be emphasised with an underbrace.
\end{rem}
Equations (2.8) and (2.12) are two of the fundamental equations of stellar structure theory.(REFS) The total pressure $p(r)$ in the star is the sum of the gas pressure $p_{gas}(r)$ and the photon or radiation pressure $p_{rad}(r)$ so that $\bm{[1]}$
\begin{align}
\frac{d}{dr}(p_{gas}(r)+p_{rad}(r))=-\frac{\mathscr{G}\mathcal{M}(r)\rho(r)}{r^{2}}
\end{align}
where for a mixture of radiation and a perfect gas of ionised hydrogen
\begin{align}
p(r)=p_{gas}(r)+p_{rad}(r)= \frac{k_{B}}{\mu m_{p}}\rho(r) \Theta(r)+ \frac{1}{3}a\Theta^{4}
\end{align}
where $\Theta(r)$ is the temperature of the gas, $k_{B}$ is the Boltzmann constant,'a' is the Stefan constant, $\mu$ is the mean molecular weight, and $m_{p}$ is the mass of the proton. One can also write $p_{gas}=\mathscr{R}\rho \Theta$ where $\mathscr{R}=\frac{k_{B}}{\mu m_{p}}$ is the fundamental gas constant. The gradient of the radiation pressure is determined by the radiant photon flux $\mathcal{L}(r)$ through a spherical surface by
\begin{align}
\frac{dp_{rad}(r)}{dr}=-\frac{\kappa\mathcal{L}(r)}{4\pi c r^{2}}\rho(r)
\end{align}
where
\begin{align}
d\mathcal{L}(r)=4\pi r^{2}\rho\mathfrak{T} dr
\end{align}
and $\kappa$ is the Krammer's opacity $\bm{[2,3]}$. Since $\mathfrak{T}$ is essentially the thermonuclear energy production rate per mass, we can see that the effect of thermonuclear reactions is to maintain the photon/radiation pressure gradient and thus hydrostatic equilibrium.

The hydrostatic equilibrium equations can be easily solved for an incompressible gas star when $\rho(r)$ is constant and for density profiles of the form$\rho(r)=f(r/R)$.
\begin{lem}
Let a gaseous star fill a domain $\mathbf{B}\subset\mathbf{R}^{n}$ with $V(\mathbf{B})=\tfrac{4}{3}p R^{3}$. If $\rho(r)=\rho(0)=\rho_{c}=const.$ then (2.8) and (2.12) can be solved to give $\mathcal{M}r)=\tfrac{4\pi}{3}\rho_{c}r^{3}$ and the pressure profile has the form
\begin{align}
p(r)=\frac{1}{2}\rho_{c}\beta\left(1-\left|\frac{r}{R}\right|^{2}\right)
\end{align}
with $\beta=3M/8\pi R^{3}$. The central pressure is $p_{c}=p(0)=\tfrac{1}{2}\rho_{c}\beta$.
\end{lem}
\begin{proof}
Equation (2.8) can be easily integrated to give $\mathcal{M}(r)=\tfrac{4}{3}\pi r^{3}\rho_{c}$ and $M=\mathcal{M}(R)=\tfrac{4}{3}\pi R^{3}\rho_{c}$. Then the FEHE (2.12) becomes
\begin{align}
\frac{dp(r)}{dr}=-\frac{\mathscr{G}\mathcal{M}(r)|\rho_{c}|^{2}}{r^{2}}=
-\frac{4}{3}{\mathscr{G}} \pi r|\rho_{c}|^{2}
\end{align}
Integrating
\begin{align}
\int_{p(r)}^{p(R)}dp(r)&=-p(r)=\left|-\frac{2}{3}\pi{\mathscr{G}}\rho_{c}r^{2}
\right|_{r}^{R}=-2\pi {\mathscr{G}}|\rho_{c}|^{2}R^{2}\left(1-\left|\frac{r}{R}\right|^{2}\right)
\nonumber\\&
=\frac{1}{2}|\rho_{c}|^{2}\left(1-\left|\frac{r}{R}\right|^{2}\right)
\end{align}
\end{proof}
\begin{thm}
The gravitational potential $\Phi(r)$ for a ball of fluid/gas supported within $\mathbf{D}=\mathbf{B}_{3}(R)\subset \mathbf{R}^{3}$ of radius R of uniform mass density $\rho$ inside the matter, and outside the ball in vacuum is
\begin{align}
&\Phi(x)=\frac{C}{4\pi}\int_{\mathbf{D}}\frac{\rho d^{3}y}{|x-y|}
=\rho\left(\frac{1}{2}R^{2}-\frac{1}{6}\|x|\|^{2}\right),~~\|x\|\le R\\&\Phi(x)=\frac{C}{4\pi}\int_{\mathbf{D}}\frac{\rho d^{3}y}{|x-y|}=\frac{C\rho R^{3}}{3\|x\|},~~\|x\|>R
\end{align}
\end{thm}
Outside the ball, the potential behaves as thought it is due to a point mass with all the mass of the ball concentrated at the point of origin. This property also enables gravitational interactions of stars and planets within celestial mechanics to be reduced to the problem of interacting point masses. The result holds for all dimensions $n\ge 4$ and also for logarithmic potentials of discs within $\mathbf{R}^{2}$. When the density is constant, the potential across the surface $\partial\mathbf{B}(0)$ is constant.
\subsubsection{Fundamental stellar structure equations and nuclear reactions}
The hydrostatic equilibrium equations (2.8) and (2.12) are two of the four fundamental first-order differential  equations of stellar structure theory, the other two describing radiative transport. $\mathcal{L}(r)$ is the total luminosity or radiant energy flux through a spherical surface, $\kappa(r)$ the Kramers opacity and $\mathfrak{T}(r)$ the thermonuclear energy production rate per mass. With the chemical composition fixed, then these quantities are taken to be fixed functions of the density $\rho(r)$ and temperature $\Theta(r)$ of the gas so that $p(r)=p(\rho(r),\Theta(r)),\mathcal{L}(r)(\rho(r),\Theta(r))$ and
$\mathfrak{T}=\mathfrak{T}(r,\Theta(r)$. The four fundamental differential equations governing the pressure, mass, flux and temperature gradients are then
\begin{align}
&\frac{dp(r)}{dr}=\frac{dp_{gas}(r)}{dr}+\frac{dp_{rad}(r)}{dr}
=-\frac{\mathscr{G}\mathcal{M}(r)\rho(r)}{r^{2}}\\&
\frac{d\mathcal{M}(r)}{dr}=4\pi\rho(r) r^{2}dr\\&
\frac{d\mathcal{L}(r)}{dr}=4\pi r^{2}\mathfrak{T}\rho(r)\\&
\frac{d \Theta(r)}{dr}=\frac{3\kappa(r)\rho(r)\mathcal{L}(r)}{4 c a |\Theta(r)|^{3} 4\pi r^{2}}
\end{align}
with BCs $\mathcal{M}(0)=\mathcal{L}(0)=0$ and $\rho(r)=\Theta(r)=0$. The star's structure is then determined by four unknowns in term of the variable $r$, namely $\mathcal{M}(r),\mathcal{L}(r),\rho(r)$ and $\Theta(r)$. However, in this work we are only concerned with the FEHE equations (2.8)and (2.12)and their consequences, application and derivation.

The nuclear energy production per mass is expected to obey a power law roughly of the form
\begin{align}
\mathfrak{T}(\rho,\Theta)\sim \epsilon_{1}|\rho|^{\lambda}(k_{B}\Theta)^{\mu}
\end{align}
A steady $\mathfrak{T}$ then ensures  steady pressure, flux and temperature gradients within the star. Very light nuclei like $^{1}\!\!H,^{2}\!\!H,^{3}\!\!He,^{7}\!\!Li$ have less binding energy per nucleon than medium atomic weight nuclei like nickel or iron. Therefore, energy can be extracted by the fusion of hydrogen and helium into heavier elements. While there is hydrogen in the core of the star, the dominant source of nuclear energy is from the fusion of $^{1}\!\!H$ into $^{4}\!\!He$, which has the largest binding energy of the lighter elements.

For most stars in the Universe, there are two main processes by which $^{1}\!\!H$ can fuse to $^4\!\!He$: these are the pp-chain and the CNO-cycle $\bm{[2-5]}$. In yellow main sequence stars like the Sun, the PP-chain is the main thermonuclear energy source. The main branches are
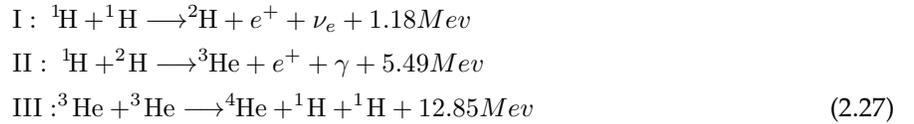
\begin{align}
&\mathrm{I}:~^{1}\!\mathrm{H}+^{1}\!\mathrm{H}\longrightarrow ^{2}\!\!\mathrm{H}+e^{+}+\nu_{e}+1.18Mev\nonumber\\&
\mathrm{II}:~^{1}\!\mathrm{H}+^{2}\!\mathrm{H}\longrightarrow ^{3}\!\!\mathrm{He}+e^{+}+\gamma+5.49Mev\nonumber\\&
\mathrm{III}:^{3}\!\mathrm{He}+^{3}\!\mathrm{He}\longrightarrow ^{4}\!\!
\mathrm{He} +^{1}\!\mathrm{H}+^{1}\!\mathrm{H}+12.85Mev
\end{align}
The pp cycle can be briefly analysed as follows:
\begin{enumerate}
\item \underline{PP-I branch}: The cross-section for this reaction is very small for protons with energy $<1MeV$. This is due to the weak nuclear force
which 'transmutes' a proton into a neutron. The average lifetime of a proton in the core of any star like the Sun is of the order of $10^10$ years which is comparable to the age of very old stars. In terms of the inner quark structure of the proton, the heavy $W^{+}$ boson which carries the weak nuclear force, transmutes a u-quark in the proton to a d-quark, thus creating a neutron. The overall cross-section for this reaction is so very small because of the weak interaction and the suppression factor (~1/100) of the mutual proton Coulomb barriers. This very small reaction rate is the reason why stars burn their nuclear fuel so slowly. The quantum mechanical calculation of probability of deuteron formation consists of two parts:(i)The penetration of the mutual Coulomb barriers by two protons in a thermal bath at temperature $\Theta\ge \Theta_{I}$, where $\Theta_{I}$ is the 'ignition temperature'.(ii) The probability of inverse $\beta$ decay with positron and neutrino emission. For this part, Bethe $\bm{[4]}$ applied the original Fermi theory of weak interactions, which is sufficient for this low energy process.
\item\underline{PP-II branch}: This branch of the cycle is a non-resonant direct capture in a thermal bath to the ground state of Helium-3, and involves only the strong nuclear force. This reaction therefore proceeds very rapidly. Other possible reactions such as $D+D\rightarrow P+^{3}H$,$D+D\rightarrow n+^{3}He$, $D+^{3}He\rightarrow p+^{4}He$ have larger cross-sections, but the overwhelming number density of protons in the cores of stellar infernos ensures that the reaction of deuterium fusing with a proton dominates the $PP-cycle$. Consequently, deuterium is virtually instantaneously destroyed in all stars.
\item\underline{PP-III branch}: The pp-chain is complete with the burning of helium-3, creating a helium-4 nucleus and also 2 protons which become potential new fuel within the stellar core. Another possible reaction is $^{3}He+D\rightarrow ^{4}He+p$ but the concentration of deuterium is so small that (PP-II) dominates. Helium-3 can also be fused with Helium-4 to create $^{7}Be$ and the berylium can in turn capture a proton to form boron-8. These reactions are important in neutrino astronomy $\bm{[2]}$.
 \end{enumerate}

For heavier and hotter (bluer) stars the main nuclear energy production process is the CNO-cycle involving carbon, nitrogen and oxygen $\bm{[2,3]}$.
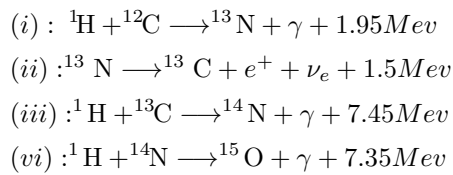
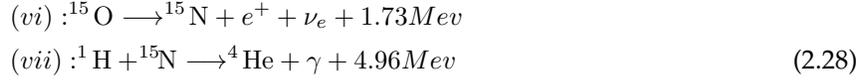
\begin{align}
&(i):~^{1}\!\mathrm{H}+^{12}\!\!\mathrm{C}\longrightarrow ^{13}\!\mathrm{N}+\gamma+1.95Mev\nonumber\\&
(ii): ^{13}\mathrm{N}\longrightarrow^{13}\mathrm{C}+e^{+}+\nu_{e}+1.5Mev\nonumber\\&
(iii):^{1}\!\mathrm{H}+^{13}\!\!\mathrm{C}\longrightarrow ^{14}\!\mathrm{N}+\gamma+7.45Mev\nonumber\\&
(vi):^{1}\!\mathrm{H}+^{14}\!\!\mathrm{N}\longrightarrow ^{15}\!\mathrm{O}+\gamma+7.35Mev\nonumber\\&
(vi):^{15}\!\mathrm{O}\longrightarrow ^{15}\!\mathrm{N}+e^{+}+\nu_{e}+1.73Mev\nonumber\\&
(vii):^{1}\!\mathrm{H}+^{15}\!\!\mathrm{N}\longrightarrow ^{4}\!\mathrm{He}+\gamma+4.96Mev
\end{align}
The $\mathrm{C,N,O}$ nuclei were created in earlier generations of stars and act like thermonuclear catalysts in that they are neither created or destroyed over any complete $\mathrm{CNO}$ cycle. The energies listed here are the total energies dumped into the stellar material, excluding the neutrino energies since the neutrinos pass unhindered through the star, but including the energies from the positron $m_{e}c^{2}$ and the electron with which it annihilates.

Hydrostatic equilibrium and nuclear energy production are in positive feedback: HE ensures that the nuclear energy production rate $\mathfrak{T}$, and hence the luminosity or energy flux $\mathcal{L}(r)$ remains steady over billions of years. The nuclear energy production in turn maintains the radiation/photon pressure gradient throughout the star which maintains the hydrostatic equilibrium. From (2.24) one has
\begin{align}
\frac{d\mathcal{L}(r)}{dr}=4\pi r^{2}\mathfrak{T}(r)\rho(r)=\frac{d\mathcal{M}(r)}{dr}\mathfrak{T}(r)
=\frac{d\mathcal{M}(r)}{dp(r)}\frac{dp(r)}{dr}\mathfrak{T}(r)
\end{align}
and the temperature, pressure and flux gradients are stable and (to a high approximation) constant and consistent.
\begin{defn}
For a star in hydrostatic equilibrium, one can estimate the nuclear time, aka the Einstein time, which is the time required to burn through its nuclear fuel (hydrogen), concentrated in a central core of mass $M_{core}$. It is $\bm{[4]}$.
\begin{align}
T_{Nuc}\sim \frac{{\epsilon} M_{core}c^{2}}{\mathcal{L}}
\end{align}
where $\mathcal{L}$ is the luminosity and $\epsilon$ is an efficiency factor, typically $\epsilon>0.7$.For the Sun, $\mathcal{L}\sim 3.9\times 10^{9}Js^{-1}$ so that $T_{Nuc}\sim 10^{10}$ years, which is around $10,000$ million years.
\end{defn}
\subsubsection{Central pressure and temperature estimates}
\begin{defn}
The Newtonian gravitational potential ${\Phi}(r)$ in the interior of a spherically symmetric distribution of matter has the gradient
\begin{align}
\nabla{\Phi}(r)=\frac{d\Phi(r)(r)}{dr}=\frac{\mathscr{G}\mathcal{M}(r)}{r^{2}}
\end{align}
Then the FEHE can be expressed as
\begin{align}
\frac{dp(r)}{dr}=-\frac{\mathscr{G}\mathcal{M}(r)\rho(r)}{r^{2}}=-\rho(r)\frac{d{\Phi}(r)}{dr}
\end{align}
Integrating with ${\mathcal{M}}(0)=0$ and $\mathcal{M}(R)=M$ the potential $\Phi(r)$ for $
r\ge R$ is $\Phi(r)=-\frac{\mathscr{G} M}{r}$ and the boundary potential is  $\Phi(R)=-\frac{\mathscr{G}M}{R}$                                                                                 \end{defn}
\begin{lem}
The FEHE can be expressed in the Poisson equation forms
\begin{align}
&\frac{1}{r^{2}}\frac{d}{dr}\left(\frac{r^{2}}{\rho(r)}\frac{dp(r)}{dr}
\right)=-4\pi\mathscr{G}\rho(r)\\&\Delta_{r}\Phi(r)=
\frac{1}{r^{2}}\frac{d}{dr}\left(r^{2}\frac{d\Phi(r)}{dr}
\right)=-4\pi{\mathscr{G}}\rho(r)
\end{align}
where (2.32) is the Emden equation.
\begin{proof}
This follows easily from (2.8) and (2.12) since
\begin{align}
\frac{r^{2}}{\rho(r)}\frac{dp(r)}{dr}=-{\mathscr{G}}\mathcal{M}(r)
=-{\mathscr{G}}\int_{0}^{r}4\pi r^{\prime 2}\rho(r^{\prime})dr^{\prime}
\end{align}
Taking the derivative
\begin{align}
\frac{d}{dr}\left(\frac{r^{2}}{\rho(r)}\frac{dp(r)}{dr}\right)
=-4\pi{\mathscr{G}} r^{2}\rho(r)
\end{align}
and (2.32) then follows. Equation (2,31) follows immediately from substituting $\tfrac{1}{\rho(r)}(dp(r)/dr)=-d\Phi(r)/dr$.
\end{proof}
\end{lem}
\begin{defn}
The Lagrangian form of the FEHE is derived quite simply. The FEHE is
\begin{align}
\frac{dp(r)}{dr}= \frac{dp(r)}{d\mathcal{M}(r)}\frac{d\mathcal{M}(r)}{dr}=
-\frac{{\mathscr{G}}\mathcal{M}(r)\rho(r)}{r^{2}}
\end{align}
so that
\begin{align}
\mathbf{D}_{\mathcal{M}}p(r)=\frac{dp(r)}{d\mathcal{M}(r)}
=-\frac{{\mathscr{G}}\mathcal{M}(r)\rho(r)}{r^{2}}\frac{dr}{d\mathcal{M}(r)}
=-\frac{\mathscr{G}G\mathcal{M}(r)\rho(r)}{r^{2}}\frac{1}{4\pi r^{2}\rho(r)}=-\frac{{\mathscr{G}}\mathcal{M}(r)}{4\pi r^{4}}
\end{align}
This form is often more useful in stellar structure codes and other applications.
\end{defn}
An immediate consequence of the hydrostatic equilibrium equation in Lagrangian form, are strict lower bound estimates on the central pressure and central temperature of the Sun or any main sequence star of mass M and radius R composed primarily of hydrogen and obeying the universal gas law $\bm{[1,2,3]}$.
\begin{thm}
Any gaseous main-sequence star like the Sun of mass $M$ and radius $R$ composed primarily of ionised hydrogen and obeying the universal/perfect gas law, has the following lower bounds for its central pressure $p_{c}=p(0)$ and temperature $\Theta_{c}=\Theta(0)$.
\begin{align}
p_{c}\ge \frac {\mathscr{G}M^{2}}{8\pi R^{4}},~~~\Theta_{c}\ge \frac{\mu k_{B}}{\mathscr{R}}\frac{\mathscr{G}M}{R}
\end{align}
where $\mathscr{R}$ is the universal gas constant, $\mu$ is the molecular weight of hydrogen, and $k_{B}$ is the Boltzmann constant.
\end{thm}
\begin{proof}
Integrating the Lagrangian form of the FEE (2.38)
\begin{align}
\int_{p(0)}^{p(R)}dp(r)=p(0)=\int_{0}^{M}\frac{\mathscr{G}\mathcal{M}(r)d\mathcal{M}(r)}{4\pi r^{2}}=\frac{\mathscr{G}M^{2}}{8\pi R^{4}}
\end{align}
since $p(R)=0, \mathcal{M}(R)=M$ and $\mathcal{M}(0)=0$. Since $1/r^{4}>1/R^{4}$ we have the lower bound estimate on the central pressure.
\begin{align}
p_{c}\ge \frac{\mathscr{G}M^{2}}{8\pi R^{4}}
\end{align}
Since the gas obeys $p(r)=\frac{\mu k_{b}}{\mathscr{R}}\Theta(r)\rho(r)$ so at the centre of the star $p(0)=\frac{\mu k_{b}}{\mathscr{R}}\Theta(0)\rho(0)$. Then using (2.41)
\begin{align}
p_{c}=p(0)=\frac{\mu k_{b}}{\mathscr{R}}\Theta(0)\rho(0) \ge \frac{\mathscr{G}M^{2}}{8\pi R^{4}}
\end{align}
In terms of the central temperature $\Theta(0)$
\begin{align}
\Theta(0)\ge \frac{\mathscr{G}M^{2}}{8\pi R^{4}}\left|\frac{1}{\rho(0)}\right|\frac{\mu k_{b}}{\mathscr{R}}\Theta(0)
\end{align}
The central density can be represented by the average density so that
$\rho(0)\sim \big\langle\rho\big\rangle=3M/4\pi R^{2}$ and the lower bound on the central temperature of the star is
\begin{align}
\Theta(0)\ge \frac{3}{8}\frac{GM}{R}\left|\frac{1}{\rho(0)}\right|\frac{\mu k_{b}}{\mathscr{R}}\Theta(0)
\end{align}
\end{proof}
These are rough estimates/bounds but for the Sun give correct the order of magnitude for pressure and temperature. The central core temperatures and pressures of stars are therefore hot enough to initiate thermonuclear fusion of hydrogen and the radiation/photon pressure which is generated as a result and maintains equilibrium. It can be see that the central temperature and pressure in the star will increase if R is decreased or M is increased. Note also that in the limit as the radius tends to zero.
\begin{align}
\lim_{R\uparrow 0}p_{c})=\lim_{R\uparrow 0}\frac{\mu k_{b}}{\mathscr{R}}\Theta(0)\rho(0) \ge \frac{\mathscr{G}M^{2}}{8\pi R^{4}}=\infty
\end{align}
This does not hold in general relativity since an infinite central pressure occurs for a critical finite mass and radius of $\mathscr{G}
M/R=4/9$. (See Section 10)).
\subsection{Hydrostatic equilibrium in higher dimensions}
These results extend naturally to n dimensions $\mathbf{R}^{n}$, which is now very briefly considered However, one can conclude that stars in equilibrium cannot exist for $n>3$.  Let $\mathbf{B}(0,R)\subset\mathbf{R}^{n}$ be an n-dimensional ball of radius R centered at the origin. The boundary is denoted $\partial\mathbf{B}(0,R)$. Let $\mathrm{A}(\partial\mathbf{B}(0,1))$ be the surface are of the unit n-dimensional ball and let $\mathrm{V}(\mathbf{B}(0,1))$ be the volume. Then the surface area of the n-ball of radius r is
\begin{align}
\mathrm{A}(\partial\mathbf{B}(0,r))=\mathrm{A}(\partial\mathbf{B}(0,1))r^{n-2}= \frac{2\pi^{n/2}}{\Gamma(n/2)}r^{n-2}
\end{align}++
and the volume is
\begin{align}
\mathrm{V}(\mathbf{B}(0,r))=\mathrm{V}(\mathbf{B}(0,1))r^{n}
=\frac{2\pi^{n/2}}{\Gamma(n/2+1)}r^{n}
\end{align}
where $\Gamma$is the standard gamma function. If $\mathbf{B}(0,R)$ supports a gaseous star of density $\rho(r)$ and pressure $p(r)$ then the FEHE is
\begin{align}
\frac{dp(r)}{dr}=-\frac{\mathscr{G}\mathcal{M}(r)\rho(r)}{r^{n-1}}
\end{align}
and the mass $\mathcal{M}(r)$ contained within a radius $r$ is
\begin{align}
\mathcal{M}(r)=\int_{0}^{r}|\mathrm{A}(\partial\mathbf{B}(0,1))|r^{n-1}dr=\int_{0}^{r}
\frac{2\pi^{n/2}}{\Gamma(n/2)}r^{n-2}dr
\end{align}
The n-dimensional generalisation is
\begin{align}
\frac{1}{r^{n-1}}\frac{d}{dr}\left(\frac{r^{n-1}}{\rho(r)}\frac{dp(r)}{dr}
\right)=-\mathrm{A}(\partial\mathbf{B}(0,1))\mathscr{G}\rho(r)
\end{align}
The Lagrangian form of (-) for $n>3$ is then
\begin{align}
\frac{dp(r)}{d\mathcal{M}(r)}=\frac{\mathscr{G}
\mathcal{M}(r)}{\mathrm{A}(\mathbf{B}(0,1))r^{2(n-1)}}
\end{align}
The central temperature and pressure lower bounds can be estimated as before. However, the central and pressure and temperature estimates are now drastically lowered when $n>3$ for a star of the same mass and radius.
\begin{thm}
For $n>3$, the central pressure and temperature lower bounds can be estimated as
\begin{align}
&p_{c}\ge \frac{\mathscr{G}M^{2}}{2\mathcal{A}(\mathbf{B}(0,1)R^{2(n-1)}}
=\frac{\mathscr{G}M^{2}\Gamma(n/2)}{4\pi^{n/2}R^{2(n-1)}}\\&
\Theta_{c}\ge \frac{\mathscr{R}}{k_{B}\mu}\frac{\Gamma(n/2)}
{\Gamma(\tfrac{n}{2}+1)}\frac{\mathscr{G}M}{R^{(n-2)}}
\end{align}
\begin{proof}
The central pressure is estimated as before via integrating (2.48) so that
\begin{align}
p(0)\ge \int_{0}^{M}\frac{\mathscr{G}\mathcal{M}(r)d\mathcal{M}(r)}
{\mathrm{A})\mathbf{B}(0,1)r^{2(n-1)}}
=\frac{\mathscr{G}M^{2}}{2\mathrm{A}(\mathbf{B}(0,1))R^{2(n-1)}}
\end{align}
To estimate the central temperature bound, apply the perfect gas law as before then
\begin{align}
\frac{k_{B}\mu}{\mathscr{R}}\Theta(0)\rho(0)=p(0)\ge \frac{\mathscr{G}M^{2}}{2\mathrm{A}(\mathbf{B}(0,1))R^{2(n-1)}}
=\frac{\mathscr{G}M^{2}\Gamma(n/2)}{4\pi^{n/2}R^{3n-4}}
\end{align}
Using the mean density for $\rho(0)$
\begin{align}
\rho(0)\sim \big\langle\rho\big\rangle=\frac{M}{\mathrm{V}(\mathbf{B}(0,R))}=
\frac{M\Gamma(\frac{n}{2}+1)}{2\pi^{n/2}R^{n}}
\end{align}
and substituting in (2.55), the result (2.52) then follows.
\end{proof}
\end{thm}
It can be observed that for $n>$ the central temperature and pressure in the star are suppressed for the same mass and radius. In the limit as $n\rightarrow\infty$ the temperature and pressure are zero. This suggests that for $n>3$ and especially $n\gg 3$, stars are unstable and cannot achieve equilibrium or ignition temperature for thermonuclear fusion of hydrogen.
\subsubsection{Some basic theorems concerning hydrostatic equilibrium}
Returning to $n=3$, a star in hydrostatic equilibrium also obeys a number of basic theorems and presented in Chandrasekhar in $\bm{[1]}$ with proofs given in $\bm{[1]}$.
\begin{thm}
If the average density $\langle\rho(r)\rangle\sim\mathcal{M}(r)/[(4\pi/3)r^{3}]$ is monotonically non-increasing as one moves outward in the star towards the surface, then in any equilibrium configuration, the function $\mathcal{F}(r)$ given by
\begin{align}
\mathcal{F}(r)=p(r)+\frac{3}{8\pi}\frac{|\mathcal{M}(r)|^{2}}{r^{4}}
\end{align}
is also non-increasing.
\end{thm}
\begin{proof}
\begin{align}
&\frac{dp(r)}{dr}=
\frac{d}{dr}[p(r)+\frac{3}{8\pi}\frac{|\mathcal{M}(r)|^{2}}{r^{4}}]\\&=
\underbrace{\frac{dp(r)}{dr}}_{use~Lagrange~form~of~FEHE}+
\frac{3}{4\pi}\mathscr{G}\frac{\mathcal{M}(r)\mathcal{M}^{\prime}(r)}
{r^{4}}-\frac{3}{2\pi}
\frac{|\mathcal{M}(r)|^{2}}{r^{5}}\nonumber\\&
=-\frac{\mathscr{G}\mathcal{M}^{\prime}(r)\mathcal{M}(r)}{4\pi r^{4}}+\frac{3}{4\pi}\frac{\mathscr{G}\mathcal{M}^{\prime}(r)\mathcal{M}(r)}{4\pi r^{4}}
-\frac{3}{2\pi}\frac{|\mathcal{M}(r)|^{2}}{r^{5}}\nonumber\\&
=\frac{1}{2\pi}\frac{\mathscr{G}\mathcal{M}(r)\mathcal{M}^{\prime}(r)}{r^{4}}-\frac{3}{2\pi}
\frac{|\mathcal{M}(r)|^{2}}{r^{5}}\nonumber\\&
=\frac{1}{2\pi}\frac{\mathscr{G}\mathcal{M}(r)}{r^{4}}\left(\mathcal{M}^{\prime}(r)-\frac{3\mathcal{M}(r)}{r}\right)\nonumber\\&
=\frac{1}{2\pi}\frac{\mathscr{G}\mathcal{M}(r)}{r}\frac{d}{dr}\left|\frac{\mathcal{M}(r)}{r^{3}}\right|
=\frac{2}{3}\mathscr{G}\frac{\mathcal{M}(r)|}{r}\frac{d}{dr}[\langle \rho(r)\rangle]\le 0
\end{align}
\end{proof}
An immediate corollary is a lower bound on the central pressure in the star
\begin{cor}
\begin{align}
p(0)=p_{c}\ge p(r)+\frac{3}{8\pi} \frac{\mathscr{G}|\mathcal{M}(r)|^{2}}{r^{4}} \ge \frac{3\mathscr{G}}{8\pi}\left|\frac{M^{2}}{R^{4}}\right|
\end{align}
which is consistent with (2.41).
\end{cor}
\begin{thm}
For any equilibrium configuration
\begin{align}
\mathbf{I}_{\gamma}
=\int_{o}^{R}\frac{{\mathscr{G}}\mathcal{M}(r)d\mathcal{M}(r)}{r^{\xi}}=
4\pi(4-\xi)\int_{0}^{R}p(r)r^{3-\xi}d\xi
\end{align}
if $\xi<4$.
\end{thm}
\begin{thm}
\begin{align}
\xi\pi p_{c}R^{4}r^{4-\xi}+\frac{|4-\xi|}{8}\frac{\mathscr{G}M^{2}}{R^{3}} >\mathbf{I}_{\xi}>\frac{\mathscr{G}M^{2}}{2R^{\xi}}
\end{align}
if $\xi<4 $.
\end{thm}
\subsection{Gravitational and Thermal Energies of a Star in Hydrostatic Equilibrium}
The following energies can be defined for a self-gravitating star. Because a star is a self-gravitating thermodynamic system this leads to some interesting properties; in particular, the fact that a star has negative specific heat: as it loses energy (via radiation) it can get hotter. Gravitational energy can also be converted to thermal energy; for example, as a protostar radiates and collapses in a (Kelvin-Helmholtz contraction) it gets hotter until thermonuclear reactions are initiated in its core. The star achieves hydrostatic equilibrium and becomes stable. The most commonly utilised energies in Newtonian astrophysical theory are the gravitational and thermal energies which are now defined.
\begin{defn}
The total gravitational potential energy (TGPE) or binding energy $\mathbf{E}_{G}$ of a star is the energy required to disperse the matter comprising the star to infinity, peeling it shell by shell from the outside in $\bm{[1,2,3]}$.
\begin{align}
\mathbf{E}_{G}=-4\pi\mathscr{G}\int_{0}^{R}r\mathcal{M}(r)
\rho(r)dr\equiv\int_{0}^{M}\frac{\mathscr{G}\mathcal{M}(r)
d\mathcal{M}(r)}{r}
\end{align}
The TGPE can also be written as
\begin{align}
\mathbf{E}_{G}=-3\int_{0}^{R}4\pi r^{2} p(r)dr
\end{align}
This follows by using the FEHE and integrating by parts so that
\begin{align}
\mathrm{E}_{G}=4\pi\int_{0}^{R}\frac{dp(r)}{dr}r^{3}dr=[r^{3}p(r)]_{0}^{R}
-3\int_{0}^{R}4\pi r^{2}p(r)r^{2}dr
\end{align}
and the first term on the RHS vanishes since $p(R)=0$.
\end{defn}
\begin{defn}
The thermal energy of the star is defined as $\bm{[2,3]}$.
\begin{align}
\mathbf{E}_{T}=\int_{0}^{R}\mathcal{E}(r) 4\pi r^{2}dr=\int_{0}^{R}\mathcal{E}(r)d\mathrm{V}(r)
\end{align}
where $\mathcal{E}(r)$ is the density of the internal thermal energy only, not including the gravitational energy or self gravitation and the rest-mass energy. For example, a perfect/ideal gas of N particles at temperature $\Theta$ contained in a domain $\mathbf{B}\subset\mathbf{R}^{3}$ of volume $\mathrm{V}(\mathbf{B})$, the homogenous energy density is $\mathcal{E}=\tfrac{3}{2}Nk_{B}\Theta/\mathrm{V}=E_{\Theta}/\mathrm{V}(\mathbf{B})$ so that $\mathbf{E}_{\Theta}=\mathrm{V}(\mathbf{B})\mathcal{E}=\tfrac{4}{3}\pi R^{3}\mathcal{E}$. For a gas in spherical 'box' the thermal energy is $\mathbf{E}_{\Theta}=\int_{\mathbf{B}}\mathcal{E}(r)d\mathrm{V}=\int_{\mathbf{B}}\mathcal{E}(r)4\pi r^{2}dr$. The thermal energy can also be expressed as integrals over $d\mathcal{M}(r)$ so that
\begin{align}
\mathbf{E}_{T}=\int_{0}^{R}\mathcal{E}(r)4\pi r^{2}dr=\int_{0}^{R}\frac{\mathcal{E}(r)}{\rho(r)}d\mathcal{M}(r)
=\int_{0}^{M}\mathscr{J}(r,\rho(r))d\mathcal{M}(r)
\end{align}
where $\mathscr{J}$ is the thermal energy density per unit mass.
\end{defn}
\begin{defn}
The total energy $E$ of the star is then the sum of the thermal and gravitational energies
\begin{align}
 \mathbf{E}=\mathbf{E}_{T}+\mathbf{E}_{G}=\int_{0}^{R}[\mathcal{E}(r)-2p(r)]4 \pi ^{2}dr\equiv \int_{0}^{R}[\mathcal{E}(r)-2p(r)]d\mathrm{V}(r)
\end{align}
The star must have negative energy and is stable against thermal dispersion/expansion of its matter if $\mathcal{E}(r)<3p(r)$.
\end{defn}
\begin{rem}
The following basic (heuristic) observations can be made regarding the total energy.
\begin{itemize}
\item If $\mathbf{E}<0$ with $\mathbf{E}(r)<3p(r)$ then the star is stable to dispersion/explosion.
\item If $\mathbf{E}>0$ then the star will disperse/explode or collapse.
\item A star with $\mathbf{E}=0$ is at the point of dynamical instability.
\end{itemize}
In Section (4) it will be shown that the vanishing of the 1st variation of E (namely $\delta\mathbf{E}=0$) leads to the FEHE.
\end{rem}
\subsection{The stellar virial theorem}
The virial theorem, well known in classical and celestial mechanics is another consequence of the FEHE and the gravitational potential/binding energy. The virial theorem for self-gravitating gaseous stars follows quite simply from the Lagrangian form of the FEHE and the gravitational potential energy $\mathbf{E}_{G}$.
\begin{thm}
The virial theorem for a self-gravitating gaseous star is
\begin{align}
3\int_{0}^{\mathrm{V}}p(r)d\mathrm{V}(r)
=\int_{0}^{M}\frac{\mathscr{G}\mathcal{M}(r)d\mathcal{M}(r)}{r}
\end{align}
 or
\begin{align}
3\int_{0}^{\mathrm{V}}p(r)d\mathrm{V}(r)
+{\mathbf{E}}_{\scriptstyle{G}}=0
\end{align}
\end{thm}
\begin{proof}
Using the Lagrangian form of the FEHE
\begin{align}
\frac{dp(r)}{d\mathcal{M}(r)}=-\frac{\mathscr{G}\mathcal{M}(r)}{4\pi r^{4}}
\end{align}
multiply through by $4\pi r^{2}$ so that
\begin{align}
4\pi r^{3}\frac{dp(r)}{d\mathcal{M}(r)}=-\frac{\mathscr{G}\mathcal{M}(r)}{r}
\end{align}
and
\begin{align}
4\pi r^{3}dp(r)=-\frac{\mathscr{G}\mathcal{M}(r)}{r}d\mathcal{M}(r)
\end{align}
Integrating over the volume of the star with $\mathrm{V}(r)=\tfrac{4}{3} r^{3}$ and
volume element $d\mathrm{V}=4\pi r^{2}dr$
\begin{align}
\underbrace{3\int_{p(0)}^{p(R)}\mathrm{V}(r)dp(r)}_{int~by~parts}
=-\int_{0}^{M}\frac{\mathscr{G}\mathcal{M}(r)}{r}d\mathcal{M}(r)
\end{align}
\begin{align}
|p(r)\mathrm{V}(r)|_{p(0)}^{p(R)}-3\int_{\mathrm{V}(0)}^{\mathrm{V}(R)} p(r)d\mathrm{V}(r)=-
\int_{0}^{M}\frac{\mathscr{G}\mathcal{M}(r)d\mathcal{M}(r)}{r}
\end{align}
and $\mathrm{V}(0)=0$ and $p(R)=0$. Hence (2.69) or (2.70) follows.
\end{proof}
The virial theorem can also be expressed in terms of the thermal energy $\mathbf{E}_{T}$ and gravitational energy $\mathbf{E}_{G}$. An immediate corollary is that the specific heat of a gaseous star at fixed volume is actually negative.
\begin{thm}
The virial theorem can be expressed in terms of the thermal and gravitational energies as
\begin{align}
2{\mathbf{E}}_{T}=-{\mathbf{E}}_{G}
\end{align}
so that the total energy is minus the thermal energy
\begin{align}
\mathbf{E}=\mathbf{E}_{T}+\mathbf{E}_{G}=-\mathbf{E}_{T}
\end{align}
The specific heat is then negative so $\mathrm{C}_{V}<0$. For a perfect gas
\begin{align}
\mathrm{C}_{V}=-\tfrac{3}{2}Nk_{b}<0
\end{align}
\end{thm}
\begin{proof}
Since the pressure $p(r)$ and the thermal energy density $\mathcal{E}(r)$ are related by
$3p(r)=2\mathcal{E}$, the virial theorem becomes
\begin{align}
2\int_{0}^{R}{\mathcal{E}}(r)4\pi r^{2}dr=2{\mathbf{E}_{T}}=-{\mathbf{E}_{G}}
\end{align}
The total energy $\mathbf{E}$ is then
\begin{align}
{\mathbf{E}}={\mathbf{E}_{T}}+{\mathbf{E}_{G}}=-\mathbf{E}_{T}
\end{align}
The specific heat in hydrostatic equilibrium at fixed $\mathrm{V}$ is then the derivative of $\mathbf{E}$ with respect to temperature $\Theta$ so that
\begin{align}
{\mathrm{C}}_{\mathrm{V}}=\frac{d\mathbf{E}}{d\Theta}=-\frac{d\mathbf{E}_{T}}{d\Theta}<0
\end{align}
For an ideal gas at temperature $\Theta$ one has $\mathbf{E}_{T}=\frac{3}{2}Nk_{b}\Theta$ so that$\mathrm{C}_{\mathrm{V}}=-\frac{3}{2}Nk_{b}<0$.
\end{proof}
Physically, this means that as a self-gravitating system loses energy by radiation it gravitationally contracts since it cannot be supported in hydrostatic equilibrium. Therefore it heats up. During the lifetime of the star, the radiation pressure is maintained by thermonuclear reactions in the core. But once this reservoir of fuel becomes exhausted the star will shift out of hydrostatic equilibrium and contract, and gravitational energy will be converted to thermal energy. In a main sequence star like the Sun, the increased core temperature initiates the next phase of thermonuclear burning whereby helium is fused into carbon and oxygen. The temperature of ordinary matter with positive specific heat, decreases as the matter radiates energy; for example, a red-hot piece of iron will cool off. The statistical mechanics and thermodynamics of self gravitating systems of particles has been discussed in various works(refs).
\subsection{The isothermal sphere for a perfect gas}
The equilibrium of an isothermal gas sphere (IGS) is now considered for a perfect gas. This is a self-gravitating gaseous sphere of uniform temperature $\bm{[1,6]}$. Real stars have thermal gradients so this is not a realistic model of any normal star. However, once a star has exhausted its nuclear fuel, the central core of the star--consisting primarily of helium for main sequence stars--will become isothermal, as a prelude to becoming a red giant. It can also describe globular clusters. Thus, the isothermal core is still astro-physically relevant. Also, the isothermal sphere has a number of interesting thermodynamic and statistical mechanical properties. $\bm{[1,6,7,8]}$.
\begin{lem}
Let $\mathbf{B}(0,R)\subset\mathbf{R}^{3}$ be a ball/domain of radius $R$ centred at the origin and supporting a perfect gas at uniform temperature $\Theta$ consisting of N particles of mass $m$. The total mass is then $M=mN$. The pressure and density at any radius $r<R$ are then related by $p(r)=\mathscr{R}\rho(r)=\frac{k_{B}\Theta}{m}\rho(r)$. The equilibrium equations for the isothermal gas sphere are then the Poisson and Poisson-Boltzmann equations
\begin{align}
&\Delta(\log\rho(r))=\frac{1}{r^{2}}\frac{d}{dr}\left(r^{2}\frac{d}{dr}(\log\rho(r))\right)=-\frac{4\pi \mathscr{G}\rho(r)}{\mathcal{R}}\\&
\Delta\Phi(r)=\frac{1}{r^{2}}\frac{d}{dr}\left(r^{2}\frac{d\Phi(r)}{dr}
\right)=4\pi \mathscr{G}\rho_{c}
\exp\left(-\frac{\Phi(r)}{\mathscr{R}}\right)
\end{align}
\begin{proof}
Differentiating the perfect gas law
\begin{align}
\frac{dp(r)}{dr}=\mathscr{R}\frac{d\rho(r)}{dr}
\end{align}
and applying the FEHE $dp(r)/dr=-\frac{\mathscr{G}\mathcal{M}(r)}{r^{2}}$ gives
\begin{align}
\frac{dp(r)}{dr}=\mathscr{R}\frac{d\rho(r)}{dr}
=-\rho(r)\frac{d\Phi(r)}{dr}
=-\frac{\mathscr{G}\mathcal{M}(r)\rho(r)}{r^{2}}
\end{align}
Integrating
\begin{align}
\mathscr{R}\int_{\rho(0)}^{\rho(r)}\frac{d\overline{\rho(r)}}{\overline{\rho(r)}}
=-\int_{\Phi(0)}^{\Phi(r)}d\overline{\Phi(r)}=-\Phi(r)
\end{align}
since $\Phi(0)=0$. Then
\begin{align}
p(r)=\log\left|\frac{\rho(r)}{\rho(0)}|\right|=
-\frac{\Phi(r)}{\mathscr{R}}
\end{align}
so that
\begin{align}
\rho(r)=\rho_{c}\exp\left(-\frac{\Phi(r)}{\mathscr{R}}\right)
\end{align}
To derive (2.82), begin with the Emden equation
\begin{align}
\frac{1}{r^{2}}\frac{d}{dr}\left(\frac{r^{2}}{\rho(r)}\frac{dp(r)}{dr}
\right)=-4\pi\mathscr{G}\rho(r)
\end{align}
Then from the FEHE and (2.87)
\begin{align}
&\frac{1}{\rho(r)}\frac{dp(r)}{dr}=\frac{-\mathscr{G}\mathcal{M}(r)}{r^{2}}=-\frac{d\Phi(r)}{dr}=
\mathscr{R}\frac{d}{dr}\log\left|\frac{\rho(r)}{\rho(0)}\right|\\&
=\mathscr{R}\frac{d}{dr}|\log\rho(r)-\log\rho(0)|=\mathscr{R}\frac{d}{dr}\log\rho(r)
\end{align}
so that the Emden equation becomes
\begin{align}
\frac{1}{r^{2}}\frac{d}{dr}\left(r^{2}\frac{d}{dr}(\log\rho(r))\right)
=-\frac{4\pi \mathscr{G}\rho(r)}{\mathcal{R}}
\end{align}
\end{proof}
\end{lem}
\begin{cor}
The IGE has a singular solution
\begin{align}
\Phi(r)=\frac{1}{2\pi\mathscr{R}\mathscr{G} r^{2}}
\end{align}
Since $\rho(r)\sim r^{-2}$ the total mass is $M=\int_{0}^{\infty}4\pi r^{2}dr=\int_{0}^{\infty}4\pi dr=\infty$
\end{cor}
\subsubsection{Relevance of the isothermal gas sphere to red giant evolution}
Once the core of a main-sequence star has exhausted its supply of hydrogen fuel, it gradually ceases to produce nuclear energy
and, in the limit of thermal equilibrium, becomes isothermal. To see this, consider equations (2.24) and (2.25). Once thermonuclear reactions begin to cease within the core region then the thermonuclear energy production rate per mass diminishes so that $\mathfrak{N}\rightarrow 0$ and
\begin{align}
\frac{d\mathcal{L}(r)}{dr}=4\pi r^{2}\rho(r)\mathfrak{T}=0
\end{align}
Equation (-) for the temperature gradient then becomes
\begin{align}
\frac{d\Theta(r)}{dr}\sim C\frac{1}{|\Theta(r)|^{3}}\frac{\mathcal{L}(r)}{4\pi r^{2}}=0
\end{align}
so that $\Theta\rightarrow const.$ and the core region becomes isothermal; essentially an isothermal gas sphere of $He^{4}$. Schonberg and  Chandrasekhar (1942) showed that, if the envelope is polytropic with index $n = 3$, then there is a maximum fractional mass that the core can achieve. If the core is less massive, it can remain isothermal while nuclear reactions continue in a surrounding shell, and the isothermal core temperature is set by the temperature of this shell.
If the SC mass is exceeded then the core contracts until it is supported by electron degeneracy pressure or helium begins to burn/fuse at the centre (the triple-alpha reaction) to create carbon and oxygen. The idealised result is sufficiently accurate that it has become a well-established element of the theory of the evolution of stars once they move off the main sequence. It is referred to simply as the Schonberg–Chandrasekhar (SC) limit $\bm{[9]}$.
\begin{defn}
The Schonberg-Chandrasekhar limit is the maximum fraction of a star's mass that can exist as an isothermal core and
still support the outer mass envelope in a state of hydrostatic equilibrium. It is
\begin{align}
\frac{M_{CORE}}{M}\sim 0.37|\frac{\mu_{ENV}}{\mu_{CORE}}
\end{align}
where $\mu_{ENV}$ and $\mu_{CORE}$ are the molecular weights in the envelope and core regions
respectively.
\end{defn}
Similar limits have been computed for other polytropic solutions. Beech (1988) calculated the corresponding limit for
an isothermal core surrounded by an envelope with n = 1. Eggleton, Faulkner and Cannon (1998) discovered that, when n = 1 in
the envelope and n = 5 in the core, a fractional mass limit exists if the density decreases discontinuously at the core-envelope
boundary by a factor exceeding 3. They went further to propose conditions on the polytropic indices of the core and envelope that
lead to fractional mass limits. We refer to all these limits, including the original result of Schonberg and Chandrasekhar (1942) as S-C-like
limits.
\section{Deviations From Hydrostatic Equilibrium and Stellar Timescales}
So far, only stars in hydrostatic equilibrium have been considered. Suppose now the
star deviates from hydrostastic equilibrium via radial perturbations or via a decrease/increase in the interior pressure. If gravitation dominates the star will collapse; if internal/thermal forces dominate the star will dissipate or explode. This deviation from equilibrium occurs only towards the end of the star's life when it has depleted its nuclear fuel. A more formal treatment of perturbations away from equilibrium is given in Section (4). Here, a heuristic result can  be derived from the FEHE $\bm{[3]}$.
\begin{align}
\frac{\mathscr{G}\mathcal{M}(r)\rho(r)}{r^{2}}+\frac{1}{\rho(r)}\frac{dp(r)}{dr}=0
\end{align}
If the star deviates slightly from hydrostatic equilibrium then
\begin{align}
\frac{\mathscr{G}\mathcal{M}(r)\rho(r)}{r^{2}}+\frac{1}{\rho(r)}\frac{dp(r)}{dr}
~\mathlarger{\gtrless}~~0
\end{align}
This induces an acceleration of the mass $\mathcal{M}(r)$ either inward or outward
so that
\begin{align}
\frac{d^{2}\lambda(t)}{dt^{2}}
=\frac{\mathscr{G}\mathcal{M}(\lambda(t))\rho(r)}{\lambda(t)^{2}}+\frac{1}{\rho(r)}\frac{dp(r)}{d\lambda(t)}
~\mathlarger{\gtrless}~~0
\end{align}
where $\lambda(t)$ is now a scale factor or time-dependent radius and $\mathcal{M}(\lambda(t))$ is the mass contained within a (moving) radius $\lambda(t)$ at time $t$. This can be written as
\begin{align}
\frac{d^{2}\lambda(t)}{dt^{2}}=\frac{r\mathcal{M}(\lambda(t))\rho(r)}{|\lambda(t)|^{2}}
\left(1+\frac{\frac{1}{\rho(r)}\frac{dp(r)}{d\lambda(t)}
}{\frac{\mathscr{G}\mathcal{M}(\lambda(t))\rho(r)}{\lambda(t)^{2}}}\right)
~\mathlarger{\gtrless}~~0\
\end{align}
\subsection{Zero-pressure gravitational collapse and the hydrodynamical time}
Suppose now the pressure is 'switched off' entirely so that $p(r)=0$ and the pressure
gradient vanishes so that $dp(r)/dr=0$. Then (3.4) reduces to
\begin{align}
\frac{d^{2}\lambda(t)}{dt^{2}}=\frac{\mathscr{G}\mathcal{M}(r)\rho(r)}{|\lambda(t)|^{2}}<0,~~\lambda(t)\le R
\end{align}
Introducing a velocity $\mathrm{U}(r(t))=d\lambda(t)/dt$
\begin{align}
\frac{d\mathrm{U}(t)}{dt}=\frac{\mathscr{G}\mathcal{M}(r)\rho(r)}{\lambda(t)^{2}}<0
\end{align}
Pressureless matter or matter with $p\ll 1$ is often referred to as 'dust' in the general relativity literature. Within the Newtonian theory, the time for the sphere of pressureless dust to collapse to zero, is the free fall time or the hydrodynamic time. With $p(r)=0$ within the star, it undergoes gravitational collapse to zero size or infinite density within a time $T_{H}$, the hydrodynamical time.
\begin{lem}
Let $\mathbf{B}\subset\mathbf{R}^{3}$ be a domain supporting a gaseous star of radius R and mass M at time t=0. The initial density is then $\rho(0)=3M/4\pi R^3$ and $\mathrm{V}(\mathbf{B}(0,R))=\tfrac{4}{3}\pi R^{3}$. For $t>0$ the star is no longer in hydrostatic equilibrium so that the support $\mathbf{B}(0,\lambda(t))$ evolves (shrinks) for increasing t with $\mathbf{B}(0,\lambda(t))\subset\mathbf{B}(0,R)$ and $\mathrm{V}(0,\lambda(t))=\tfrac{4}{3}\pi |\lambda(t)|^{3}$. Since the star is collapsing for $t>0$ then at a sequence of times $t_{1}<t_{2}<...<t_{n-1},t_{n}$.
\begin{align}
\mathbf{B}(0,\lambda(t_{n})
\subset\mathbf{B}(0,\lambda(t_{n-2})\subset...\subset\mathbf{B}(0,\lambda(t_{2}))\subset
\mathbf{B}(0,\lambda(t_{1}))\nonumber
\end{align}
Then the star collapses to zero size at the hydrodynamical time $T_{H}$ which is
\begin{align}
T_{H}=\left(\frac{3\pi}{32\mathscr{G}\rho(0)}\right)^{1/2}
\end{align}
Then  $\mathrm{V}(\mathbf{B}(\lambda(T_{H}))=0$ and $\rho(0,\lambda(t_{H}))=3M/4\pi|\lambda(t)|^{3}=\infty$
\end{lem}
\begin{proof}
The ODE (3.6) is expressible as
\begin{align}
&\frac{d^{2}\lambda(t)}{dt^{2}}=
\frac{d}{dt}\left(\frac{d\lambda(t)}{dt}\right)
=\frac{d}{dt}\mathrm{U}(\lambda(t))=\frac{d\lambda(t)}{dt}
\frac{d}{d\lambda(t)}\mathrm{U}(\lambda(t))\nonumber\\&=\mathrm{U}(\lambda(t))
\frac{d}{d\lambda(t)}\mathrm{U}(\lambda(t))
=\frac{1}{2}\frac{d}{d\lambda(t)}|\mathrm{U}(\lambda(t))|^{2}
=\frac{\mathscr{G}\mathcal{M}(r)}{\lambda(t)^{2}}=-
\frac{\mathscr{G}M^{2}}{|\lambda(t)|^{2}}
\end{align}
Using the boundary conditions $\mathcal{M}(0)=0,\mathcal{M}(R)=M,
\rho(R)=0,\mathrm{U}(\lambda(0))=0$ and integrating gives
\begin{align}
\int_{0}^{\mathrm{U}(\lambda(t))}
d|\mathrm{U}(\lambda(t))|^{2}
=-2\int_{R}^{r}\frac{\mathscr{G}M}{\overline{r}^{2}}
d\overline{r}=2\mathscr{G}M\left(\frac{1}{\lambda(t)}-\frac{1}{R}\right)
\end{align}
where the radial integral is inward since the sphere is collapsing. Then
\begin{align}
\mathrm{U}(\lambda(t))=\frac{d\lambda(t)}{dt}=\pm 2\mathscr{G}M\left(\frac{1}{\lambda(t)}-\frac{1}{R}\right)^{1/2}
\end{align}
Taking the negative sign and integrating
\begin{align}
\int_{R}^{\lambda(t)}\frac{d\overline{\lambda(t)}}
{\left(\frac{1}{\overline{\lambda(t)}}-\frac{1}{R}\right)^{1/2}}=2\mathscr{G}M\int_{0}^{t}dt
=-2\mathscr{G}M t
\end{align}
so that`
\begin{align}
\left|R\left(\frac{1}{\overline{\lambda(t)}}-\frac{1}{R}\right)^{1/2}r
+R^{3/2}\tan^{-1}\left(\sqrt{R}\left(\frac{1}{\overline{\lambda(t)}}-\frac{1}{R}
\right)\right)\right|_{R}^{\lambda(t)}=\left(2\mathscr{G}M\right)^{1/2}t
\end{align}
The lhs vanishes for $\lambda(t)=\lambda(0)=R$ so that
\begin{align}
R\left(\frac{1}{\overline{\lambda(t)}}-\frac{1}{R}\right)^{1/2}r
+R^{3/2}\tan^{-1}\left(\sqrt{R}\left(\frac{1}{\overline{\lambda(t)}}-\frac{1}{R}
\right)\right)=\left(2\mathscr{G}M\right)^{1/2}t
\end{align}
Since $\tan^{-1}(0)=0,\tan^{-1}(\infty)=\frac{\pi}{2}$ then at $r=\lambda(t)=\lambda(T_{H})=0$
\begin{align}
R^{3/2}\frac{\pi}{2}=(2\mathscr{G}M)^{1/2}T_{H}
\end{align}
and squaring
\begin{align}
R^{3}\frac{\pi}{2}=(2\mathscr{G}M)T_{H}
\end{align}
Since the initial density is $\rho(0)=3M/4\pi R^{3}$ then
\begin{align}
\frac{8\pi \mathscr{G}\rho_{o}T_{H}^{2}}{3}=\frac{\pi^{2}}{4}
\end{align}
Solving for $T_{H}$ then gives (3.7) as required.
\end{proof}
The exact same result also follows from the Tolman-Oppenheimer-Syner models which apply an Einstein-matter system to pressureless gravitational collapse.(See $\bm{[10]}$. The star collapses to zero size within a finite comoving proper time and essentially forms a black hole. However, from the Newtonian theory one can still deduce that a sufficiently heavy star with negligible pressure could collapse out of existence within a finite time, once it is no longer in HE. However, main-sequence Newtonian stars like the Sun and stars of comparable mass, will collapse to form white dwarfs since the pressure is always finite and never exactly zero.
\subsection{The Kelvin-Helmholtz time scale}
A self-gravitating gaseous body can convert gravitational energy into internal thermal energy by a process of contraction $\bm{[2,3]}$. From (2.39) it was seen that the central pressure and temperature increase as the radius decreases. This is also consistent with the fact that self-gravitating gases have negative specific heats so the star will heat up internally as it radiates away energy and contracts. The total energy of the star is the sum of the thermal and gravitational energies so that $\mathbf{E}=\mathbf{E}_{T}+\mathbf{E}_{G}<0$. The Kelvin-Helmhotz time $T_{KH}$ for a self-gravitating gaseous body of mass M and initial radius R, to convert its gravitational potential energy to thermal energy, at a constant luminosity $\mathcal{L}$. Suppose a star is powered only by the process of conversion of gravitational to thermal energy then $T_{KH}$ is given by
\begin{align}
T_{KH}=-\frac{\mathbf{E}_{G}}{\mathcal{L}}=\frac{1}{\mathcal{L}}
\int_{\mathcal{M}(0)}
^{\mathcal{M}(R)}\frac{\mathscr{G}\mathcal{M}(r)d\mathcal{M}(r)}{r}
=\frac{\mathscr{G}M^{2}}{2R\mathcal{L}}
\end{align}
For the Sun, $\mathcal{L}\sim 3.9\times 10^{26}W$ so that $T_{KH}\sim 30\times 10^{6}$ years.
\newline
\underline{Historical~remark}\newline
For the Sun, $T_{KH}\sim 30\times 10^{6}$ years, which is much less than the age of the Sun or the Earth. Also, the Sun would have greatly contracted in size over this timescale if this was its only energy source. In the 19th century this theory was proposed by Lord Kelvin (William Thompson) and Helmholtz (and also Sitter) as an explanation as to how the Sun and stars could produce their energy. At this time, enough was known about gravitation, gas dynamics, hydrodynamics and thermodynamics that physically viable models of the Sun and stars could be formulated. But even at the time of Kelvin, there was geological and fossil evidence to support the age of the Earth being far greater than the estimate $T_{KH}$ of Kelvin. Two years earlier, Darwin had also published his work \emph{Origin of Species} which suggested that life on Earth took billions of years of evolve, so the Sun had to be at least as old. Kelvin also abandoned his theory because he learned that Alexander the Great had observed a total solar eclipse when he crossed the River Oxus in 329 BC. This put an upper limit on the size of the Sun at that date, suggesting that the Sun was not contracting, or at least not contracting fast enough, to provide the required energy output. It would be many decades until the discovery of radioactivity and also quantum mechanics which would (in the 20th century) lead to a full understanding of thermonuclear reactions and processes; the actual source of the immense heat of the stars that maintains them in perpetual equilibrium with constant luminosity. However, Kelvin-Helmholtz contractions still describes collapsing protostars which form in gas clouds and which radiate thermally by this process in the phase before nuclear reactions are initiated in their core.
\section{Derivation Of The Hydrostatic Equilibrium Equation Via Variational Methods}
The hydrostatic equilibrium criteria are now examined from a more formal perspective via variational methods. The HEE for the isothermal gaseous sphere
can also be derived via a variation method from a statistical mechanical distribution function. It is seen in both cases that the equilibria correspond to the critical points. First, the definitions of 1st and 2nd variations of a function are briefly reviewed.
\begin{defn}
Let $\Theta=\Theta(f(r))$ be a functional of a function $f(r)$ with
$f:\mathbf{R}\rightarrow\mathbf{R}$ and $\Theta:\mathbf{R}\rightarrow\mathbf{R}$.
For spherical symmetry $r\le R$. The 1st variation with respect to $\delta f(r)$ is then
\begin{align}
\delta\Lambda(f(r))=\frac{\partial\Lambda(f(r))}{\partial f(r)}\delta f(r)=
\frac{\partial\Lambda(f(r))}{\partial f(r)}\frac{df(r)}{dr}\delta r
\end{align}
The second variation is
\begin{align}
\delta^{2}\Lambda(f(r))=\frac{\partial \Lambda(f(r))}{\partial f(r)}\delta f(r)+
\frac{1}{2}\frac{\partial^{2}\Lambda(f(r))}{|\partial f(r)|^{2}}|\delta f(r)|^{2}
\end{align}
For a functional $\Lambda=\Lambda(f(r),h(r))$, the 1st variation is
\begin{align}
\delta\Lambda(f(r))
=\frac{\partial\Lambda(f(r),h(r))}{\partial f(r)}\delta f(r)+\frac{\partial \Lambda(f(r),h(r))}{\partial h(r)}\delta h(r)
\end{align}
and the second variation is
\begin{align}
&\delta\Lambda(f(r))
=\frac{\partial\Lambda(f(r),h(r))}{\partial f(r)}\delta f(r)+\frac{\partial\Lambda(f(r),h(r))}{\partial h(r)}\delta h(r)\nonumber\\&
+\frac{1}{2}\frac{\partial^{2}\Lambda(f(r),h(r))}{|\partial f(r)|^{2}}|\delta f(r)|^{2}
+\frac{1}{2}\frac{\partial^{2}\Lambda(f(r),h(r))}{|\partial h(r)|^{2}}|\delta h(r)|^{2}
\end{align}
\end{defn}
The first theorem derives the FEHE via the 1st variation of the total energy $\bm{[11]}$.
\begin{thm}
Let a Newtonian gaseous star of mass $M$ and radius $R$ in hydrostatic equilibrium have total energy $\mathbf{E}=\mathbb{E}_{T}+\mathbb{E}_{G}$, occupying a ball $\mathbf{B}_{R}(0)\subset\mathbf{R}^{3}$. Let $\mathscr{S}$ be the entropy per nucleon or mass element and $(\mathcal{M}(r),p(r),\rho(r))$ have the usual definitions. Then the vanishing of the 1st variation of the total energy is
\begin{align}
\delta\mathbf{E}=\frac{\partial \mathbf{E}}{\partial\rho(r)}\delta\rho(r)+
\frac{\partial \mathbf{E}}{\partial r}\delta r=0
\end{align}
yields the fundamental hydrostatic equilibrium equation. The equilibrium state is therefore a minimizer or extremum of $\mathscr{H}$.
\end{thm}
\begin{proof}
The gravitational and thermal energies of the star are given by
\begin{align}
&\mathbf{E}_{G}=-4\pi \mathscr{G}\int_{0}^{R}\frac{\mathcal{M}(r)\rho(r)}{r}=
-\int_{0}^{M}\frac{\mathcal{M}(r)d{\mathcal{M}(r)}}{r}\\&
\mathbf{E}_{T}=\int_{0}^{R}\mathcal{E}(r,\rho(r))4\pi r^{2}dr
=\int_{0}^{M}\mathscr{J}(r,\rho(r))d\mathcal{M}(r)
\end{align}
where $\mathcal{E}$ is the internal thermal energy density and $\mathscr{J}$ is the energy per unit mass. The variation of the total energy is then
\begin{align}
\delta\mathbf{E}=\frac{\partial \mathbf{E}}{\partial\rho(r)}\delta\rho(r)+
\frac{\partial \mathbf{E}}{\partial r}\delta r=\frac{\partial
\mathbf{E}_{T}}{\partial\rho(r)}\delta\rho(r)+\frac{\partial \mathbf{E}_{G}}{\partial r}\
\delta r=0
\end{align}
\begin{align}
&\delta\mathbf{E}=\frac{\partial}{\partial\rho(r)}\int_{0}^{M}
\frac{\mathcal{E}(r,\rho(r))}{\rho(r)}d\mathcal{M}(r)\delta\rho(r)-\frac{\partial}{\partial r}\int_{0}^{M}\frac{\mathcal{M}(r)d\mathcal{M}(r)}{r}\delta r\nonumber\\&
=\int_{0}^{M}\frac{\partial}{\partial\rho(r)}{\mathscr{J}(r,\rho(r))}d\mathcal{M}(r)\delta\rho(r)
-\int_{0}^{M}\frac{\partial}{\partial r}\left|\frac{\mathcal{M}(r)d\mathcal{M}(r)}{r}\right|\delta r
\end{align}
Using the thermodynamic relationship
\begin{align}
p(r)=\frac{\partial\mathcal{E}}{\partial ( 1/\rho(r))}
=|\rho(r)|^{2}\frac{\partial\mathscr{J}}{\partial\rho(r)}
\end{align}
for the pressure at constant entropy as well as $\delta(\tfrac{1}{\rho(r)})=
\delta(4 \pi r^{2}\tfrac{dr}{d \mathcal{M}(r)})$ and the variational calculus result $\delta(\tfrac{dr}{df(r)})g(r)=\tfrac{dg(r)}{df(r)}\delta r$,
for functions $f(r)$ and $g(r)$, the 1st variation $\delta\mathbf{E}$ becomes
\begin{align}
\delta \mathbf{E}&=\int_{0}^{M}\frac{p(r)}{|\rho(r)|^{2}}\delta \rho(r)d\mathcal{M}(r)-
{\mathscr{G}}\int_{0}^{M}\frac{\partial}{\partial r}
\frac{\mathcal{M}(r)d\mathcal{M}(r)}{r}\delta r\nonumber\\&
=\int_{0}^{M}\frac{p(r)}{|\rho(r)|^{2}}\delta \rho(r)d\mathcal{M}(r)+
{\mathscr{G}}\int_{0}^{M}\frac{\mathcal{M}(r)}{r^{2}}d\mathcal{M}(r)\delta r\nonumber\\&
=-\int_{0}^{M}p(r)\delta(\frac{1}{\rho(r)})d\mathcal{M}(r)+
{\mathscr{G}}\int_{0}^{M}\frac{\mathcal{M}(r)}{r^{2}}d\mathcal{M}(r)\delta r\nonumber\\&
=-\int_{0}^{M}p(r)\delta\left|4\pi r^{2} \frac{dr}{d\mathcal{M}(r)}\right|
d\mathcal{M}(r)+{\mathscr{G}}\int_{0}^{M}\frac{\mathcal{M}(r)}{r^{2}}d\mathcal{M}(r)\delta r\nonumber\\&
=-\int_{0}^{M}p(r)\delta(4\pi r^{2})\frac{dr}{d\mathcal{M}(r)}d\mathcal{M}(r)
-\int_{0}^{M}p(r)4\pi r^{2}\delta(\frac{dr}{d\mathcal{M}(r)})
d\mathcal{M}(r)+{\mathscr{G}}\int_{0}^{M}\frac{\mathcal{M}(r)}{r^{2}}d\mathcal{M}(r)\delta r\nonumber\\&
=-\int_{0}^{M}8\pi P(r) r\frac{dr}{d{\mathcal{M}}(r)}d\mathcal{M}(r)\delta r
-\int_{0}^{M}4\pi\delta(\frac{dr}{d\mathcal{M}(r)})p(r) r^{2}
d\mathcal{M}(r)+{\mathscr{G}}\int_{0}^{M}\frac{\mathcal{M}(r)}{r^{2}}d\mathcal{M}(r)\delta r\nonumber\\&
=-\int_{0}^{M}8\pi p(r) r\frac{dr}{d\mathcal{M}(r)}\mathcal{M}(r)\delta r
-\int_{0}^{M}4\pi\frac{d}{d\mathcal{M}(r)}(p(r) r^{2})\delta r
d\mathcal{M}(r)+\mathscr{G}\int_{0}^{M}\frac{\mathcal{M}(r)}{r^{2}}d\mathcal{M}(r)\delta r
\nonumber\\&=-\int_{0}^{M}8\pi p(r)r\frac{dr}{d\mathcal{M}(r)}d\mathcal{M}(r)\delta r
-\int_{0}^{M}8\pi p(r) r\frac{dr}{d\mathcal{M}(r)}d\mathcal{M}(r)\delta r\\&
+\int_{0}^{M}4\pi r ^{2}\frac{dp(r)}{d\mathcal{M}(r)}(p(r)\delta r
d\mathcal{M}(r)+\mathscr{G}\int_{0}^{M}\frac{\mathcal{M}(r)}{r^{2}}d\mathcal{M}(r)\delta r
\nonumber\\&=\int_{0}^{M}\frac{1}{\rho(r)}\frac{dp(r)}{dr}+
\int_{0}^{M}\frac{\bm{\mathscr{G}}\mathcal{M}(r)}{r^{2}}=\int_{0}^{M}
\left[\frac{1}{\rho(r)} \frac{dp(r)}{dr}
+\frac{\bm{\mathscr{G}}\mathcal{M}(r)}{r^{2}}\right]
d\mathcal{M}(r)\delta r =0
\end{align}
Hence, the condition for hydrostatic equilibrium follows
\begin{align}
\underbrace{\frac{dp(r)}{dr}=-\frac{\mathscr{G}\mathcal{M}(r)\rho(r)}{r^{2}}}
\end{align}
\end{proof}
\subsection{2nd variation of the thermodynamic Massieu function for an isothermal self-gravitating sphere of perfect gas}
The condition for hydrostastic equilibrium can be established via the isothermal gas sphere via the 2nd variation of the thermodynamic Massieu function for a self-gravitating perfect gas. Equilibrium is the critical point for te Massieu function. The statistical mechanics of a self-gravitating system of particles raises many new issues and conundrums, initiated with the work of Antonov, and which have been discussed extensively in a number of works. (See $\bm{[7,8,9]}$] and reference their in). However, one can still (carefully) apply statistical mechanics to a Newtonian gas of particles in a finite domain. In what follows the canonical ensemble is utilised, whereby the gas is at fixed temperature $\Theta$ and energy can fluctuate.

Let $\mathbf{B}(0,R)\subset\mathbf{R}^{3}$ be a spherical domain of radius of radius $R$ and $r\le R$ and volume $\mathrm{V}(\mathbf{B}(0,R))=\tfrac{4}{3}\pi r^{3}$. Let the domain support an isothermal mono-atomic gas of N particles of mass m so that the total mass of the gas is $M=N m$ and the temperature $\Theta$ is fixed. The system is in non-rotating and is not expanding/collapsing. This can describe a star of N hydrogen atoms or a globular cluster of N stars with $N\gg 0$. For a self-gravitating gas there is now a long-range interaction described by the Newtonian potential $\Phi(r)$. Let $\mathfrak{F}(r,\bm{u},t)$ be a distribution function for the velocities and positions. Let $\mathbf{C}$ be a 'cell' of the phase space $\lbrace x_{1},x_{2},x_{3},u_{1},u_{2},u_{3}\rbrace$ so that
\begin{align}
\mathbf{C}
=\lbrace((x_{1},x_{1}+dx_{1}), (u_{2}+du_{2}),
(x_{3},u_{3}+dx_{3}),(u_{1}+ +du_{1}),(u_{2},u_{2}+dx_{1}),(u_{3},u_{3}+dx_{1})\rbrace
\end{align}
where $r=\|\bm{r}\|=\sqrt{x_{1}^{2}+x_{2}^{2}+x_{3}^{2}}$ and $\|\bm{u}\|=\sqrt{u_{1}^{2}+u_{2}^{2}+u_{3}^{2}}$. The volume elements is $d\mathrm{V}=dx_{1}dx_{2}dx_{3}=4\pi r^{2}dr $ and $ d\mathcal{U}=du_{1}du_{2}du_{3}$. The distribution function $\mathfrak{F}(\bm{r},\bm{u},t)d\mathrm{V}d\mathcal{U}$ then gives the number of particles in the cell at time with positions between $x_{1}$ and $x_{1}+dx_{1}$ etc. and velocities $u_{1}+du_{1}$ etc. The spatial density and total mass are then
\begin{align}
&\rho(r,t)=\int_{\mathbf{B}(0,R)}\mathfrak{F}(\bm{r},\bm{u},t)d\mathrm{V}\\&
M=mN=\int_{\mathbf{B}(0,R)}\rho(r)d\mathrm{V}=const.
\end{align}
The total energy $\mathbf{E}$ is now
\begin{align}
\mathbf{E}=\frac{1}{2}\int_{\mathbf{B}(0,R)}\mathfrak{F}|\bm{u}|^{2}d\mathrm{V}
+\frac{1}{2}\int_{\mathbf{B}(0,R)}\rho(r)\Phi(r)d\mathrm{V}=\mathbf{E}_{K}+\mathbf{E}_{G}=\mathbf{E}_{K}+\mathbf{E}_{G}
\end{align}
where $\Phi(r)$ is the usual Newtonian potential. $\mathbf{E}_{K}$ is the kinetic energy of the particles and this can be related to thermal energy for a perfect gas to reproduce since
$t\frac{1}{2}m|\bm{u}|^{2}\sim k_{b}\Theta$. The Newtonian potential in the gas is
\begin{align}
\Phi(r)=-4\pi \mathscr{G}\int_{\mathbf{B}(0,R)}\frac{\rho(r,t)d\mathrm{V}(r)}
{|\bm{r}-\bm{r}^{\prime}|}
\end{align}
The Boltzmann entropy is the standard formula
\begin{align}
\mathcal{S}=-k_{B}\int\!\!\!\int\left|\frac{\mathfrak{F}}{m}\right|\log\left|\frac{\mathfrak{F}}{m}
\right|d\mathcal{U}d\mathrm{V}
\end{align}
In a canonical ensemble in which the temperature $\Theta$ is fixed, the energy is allowed to fluctuate. The relevant thermodynamic potential is
then the Massiue function $\mathcal{J}$ related to the Helmholtz free energy $\mathcal{F}$. They are defined as
\begin{align}
&\mathcal{F}=\mathbf{E}-\Theta\mathcal{S}\\&
\mathcal{J}=-\frac{\mathcal{F}}{\Theta}=\mathcal{S}-\frac{1}{\Theta}\mathbf{E}
\end{align}
The energy density for a perfect isothermal gas is
\begin{align}
\mathcal{E}=\frac{\frac{3}{2}Nk_{B}\Theta}{|V(\mathbf{B}(0,R))|}
\end{align}
For a perfect isothermal gas, the total energy $\mathbf{E}$ and entropy $\mathcal{S}$ become
\begin{align}
\mathbf{E}&=\int_{\mathrm{V}(\mathbf{B}(0,R))}\frac{\frac{3}{2}Nk_{B}\Theta}{\mathrm{V}(\mathbf{B}(0,R))}
d\mathrm{V}+\frac {1}{2}\int\rho(r)\phi(r)d\mathrm{V}\nonumber\\&
=\frac{\frac{3}{2}Nk_{B}\Theta}{\mathrm{V}(\mathbf{B}(0,R))}\int_{\mathrm{V}(\mathbf{B}(0,R))}
d\mathrm{V}+\frac{1}{2}\int\rho(r)\phi(r)d\mathrm{V}\\&
=\frac{\frac{3}{2}Nk_{B}\Theta}{\mathrm{V}(\mathbf{B}(0,R))}
\mathrm{V}(\mathbf{B}(0,R))+\frac{1}{2}\int_{\mathrm{V}(\mathbf{B}(0,R))}
\rho(r)\phi(r)d\mathrm{V}
\nonumber\\&
=\frac{3}{2}Nk_{B}\Theta+\frac{1}{2}\int_{\mathrm{V}(\mathbf{B}(0,R))}
\rho(r)\phi(r)d\mathrm{V}\equiv
\int_{\mathrm{V}(\mathbf{B}(0,R))}\mathscr{E}d\mathrm{V}+
\mathbf{E}_{G}=\mathbf{E}_{T}+\mathbf{E}_{G}
\end{align}
The entropy (4.19) for the perfect gas is
\begin{align}
\mathcal{S}=\frac{3k_{B}N}{2}+
\frac{3k_{B}N}{2}\log\left(\frac{2\pi k_{B}\Theta}{m}
\right)-k_{B}\int_{\mathbf{B}(0,R)}\left|\frac{\rho(r)}{m}\right|
\log\left|\frac{\rho(r)}{m}\right|d\mathrm{V}
\end{align}
The main theorem for the isothermal gas sphere is now stated.
\begin{thm}
Let a spherical domain $\mathbf{B}(0,R)$ of radius R support a gas obeying the perfect gas law at fixed temperature so that $d\Theta(r)/dr=0$. Then the entropy and energy are given by (4.23) and (4.24). The Massieu thermodynamic potential $\mathcal{J}=\mathcal{S}-\frac{\mathcal{E}}{\Theta}$ is then
\begin{align}
\mathcal{J}&=\frac{3k_{B}N}{2}+
\frac{3k_{B}N}{2}\log\left(\frac{2\pi k_{B}\Theta}{m}
\right)-k_{B}\int_{\mathbf{B}(0,R)}\left|\frac{\rho(r)}{m}\right|
\log\left|\frac{\rho(r)}{m}\right|d\mathrm{V}\nonumber\\&
-\frac{3}{2}Nk_{B}-\frac{1}{2\Theta}\int_{(\mathbf{B}(0,R))}
\rho(r)\phi(r)d\mathrm{V}
\end{align}
Let the 2nd variation of $\mathcal{J}$ with respect to density $\rho(r)$ be $\bm{\delta}^{2}\mathcal{J}$. Let the first variation of the total mass with respect to density be $\delta M=0$ so that the total mass M is an invariant. Then the constrained optimisation condition with Lagrange multiplier $\chi$
\begin{align}
&\delta^{2}\mathcal{J}-\chi\delta M=\frac{1}{2}\frac{\partial^{2}\mathcal{J}}{\partial \rho(r)^{2}}|\delta\rho(r)|^{2}+\frac{\partial\mathcal{J}}{\partial \rho(r)}\delta\rho(r)-\chi\frac {\partial M}{\partial \rho(r)}(\delta\rho(r))\nonumber\\&=\frac{1}{2}\frac{\partial^{2}\mathcal{S}}{\partial \rho(r)^{2}}|\delta\rho(r)|^{2}+\frac{\partial\mathcal{S}}{\partial \rho(r)}\delta\rho(r)
\nonumber\\&+
\frac{1}{2}\frac{1}{\Theta}\frac{\partial^{2}\mathbf{E}}{\partial \rho(r)^{2}}|\delta\rho(r)|^{2}+\frac{\partial\mathbf{E}}{\partial \rho(r)}\delta\rho(r)-\chi\frac{\partial M}{\partial\rho(r)}(\delta\rho(r))=0
\end{align}
 yields a Boltzmann distribution for the density as a critical point of the free energy so that
 \begin{align}
\rho(r)=\rho_{c}\exp\left(-\frac{m k_{B}\Phi(r)}{k_{N}\Theta}\right)=\rho_{c}\exp\left(-\frac{\Phi(r)}{\mathscr{R}}\right)
\end{align}
The equations of hydrostatic equilibrium (2.82) and (2.83) for an isothermal perfect-gas sphere then follow.
\end{thm}
\begin{proof}
Let $\mathbf{D}_{\rho}=\frac{\partial}{\partial\rho(r)}$ and $\mathbf{D}^{2}_{\rho}=\mathbf{D}_{\rho}\mathbf{D}_{\rho}
=\frac{\partial^{2}}{|\partial\rho(r)|^{2}}$. The constrained optimizational problem to solve is then
\begin{align}
\delta^{(2)}\mathcal{J}-\chi\delta M&=-k_{B}\mathbf{D}_{\rho}\int_{\mathbf{B}(0,R)}\left|\frac{\rho(r)}{m}\right|
\log\left|\frac{\rho(r)}{m}\right|d\mathrm{V}\delta\rho(r)\nonumber\\&-\frac{1}{2}k_{B}
\mathbf{D}_{\rho}\mathbf{D}_{\rho}\int_{\mathbf{B}(0,R)}\left|\frac{\rho(r)}{m}\right|
\log\left|\frac{\rho(r)}{m}\right|d\mathrm{V}|\delta\rho(r)|^{2}\nonumber\\&
-\frac{1}{\Theta}\int_{\mathbf{B}(0,R)}\rho(r)\Phi(r)d\mathrm{V}\delta\rho(r)
-\frac{1}{2\Theta}\int_{\mathbf{B}0,R)}\rho(r)\Phi(r)d\mathrm{V}\delta\rho(r)\nonumber\\&
-\frac{1}{4\Theta}\mathbf{D}^{2}_{\rho}\int_{\mathbf{B}0,R)}\rho(r0\Phi(r)d\mathrm{V}|
\delta\rho(r)|^{2}-\chi\mathbf{D}_{\rho}\int_{\mathbf{B}0,R)}
\rho(r)d\mathrm{V}\delta\rho(r)\nonumber\\&
=-k_{B}\int_{\mathbf{B}(0,R)}\left(\frac{\rho(r)}{m}\mathbf{D}_{\rho}(\log\rho(r)-\log m)+\frac{1}{m}\log\left|\frac{\rho(r)}{m}\right|\right)d\mathrm{V}\delta\rho(r)\nonumber\\&
-k_{B}\int_{\mathbf{B}(0,R)}\mathbf{D}_{\rho}\left(\frac{\rho(r)}{m}\mathbf{D}_{\rho}(\log\rho(r)-\log m)+\frac{1}{m}\log\left|\frac{\rho(r)}{m}\right|\right)d\mathrm{V}\delta\rho(r)
\nonumber\\&
-\frac{1}{\Theta}\int_{\mathbf{B}(0,R)}\Phi(r)d\mathrm{V}\delta\rho(r)
-\frac{1}{4\Theta}\int_{\mathbf{B}(0,R)}\delta\Phi(r)d\mathrm{V}\delta\rho(r)
-\chi\int_{\mathbf{B}(0,R)}d\mathrm{V}\delta\rho(r)\nonumber\\&
=-\frac{K_{b}}{m}\int_{\mathbf{B}(0,R)}\left(1+
\log\left|\frac{\rho(r)}{m}\right|\right)d\mathrm{V}\delta\rho(r)
-\frac{1}{2}\frac{k_{B}}{m}\int_{\mathbf{B}(0,R)}
\frac{|\delta\rho(r)|^{2}}{\rho(r)}d\mathrm{V}\nonumber\\&-\frac{1}{\Theta}
\int_{\mathbf{B}(0,R)}\phi(r)\delta\rho(r)d\mathrm{V}-\frac{1}{2\Theta}
\int_{\mathbf{B}(0,R)}\delta\rho(r)\delta\phi(r)d\mathrm{V}-
\chi\int_{\mathbf{B}(0,R)}d\mathrm{V}\delta\rho(r)\nonumber\\&
=-\int_{\mathbf{B}(0,R)}d\mathrm{V}\delta\rho(r)\left\lbrace 1+\log
\left|\frac{\rho(r)}{m}\right|+\frac{\Phi(r)}{\Theta}
\right\rbrace\delta\rho(r)d\mathrm{V}\nonumber\\&-\frac{1}{2}\frac{K_{B}}{m}
\int_{\mathbf{B}(0,R)}\frac{|\delta\rho(r)|^{2}}{\rho(r)}d\mathrm{V}-\frac{1}{2\Theta}
\int_{\mathbf{B}(0,R)}\delta\rho(r)\delta\phi(r)d\mathrm{V}-
\chi\int_{\mathbf{B}(0,R)}d\mathrm{V}\delta\rho(r)
\end{align}
Then to first order
\begin{align}
\delta^{(2)}\mathcal{J}-\chi\delta M=-\int_{\mathbf{B}(0,R)}d\mathrm{V}\delta\rho(r)\left\lbrace 1+\log
\left|\frac{\rho(r)}{m}\right|+\frac{\Phi(r)}{\Theta}+\chi
\right\rbrace\delta\rho(r)d\mathrm{V}=0
\end{align}
This is satisfied if
\begin{align}
\chi=-\frac{\Phi(r)}{\Theta}-\frac{k_{B}}{m}
\left(\log\left|\frac{\rho(r)}{m}\right|+1\right)=C
\end{align}
Since the Lagrange multiplier must be constant then $d\chi/dr=0$ so that
\begin{align}
\frac{d\chi}{dr}=-\frac{1}{\Theta}\frac{d\Phi(r)}{dr}
-\frac{k_{B}}{m}\frac{d}{dr}\log\left|\frac{\rho(r)}{m}\right|=0
\end{align}
Multiplying through by $dr$
\begin{align}
-\frac{1}{\Theta}d\Phi(r)\frac{k_{B}}{m}-d\log\rho(r)
\end{align}
then integrating
\begin{align}
\frac{1}{\Theta}\int_{\Phi(0)}^{\Phi(r)}d\Phi(r)=-\frac{k_{B}}{m}\int_{\log\rho(0)}^{\log\rho(r)}
d\overline{\log\rho(r)}=-\frac{k_{B}}{m}\log\left|\frac{\rho(r)}{\rho(0)}
\right|
\end{align}
so that
\begin{align}
\frac{\Phi(r)}{\Theta}=-\frac{k_{B}}{m}\log\left|\frac{\rho(r)}{\rho(0)}\right|
\end{align}
since $\Phi(0)=0$. Solving for the density profile then  gives an equilibrium or stationary Boltzmann density distribution as the critical point of the Massieu thermodynamic potential
\begin{align}
\rho(r)=\rho(0)\exp\left(-\frac{\Phi(r)m}{k_{B}\Theta}
\right)=\rho_{c}\exp(\left(-\frac{\Phi(r)}{\mathscr{R}}
\right)
\end{align}
This must then be associated with a Newtonian gravitational potential $\Phi(r)$ that satisfies a Boltzmann-Poisson equation
\begin{align}
\Delta\Phi(r)=\frac{1}{r^{2}}\frac{d}{dr}\left(r^{2}\frac{d\Phi(r)}{dr}\right)=4\pi \mathscr{G}\rho_{c}\exp\left(-\frac{\Phi(r)}{\mathscr{R}}\right)
\end{align}
describing an equilibrium or stationary configuration. Comparing with (2.83) it can be seen that this is an isothermal perfect gas sphere. Since the gas obeys the idea gas law then the Boltzmann density of the gas obeys the differential equation (2.83)
\begin{align}
\frac{1}{r^{2}}\frac{d}{dr}\left(r^{2}\frac{d}{dr}\log\rho(r)\right)
=\frac{1}{r^{2}}\frac{d}{dr}\left(\frac{r^{2}}{\rho(r)}\frac{d\rho(r)}{dr}\right)=\frac{4\pi \mathscr{G}\rho(r)}{\mathscr{R}}
\end{align}
We now essentially take the reverse of Lemma (2.21). From the isothermal perfect gas law $p(r)=\mathscr{R}\rho(r)$ it follows that
\begin{align}
\frac{d\rho(r)}{dr}=\frac{1}{\mathscr{R}}\frac{dp(r)}{dr}
\end{align}
\begin{align}
\frac{1}{r^{2}}\frac{d}{dr}\left(\frac{r^{2}}{\mathscr{R}\rho(r)}
\frac{dp(r)}{dr}\right)=
\frac{4\pi \mathscr{G}\rho(r)}{\mathscr{R}}
\end{align}
Integrating
\begin{align}
\frac{r^{2}}{\rho(r)}\frac{dp(r)}{dr}=-\int_{0}^{r}4\pi\mathscr{G}\overline{r}^{2}\rho(\overline{r})d\overline{r}
\end{align}
then
\begin{align}
\underbrace{\frac{dp(r)}{dr}=-\frac{\left|\int_{0}^{r}4\pi\mathscr{G}\overline{r}^{2}
\rho(\overline{r})d\overline{r}\right|\rho(r)}{r^{2}}\equiv -\frac{\mathscr{G}\mathcal{M}(r)\rho(r)}{r^{2}}}
\end{align}
which is again the fundamental equation of hydrostatic equilibrium
\end{proof}
\section{Polytropic Gaseous Stars And The Lane-Emden Equation:Derivation, Existence And Uniqueness Of Solutions, Analytical Solutions And Stability}
There are classes of stars for which the pressure is proportional to the density raised to a power; that is,a polytropic gas law. Such polytropic gaseous stars (PGS) have been extensively studied beginning over a century ago. A detailed study is still the classic 1939 work by Chandrasekhar $\bm{[1]}$ and more recently $\bm{[12]}$. When combined with the FEHE this leads to the Lane-Emden stellar structure equation which has both analytical and numerical solutions. While approximations to real stars, these models have been efficacious and useful in stellar structure theory. However, it should be noted that there is still no complete mathematically rigorous treatment of the nonlinear stability of polytropic gas stars.
\begin{defn}
Let p be a function of $\rho>0$ such that $p\ge 0$ and $dp(\rho)/d\rho>0$ for $\rho>0$. Let $\Xi$ be a smooth function on $\mathbf{R}$ such that $\Xi(0)=0$. The most general polytropic gas law has the form
\begin{align}
p=\mathsf{K}|\rho|^{\lambda}(1+\Xi(\rho^{\gamma-1}))
\end{align}
where $\mathsf{K}$ is a constant and the polytropic index is $\gamma\in(1,2)$. Usually $\Xi=0$ so that
\begin{align}
p=\mathsf{K}|\rho|^{\lambda}=\mathsf{K}|\rho|^{1+\frac{1}{n}}
\end{align}
One can define an enthalpy variable $\mathfrak{E}(\rho)$
\begin{align}
\mathfrak{E}=\int_{0}^{\rho}\frac{dp}{\rho}
\end{align}
so that
\begin{align}
\frac{d\mathfrak{E}(\rho)}{d\rho}=\frac{1}{\rho}\frac{dp}{d\rho}
=\frac{1}{\rho(r)}\frac{dp(r)}{dr}\frac{dr}{d\rho(r)}
\end{align}
and the function $\psi(\rho)$
\begin{align}
\psi(\rho)=\int_{0}^{\rho}\mathfrak{E}(\rho)d\rho,~~\frac{d\psi}{d\rho}=\mathfrak{E}(\rho)
\end{align}
For a polytropic gas
\begin{align}
\mathfrak{E}(\rho)=\frac{\mathsf{K}\gamma}{\gamma-1}|\rho|^{\gamma-1}
\end{align}
\end{defn}
\begin{lem}
Let a domain $\mathbf{B}\subset\mathbf{R}^{n}$ with volume $\mathrm{V}=\mathrm{V}(\mathbf{B})$ be filled with a perfect mono-atomic gas at temperature $\Theta$; for example, ionized hydrogen. The domain is adiabatic or 'insulated' so no heat enters or leaves the domain under variations $\delta\mathrm{V}(\mathbf{B})$. For a perfect (non-relativistic) gas comprised of $N$ particles $p=Nk_{b}\Theta\mathrm{V}^{-1}$ and the internal energy is E=$\tfrac{3}{2}N k_{B}\Theta$ and total (fixed) mass is $M=mN$. This gas is then a polytropic gas with $\gamma=5/3$ so that
\begin{align}
p=\mathsf{K}|\rho|^{5/3}
\end{align}
\end{lem}
\begin{proof}
The energy of the perfect gas is $E=\tfrac{3}{2}p\mathrm{V}$ since $pV=Nk_{B}\Theta$ and the differential form of this
\begin{align}
d\mathbf{E}=\frac{3}{2}(\mathrm{V}dp+pd\mathrm{V})
\end{align}
Now the 1st law of thermodynamics states that the differential  $dE$ due to $d\mathrm{V}$ at fixed pressure is $dE=-pd\mathrm{V}$ so that (5.8) becomes
\begin{align}
-pd\mathrm{V}=\tfrac{3}{2}\mathrm{V}dp+\tfrac{3}{2}p d\mathrm{V}
\end{align}
 so that upon integrating
\begin{align}
\int \frac{dp}{p}=-\frac{5}{3}\int \frac{d\mathrm{V}}{\mathrm{V}}
\end{align}
Hence $\log p=-\frac{5}{3}\ln\mathrm{V}=\ln\mathrm{V}^{-5/3}$ and so $p\sim\mathrm{V}^{-5/3}$. Since the mass of the gas is left unchanged and $
\rho=M/\mathrm{V}$ it follows that
\begin{align}
p\sim|\rho|^{5/3}
\end{align}
and so an adiabatic perfect gas is also a polytropic gas with
$p=\mathsf{K}|\rho|^{5/3}$.
\end{proof}
\begin{rem}
The polytropic gas model can be applied to the following types of stars:
\begin{itemize}
\item Stars with efficient convective energy transport obey a polytropic law
$p=\mathsf{K}|\rho|^{5/3}$ so that $\gamma=5/3$ and $\mathsf{K}$ depends on boundary conditions such as $p(0),\rho(0)$.
\item Light white dwarf stars with $\gamma\sim 5/3$ and very heavy white dwarf stars with $\gamma\sim 4/3$. Then $\mathsf{K}$ depends on chemical composition and fundamental constants, and the equation of state can expressed in polytropic form.
    \item Supermassive stars supported predominantly by radiation pressure with
    $p=\mathsf{K}|\rho|^{4/3}$, and $\mathsf{K}$ depends on the ratio of matter and radiation pressure.
\end{itemize}
\end{rem}
The following formulas for $\mathbf{E}_{G}$ and $\mathbf{E}_{T}$ for a polytropic gas star are well known and applies the FEHE $\bm{[1,10,12]}$.
\begin{thm}
($\underline{Betti~and~Ritter}$. The gravitational potential energy $\mathbf{E}_{G}$ and the thermal energy $\mathbf{E}_{T}$ for a polytropic gaseous star with polytropic equation of state
$p(r)=\mathsf{K}|\rho(r)|^{\gamma}=\mathsf{K}|\rho(r)|^{\frac{n+1}{n}}$ is
\begin{align}
&\mathbf{E}_{G}=-\int_{0}^{M}\frac{\mathscr{G}\mathcal{M}(r)d\mathcal{M}(r)}{r}=\frac{3}{5-n}\frac{\mathscr{G}M^{2}}{R}\equiv -\frac{3(\gamma-1)}{(5\gamma-6)}\frac{\mathscr{G}M^{2}}{R}\\&
\mathbf{E}_{T}=\int_{0}^{R}\mathcal{E}4\pi r^{2} dr=\frac{1}{(5\gamma-1)}\frac{\mathscr{G}M^{2}}{R}
=-\frac{1}{(3\gamma-1)}\mathbf{E}_{G}
\end{align}
The total energy is then
\begin{align}
\mathbf{E}=\mathbf{E}_{T}+\mathbf{E}_{G}=-\frac{(3\gamma-4)}{(5\gamma-6)}\frac{\mathscr{G}M^{2}}{R}
\end{align}
\end{thm}
\begin{proof}
The proof utilises the FEHE and the polytropic gas equation. From the polytropic state equation
\begin{align}
\frac{dp(r)}{\rho(r)}=\mathsf{K}\left(\frac{n+1}{n}\right)
|\rho(r)|^{\frac{1}{n-1}}d\rho(r)
=(n+1)d\left(\frac{p(r)}{\rho(r)}\right)
\end{align}
Since there are several changes of variable, it is convenient to label the integrals throughout as $\int_{C}^{S}$, where C refers to any quantity at the centre and S refers to the surface. We take $p(R)=0,\rho(R)=0$ and $\mathcal{M}(0)=0,\mathcal{M}(R)=M$.
\begin{align}
\mathbb{E} &=-\int_{0}^{M}\mathscr{G}\frac{\mathcal{M}(r)d\mathcal{M}(r)}{r}
=-\underbrace{\frac{1}{2}\int_{C}^{S}\frac{\mathscr{G}\mathcal{M}(r)d|\mathcal{M}(r)|^{2}}{r}}_{int~by~parts}\nonumber\\&
=\left|\frac{\mathscr{G}{m}(r)}{2r}\right|_{C}^{S}-\frac{1}{2}\int_{C}^{S}
\frac{\mathscr{G}|\mathcal{M}(r)|^{2}}{r^{2}}dr\nonumber\\&=\left|\frac{\mathscr{G}
\mathcal{M}(r)}{2r}\right|_{C}^{S}-\frac{1}{2}\int_{C}^{S}
\mathcal{M}(r)\underbrace{\left(\frac{\mathscr{G}\mathcal{M}(r)}{r^{2}}dr\right)}_{use~FEHE}\nonumber\\&
=\frac{\mathscr{G}M^{2}}{2R}+\frac{1}{2}\int_{C}^{S}\mathcal{M}(r)\left(\frac{dp(r)}{ \rho(r)}\right)\nonumber\\&
=\frac{\mathscr{G}M^{2}}{2R}+\underbrace{\frac{1}{2}(n+1)\int_{C}^{S}\mathcal{M}(r)\left(\frac{dp(r)}{ \rho(r)}\right)}_{int~by~parts}\nonumber\\&
=-\frac{\mathscr{G}M^{2}}{2R}+\left|\frac{1}{2}(n+1)\mathcal{M}(r)\frac{p(r)}{\rho(r)}\right|_{C}^{S}
-\frac{1}{2}(n+1)\int_{C}^{S}\frac{p(r)}{\rho(r)}d\mathcal{M}(r)\nonumber\\&
=-\frac{\mathscr{G}M^{2}}{2R}-\frac{1}{2}(n+1)\int_{C}^{S}\frac{[p(r)}{\rho(r)}d\mathcal{M}(r)\nonumber\\&
-\frac{\mathscr{G}M^{2}}{2R}-\frac{1}{2}(n+1)\int_{C}^{S}\frac{p(r)}{\rho(r)}4\pi r^{2}\rho(r)dr\nonumber\\&
=-\frac{\mathscr{G}M^{2}}{2R}-\underbrace{\frac{1}{2}(n+1)\int_{C}^{S}p(r)
\frac{4\pi}{3}d|r^{3}|}_{int~by~parts}\nonumber\\&
-\frac{\mathscr{G}M^{2}}{2R}-\left|\frac{1}{2}(n+1)\frac{4\pi}{3}p(r)r^{3}\right|_{C}^{S}
+\frac{1}{6}(n+1)\int_{C}^{S}
\underbrace{4\pi r^{2}dp(r)}_{use~FEHE}\nonumber\\&
=-\frac{\mathscr{G}M^{2}}{2R}-\frac{1}{6}(n+1)\int_{C}^{S}4\pi r^{3}\frac{\mathscr{G} \mathcal{M}(r)}{r^{2}}\rho(r)dr\nonumber\\&
=-\frac{\mathscr{G}M^{2}}{2R}-\frac{1}{6}(n+1)\int_{C}^{S} \frac{\mathscr{G}\mathcal{M}(r)d\mathcal{M}(r)}{r}\nonumber\\&
=-\frac{\mathscr{G}M^{2}}{2R}-\frac{1}{6}(n+1)\mathbf{E}_{G}
\end{align}
so
\begin{align}
\mathbf{E}_{G}=-\frac{GM^{2}}{2R}+\frac{1}{6}(n+1)\mathbf{E}_{G}
\end{align}
which then gives (5.12) as required. To prove (5.12), first note that for a polytropic gas
the energy density $\mathcal{E}$ is related to the pressure by $\mathcal{E}(r)=(\gamma-1)^{-1}p(r)$. Then using the FEHE
\begin{align}
\mathbf{E}_{G}=4\pi\int_{0}^{R}r^{3}\frac{dp(r)}{dr}dr=-12\pi\int_{0}^{R}r^{2}p(r)dr=-12\pi\int_{0}^{R}\mathcal{E}(\gamma-1)r^{2}dr
\end{align}
Hence
\begin{align}
\mathbf{E}_{G}=-(\gamma-1)\mathbf{E}_{T}
\end{align}
and so (5.13) immediately follows using (5.12). The total energy $\mathbf{E}$ also immediately follows.
\begin{cor}
It can be see that $\mathbf{E}=0$ if $\gamma=\tfrac{4}{3}$ so such polytropic stars are teetering on the brink of catastrophic instability. Stability at least requires $\mathbf{E}<0$.
\end{cor}
\end{proof}
\subsection{The Lane-Emden equations}
Since the polytropic gas law is temperature independent, then the structure of such a star can be described by the FEHE in the Poisson form
\begin{align}
\frac{1}{r^{2}}\frac{d}{dr}\left(\frac{r^{2}}{\rho(r)}\frac {dp(r)}{dr}
\right)=-4\pi\mathscr{G}\rho
\end{align}
Using $p(r)=\mathsf{K}|\rho(r)|^{\gamma}$ this becomes
\begin{align}
\frac{1}{r^{2}}\frac{d}{dr}\left(\frac{r^{2}}{\rho(r)}\frac {d}{dr}|\rho(r)|^{\gamma}
\right)+\frac{4\pi\mathscr{G}}{\mathsf{K}}r^{2}\rho=0
\end{align}
with the boundary conditions that $\rho(0)$ has some prescribed value and $\rho^{\prime}(0)=0,\rho(R)=0$.  We can state the following theorem for a polytropic gas equation of state $\bm{[13]}$.
\begin{lem}
Let $\rho(r)$ be a solution for $r_{o}\le r<R$ and let $[r_{o},R]$ be a right maximal interval of existence for $\rho>0,R<\infty,|dp(r)/dr|_{r=r_{o}}<0$. Then $\exists C>0$ such that
\begin{align}
\rho(r)=C|R-r|^{\frac{1}{\gamma-1}}\left(1+\left\llbracket
\left|\frac{R-r}{r}\right|,\overline{C}\left(\frac{R-r}{R}\right)^{\frac{\gamma}{\gamma-1}}
\right\rrbracket
\right)
\end{align}
where $\overline{C}=Q R^{\frac{\gamma}{\gamma-1}}C^{2-\gamma}$ and $Q=4\pi\mathscr{G}(\gamma-1)/\mathsf{K}\gamma$. Here,$\big\llbracket X\rrbracket$ is a notation denoting a power series of the form $\llbracket X\rrbracket_{q}=\sum_{j\ge q}a_{j}X^{j}$ and so
\begin{align}
\big\llbracket X,Y\big\rrbracket_{q}=\sum_{j+k\ge q}q_{jk}X^{j}X^{k}
\end{align}
\end{lem}
\subsubsection{Non-dimensional reduction}
The Emden equation can be reduced to a dimensionless form, the Lane-Emden equation. This is a very standard derivation in the literature $\bm{1,2]}$. Defining a variable $\theta$ by $\theta=|\rho(r)/\rho(0)|^{\gamma-1}$ then
\begin{align}
\frac{1}{r^{2}}\frac{d}{dr}\left(r^{2}\frac{d\theta(r)}{dr}\right)+
\frac{4\pi \mathscr{G}(\gamma-1)}{\mathsf{K}\gamma}|\rho(0)|^{2-\gamma}|\theta(r)|^{1/(\gamma-1)}
\end{align}
Defining a second variable $\xi(r)$ by
\begin{align}
\xi(r)=\left|\frac{4\pi \mathscr{G}(\gamma-1)}{\mathsf{K}\gamma}\right|^{1/2}|\rho(0)|^{(2-\gamma)/2}r
\end{align}
the differential equation becomes
\begin{align}
\frac{1}{\xi^{2}}\frac{d}{d\xi}\left(\xi^{2}\frac{d}{d\xi}\theta(\xi)\right)
+|\theta(\xi)|^{1/(\gamma-1)}=0
\end{align}
with boundary conditions $\theta(0)=1,\theta^{\prime}(0)=0$ and $\theta(\xi_{1})=0$ on the surface. Since $n=\frac{1}{\gamma}-1$ the LE equation is identically.
\begin{align}
\frac{1}{\xi^{2}}\frac{d}{d\xi}\left(\xi^{2}\frac{d}{d\xi}\theta(\xi)\right)
=-|\theta(\xi)|^{n}
\end{align}
This is then the Lane-Emden equation, which can be solved analytically or numerically depending on the polytropic index $n$. The following lemma establishes that if $\theta(\xi)=0$ then one must have $[d\theta(\xi)/d\theta]_{\xi=0}=\theta^{\prime}(0)=0$; that is, the derivative must vanish at $\xi=0$ for bounded solutions that are nonsingular at $\xi=0$. The function $\theta(\xi)$ is analytic at $\xi=0$.
\begin{lem}
Given the Lane-Emden equation
\begin{align}
\frac{1}{\xi^{2}}\frac{d}{d\xi}\left(\xi^{2}\frac{d\theta(\xi)}{d\theta}\right)\equiv
\frac{1}{\xi^{2}}(\xi^{2}\theta^{\prime}(\xi))^{\prime}=f(\theta(\xi))
\end{align}
where $f(\theta(\xi))=-|\theta(\xi)|^{n}$. Let $f\in C(\mathbf{R})$ and $\theta\in C^{2}(0,Q)$ for some $Q>0$ and $\theta\in[0,Q]$ and $\frac{1}{\xi^{2}}(\xi^{2}
\theta^{\prime}(\xi))^{\prime}
=f(\theta)$. Then $\lim_{\xi\uparrow\infty}\theta^{\prime}(\xi)=0$ if $\theta(0)=0$.
\end{lem}
\begin{proof}
Let $\xi\in[0,Q]$ and express the LE equation as
\begin{align}
(\xi^{2}\theta^{\prime}(\xi))^{\prime}=\xi^{2}f(\theta(\xi))
\end{align}
then integrating
\begin{align}
\xi^{2}\theta^{\prime}(\xi)=\int_{\xi_{o}}^{\xi}\beta^{2}f(\theta(\beta))d\beta+\xi_{o}^{2}
\theta(\xi_{o})
\end{align}
and
\begin{align}
0=\int_{\xi_{o}}^{0}\beta^{2}f(\theta(\beta))d\beta+\xi_{o}^{2}
\theta(\xi_{o})
\end{align}
so that
\begin{align}
\int_{\xi_{o}}^{0}\beta^{2}f(\theta(\beta))d\beta=-\xi_{o}^{2}
\theta(\xi_{o})
\end{align}
Then (5.29) becomes
\begin{align}
\xi^{2}\theta^{\prime}(\xi)&=\int_{\xi_{o}}^{\xi}\beta^{2}f(\theta(\beta))d\beta+\xi_{o}^{2}
\theta(\xi_{o})\nonumber\\&
=\int_{\xi_{o}}^{0}\beta^{2}f(\theta(\beta)d\beta+\int_{0}^{\xi}\beta^{2}f(\theta(\beta))d\beta+
\xi_{o}^{2}\theta(\xi_{o})\nonumber\\&
=-\int_{\xi}^{0}\beta^{2}f(\theta(\eta)d\eta+\int_{0}^{\xi}\beta^{2}f(\theta(\beta))d\beta+
\xi_{o}^{2}\theta(\xi_{o})\nonumber\\&
=-\xi_{o}^{2}\theta(\xi_{o})+\xi_{o}^{2}\theta(\xi)
+\int_{0}^{\xi}\beta^{2}f(\theta(\beta))d\beta=\int_{0}^{\xi}\beta^{2}f(\theta(\beta))d\beta
\end{align}
Now let $\mathcal{L}$ be the limit as$\xi\rightarrow\infty$
\begin{align}
\lim_{\xi\uparrow 0}\xi^{2}\theta^{\prime}(\xi)=\mathcal{L}
\end{align}
and consider two cases:
\begin{enumerate}
\item Assume that the limit is strictly greater than zero so that $\mathcal{L}>0$. Then
\begin{align}
\lim_{\xi\uparrow o}\theta^{\prime}(\xi)=\frac{\mathcal{L}}{\xi^{2}}
\end{align}
If some $\alpha\ge 0$ then
\begin{align}
\lim_{\xi\uparrow 0}\theta^{\prime}(\xi)\ge \frac{\mathcal{L}}{\alpha\xi^{2}}
\end{align}
and let $\xi\in [0,\epsilon]$. Integrating
\begin{align}
(\theta(\epsilon)-\theta(0))=
\lim_{t\uparrow\infty}(\theta(\epsilon)-\theta(\xi))\ge \lim_{t\uparrow\infty}
\int_{\xi}^{\epsilon}\frac{\mathcal{L}}{\alpha\beta^{2}}d\beta
=\lim_{t\uparrow\infty}\frac{\mathcal{L}}{\alpha\xi}-\frac{\mathcal{L}}{\alpha\epsilon}=-\infty
\end{align}
This then contradicts the requirement that $\theta(0)=0$. The derivative at $\xi=0$ would then be singular and so does not exist.
\item
The case $\mathcal{L}=0$. Now $\lim_{\xi\uparrow 0}\theta^{\prime}(\xi)=\mathcal{L}=0$. Using (5.30) it is then required that
\begin{align}
\lim_{\xi\uparrow 0}\xi^{2}\theta^{\prime}(\xi)=\lim_{\xi\uparrow 0}
\int_{0}^{\xi}\beta^{2}f(\theta(\beta))d\eta=0
\end{align}
or
\begin{align}
\lim_{\xi\uparrow 0}\theta^{\prime}(\xi)=\frac{1}{\xi^{2}}\lim_{\xi\uparrow 0}
\int_{0}^{\xi}\beta^{2}f(\theta(\beta))d\beta=0
\end{align}
In terms of the supremum of $f(\theta)$
\begin{align}
\lim_{\xi\uparrow 0}\theta^{\prime}(\xi)\le\frac{1}{\xi^{2}}\lim_{\xi\uparrow 0}
\int_{0}^{\xi}\beta^{2}\sup|f(\theta(\beta))|d\beta=\frac{1}{\xi^{2}}\lim_{\xi\uparrow 0}
\int_{0}^{\xi}\beta^{2}C|d\beta=\frac{1}{\xi^{2}}\int_{0}^{\xi}\beta^{2}C
=\lim_{\xi\uparrow 0}\frac{1}{3}C\xi=0
\end{align}
so that $\lim_{\xi\uparrow 0}\theta^{\prime}(\xi)=0$ as required.
\end{enumerate}
\end{proof}
Within the LE formalism, the total energy can also be expressed as
\begin{align}
&\mathbb{E}=\frac{4-3\gamma}{\gamma-1}\int_{0}^{R}p(r) 4\pi r^{2}dr\\&=\frac{4-3\gamma}{\gamma-1} 4\pi \mathsf{K}\left(
\frac{\mathsf{K}\gamma}{4\pi\mathscr{G}(\gamma-1)}
\right)^{3/2}|\rho(0)|^{\frac{3\gamma-6}{2}}\int_{0}^{\xi_{1}}|\theta(\xi)|^{n+1}\xi^{2}d\xi
\end{align}
For $\gamma>6/5$, the solution vanishes at $\xi=\xi_{1}$ so the radius of a polytropic star is
\begin{align}
R=\left(\frac{4\pi\mathscr{G}(\gamma-1)}{\mathsf{K}\gamma)^{-1/2}}\right)|\rho(0)|^{(2-\gamma)/2}\xi
\end{align}
The mass is
\begin{align}
M=\int_{0}^{R}4\pi r^{2}\rho(r)dr&=|4\pi\rho(0)|^{(3\gamma-4)/2}\left(\frac{\mathsf{K}\gamma}{4\pi \mathscr{G}(\gamma-1)}
\right)^{3/2}\int_{0}^{\xi_{1}}\xi^{2}|\theta(\xi)|^{1/\gamma-1}d\xi\nonumber\\&
=|4\pi\rho(0)|^{(3\gamma-4)/2}\left(\frac{\mathsf{K}\gamma}{4\pi \mathscr{G}(\gamma-1)}
\right)^{3/2}\xi_{1}^{2}\left|\frac{d\theta(\xi)}{d\xi}
\right|_{\xi=\xi_{1}}
\end{align}
These formulas are useful in determining the masses and radii of white dwarf stars.
\subsection{Existence and uniqueness}
Having formulated the Lane-Emden equations, it is necessary to establish that a solution exists and is unique. The following general theorem appears in $\bm{[14]}$.
\begin{thm}
Let $u(x)$ be a positive solution to the general Lane-Emden equation
\begin{align}
\Delta u+ u^{n}=0,~~x\in\mathbf{R}^{N},N>2
\end{align}
If
\begin{align}
1\le n<\frac{N+2}{N-2}
\end{align}
then $u(x)=0$. Otherwise there is a non-trivial positive solution.
\end{thm}
The case N=2 is relevant to stellar structure which is our concern here, but for
$n=(N+2)/(N-2)$ and for $N=4$ it is relevant to Yang Mills theory, while the case $N=2$
is relevant to differential geometry. The coupled Lane-Emden system
\begin{align}
&\Delta u+ v^{n}=0\nonumber\\&
\Delta v + u^{m}=0\nonumber
\end{align}
also arises in many mathematical, biological, chemical and physical problems. $\bm{[15-20]}$

We are really concerned with the existence and uniqueness of solutions of the LE equation in the spherically symmetric 'astrophysical form' so that
\begin{align}
&\Delta_{\xi}\theta(\xi)=\frac{1}{\xi^{2}}\frac{d}{d\xi}
\left(\xi^{2}\frac{d\theta(\xi)}{d\xi}\right)\equiv
\theta^{\prime\prime}(\xi)+\frac{2}{\xi}\theta^{\prime}(\xi)+|\theta(\xi)|^{n}=0,n>0\\&
\theta(0)=\alpha,~~~\theta^{\prime}(0)=0
\end{align}
This equation can be considered as a special case of a general equation of the form
\begin{align}
&\theta^{\prime\prime}(\xi)+\mathcal{P}(\xi)\theta^{\prime}(\xi)+\mathcal{Q}(\theta(\xi),\xi)=0,~~\xi>0\\&
\theta(0)=\alpha,~~\theta^{\prime}(0)=0
\end{align}
where $\mathcal{P}:\mathbf{R}\rightarrow\mathbf{R}$, $\mathcal{Q}:\mathbf{R}\times\mathbf{R}\rightarrow\mathbf{R}$, $\alpha>0$ and $p(\xi)$ can be singular at $\xi=0$ $\bm{[21]}$.
\begin{prop}
The following general conditions are imposed on $\mathcal{P}$ and $\mathcal{Q}$
\begin{enumerate}
\item  $\mathcal{P}$ is measurable on $[0,1]$,
\item $\mathcal{P}\ge 0$ on $(0,1]$.
\item $\int_{0}^{1}s\mathcal{P}(s)ds<\infty$
\item $\mathcal{Q}$ satisfies the general Caratheodry conditions and is bounded so that:
\begin{itemize}
\item  $\exists,a,b$ with $a<\alpha<\beta$ and $C>0$ such that
for all $t\in(0,1], \mathcal{Q}(\xi,\bullet)$ is continuous on $(a,b)$.
\item For all $x\in(a,b)$, $\mathcal{Q}(\xi,\bullet)$ is measurable on $(0,1]$.
\item $\sup_{(t,x)\in[0,1]\times[a,b]}|\mathcal{Q}(\xi,x)|<C.$
\end{itemize}
\end{enumerate}
\end{prop}
\begin{defn}
$\theta(\xi)$ is considered  a solution to the LE equation iff $\exists~T>0$ such that:
\begin{enumerate}
\item $\theta$ and $\theta^{\prime}$ are absolutely continuous on $(0,T]$
\item $\theta$ satisfies (5.47) almost everywhere on $(0,T]$.
\item $\theta$ satisfies the conditions of (5.48).
\end{enumerate}
\end{defn}
\begin{lem}
Given $\mathcal{P}(\xi)$, define a function $\psi(\xi)$ by the integral
\begin{align}
\psi(\xi)=\exp\left(\int_{1}^{\xi}\mathcal{P}(s)ds\right),~\xi>0
\end{align}
where $\psi(\xi)$ is a nonnegative, nondecreasing and continuous function which is bounded
then the ODE (5.47) can be expressed in the form
\begin{align}
\frac{1}{\psi(\xi)}\frac{d}{d\xi}(\psi(\xi)\theta^{\prime}(\xi))=-\mathcal{Q}(\theta(\xi),\xi)
\end{align}
\end{lem}
\begin{proof}
Since
\begin{align}
\mathcal{P}(\xi)=\frac{\psi^{\prime}(\xi)}{\psi(\xi)}
\end{align}
then (-) becomes
\begin{align}
&\theta^{\prime\prime}(\xi)+\frac{\psi^{\prime}(\xi)}{\psi(\xi)}
\theta^{\prime}(\xi)+\mathcal{Q}(\theta(\xi),\xi)=0,~~\xi>0\\&
\theta(0)=\alpha,~~\theta^{\prime}(0)\nonumber\\&
=\frac{1}{\psi(t)}\big(\psi(\xi)\theta^{\prime\prime}(\xi)+
\psi^{\prime}(\xi)\theta^{\prime}(\xi)\big)+\mathcal{Q}(\theta(\xi),\xi)=0
\end{align}
which is (5.52)
\end{proof}
The standard Lane-Emden equation is the case $\psi(\xi)=\xi^{-2}$.
\begin{thm}
Let the conditions of Prop 4.9 be satisfied. Then a solution to the boundary value
problem (5.47) exists.
\end{thm}
\begin{proof}
The full proof is quite detailed and is given in $\bm{[21]}$.
\end{proof}
We will however, consider the uniqueness problem in detail. First two preliminary lemmas are required.
\begin{lem}
The solution of the ODE
\begin{align}
\frac{1}{\psi(\xi)}\frac{d}{d\xi}(\psi(\xi)\theta^{\prime}_{i}(\xi))=-\mathcal{Q}(\theta_{i}(\xi),\xi),~i=1,2
\end{align}
has the integral representation
\begin{align}
\theta_{i}(\xi)=\alpha-\int_{0}^{\xi}\left(\int_{s}^{\xi}\frac{\psi(s)}{\psi(\xi)}d\xi
\right)\mathcal{Q}(\theta_{i}(s),s)ds
\end{align}
\end{lem}
\begin{proof}
From (-)
\begin{align}
\psi(\xi)\frac{d\theta(\xi)}{d\xi}=-\int_{0}^{\xi}\psi(\xi)\mathcal{Q}(\xi,\theta_{i}(\xi))ds
\end{align}
so that
\begin{align}
\frac{d\theta(\xi)}{d\xi}=-\frac{1}{\psi(\xi)}\int_{0}^{\xi}\psi(\xi)\mathcal{Q}(\xi,\theta_{i}(\xi))ds
\end{align}
Integrating once more for $\xi\in[0,T]$ gives
\begin{align}
&\theta(\xi)=\alpha-\int_{0}^{t}\frac{1}{\psi(\xi)}\left(\int_{0}^{\xi}
\psi(s)\mathcal{Q}(\xi,\theta(\xi))ds\right)d\xi\nonumber\\&
=\alpha-\int_{0}^{\xi}\left(\int_{s}^{\xi}\frac{\psi(s)}{\psi(\xi)}d\xi\right)\mathcal{Q}(\theta_{i}(s),s)ds
\end{align}
where the integration order has been changed.
\end{proof}
The next lemma is the basic Gronwall lemma
\begin{lem}
Let $f:\mathbf{R}\rightarrow\mathbf{R}$ and let $C_{1},C_{2}>0$. Then if
\begin{align}
f(x)\le C_{1}+C_{2}\int_{I}f(\bar{x})d\bar{x}
\end{align}
where $I=[x_{0},x]$ is an interval then one has the estimate or bound
\begin{align}
f(x)\le C_{1}\exp\left(C_{2}\int_{I}dx\right)=C_{1}\exp(C_{2}|x-x_{0}|)
\end{align}
and $f(x)\le 0$ if $C_{1}=0$. The result holds with $>$ replacing $\le $.
\end{lem}
The main theorem on uniqueness of solutions can now be stated and proved.
\begin{thm}
Let the conditions of definition (5.10) hold and let $\mathcal{Q}$ be Lipschitz on $[a,b]$. Then the nonlinear elliptic boundary value problem (5.47) has a unique solution.
\end{thm}
\begin{proof}
If $\mathcal{Q}$ is Lipschitz on $\theta(\xi)$ on an interval $[a,b]$ then $\exists$ constant $\Lambda>0$ such that
\begin{align}
|\mathcal{Q}(\xi,\theta_{1}(\xi)-\mathcal{Q}(\xi_{2},\theta_{2}(\xi)|\le \Lambda|\theta_{1}(\xi)-\theta_{2}(\xi)|
\end{align}
for all $\theta_{1},\theta_{2}\in(a,b)$ and $\xi\in[0,1]$, where $\theta_{1},\theta_{2}$ are any two solutions to (5.47) on $[0,T]$. If $s\le \xi\le\zeta$ then $\psi(s)/\psi(\xi)|\ll 1$. Now using the integral representation (5.57)
\begin{align}
|\theta_{2}(\zeta)-\theta_{1}(\zeta)|&=\left|
\int_{0}^{\zeta}\left(\int_{s}^{\zeta}\frac{\psi(s)}{\psi(\xi)}d\xi
\right)[\mathcal{Q}(\theta_{2}(s),s)-\mathcal{Q}(s,\theta(s))ds\right|\nonumber\\&
=\int_{0}^{\zeta}\left(\int_{s}^{\zeta}d\xi
\right)\big|\mathcal{Q}(\theta_{2}(s),s)-\mathcal{Q}(\theta_{1}(s),s)\big|ds\nonumber\\&\le
\int_{0}^{\zeta}\left(|\zeta-s|
\right)\big|\mathcal{Q}(\theta_{2}(s),s)-\mathcal{Q}(\theta_{1}(s),s)\big|ds\nonumber\\&
\le \Lambda\int_{0}^{\zeta}\left(|\zeta-s|
\right)\big|\theta_{2}(s)-\theta_{1}(s)\big|ds\nonumber\\&
\le \Lambda\int_{0}^{\zeta} B \big|\theta_{2}(s)-\theta_{1}(s)\big|ds
\end{align}
Applying the Gronwall Lemma to
\begin{align}
|\theta_{2}(\zeta)-\theta_{1}(\zeta)|\le
\int_{0}^{\zeta}C_{2}\big|\theta_{2}(s)-\theta_{1}(s)\big|ds
\end{align}
where $C_{2}=\lambda B$ and $C_{1}=0$ gives $|\theta_{2}(\zeta)-\theta_{1}(\zeta)|\le 0$
so that one must have $\theta_{1}(\zeta)=\theta_{2}(\zeta)$ for any two solutions. Hence, the solution to the Lane-Emden nonlinear elliptic boundary value problem is unique.
\end{proof}
\subsection{Transformations of the LE equations and a dynamical systems perspective}
The LE equations can be transformed to alternative forms which are useful for their solution $\bm{[1,12]}$. They can be recast into an autonomous 1st-order system of ODEs much like a dynamical system.
\begin{lem}
The transformation $\theta=\chi/\xi$ reduces the LE equations to
\begin{align}
\frac{d^{2}\xi}{d\xi^{2}}=-\frac{\chi^{n}}{\xi^{n-1}}
\end{align}
The Kelvin transform $x=1/\xi$ reduces the LE equation to
\begin{align}
x^{4}\frac{d{2}\theta}{dx^{2}}=-|\theta|^{n}
\end{align}
with the solution
\begin{align}
\theta=\alpha x^{\eta},~~\eta=\frac{2}{n-1},~~\alpha
=\left(\frac{2(n-3)}{(n-1)}\right)^{\frac{1}{n-1}},~~n>3
\end{align}
\end{lem}
\begin{proof}
It is easy to show that (5.66) is a solution of (5.65).
\end{proof}
\begin{lem}
The substitution
\begin{align}
\theta(x)=\beta x^{n}w(x)
\end{align}
transforms $x^{4}\frac{d{2}\theta}{dx^{2}}=-|\theta|^{n}$ to the form
\begin{align}
x^{2}\frac{d^{2}w}{dx^{2}}+2\eta x\frac{d w}{dx}+\eta(\eta-1)w+|w|^{n}=0
\end{align}
\end{lem}
\begin{proof}
The derivatives of (5.67) are
\begin{align}
&\frac{d\theta}{dx}=\beta x^{\eta}\frac{dw}{dx}+
\eta\beta x^{\eta-1}w(x)\\&
\frac{d^{2}\theta}{dx^{2}}=\beta(x^{\eta}\frac{d^{2}w}{dx^{2}}+2\eta x^{\eta-1}\frac{d w}{dx}+\eta(\eta-1)x^{\eta-2}w
\end{align}
Substituting into (5.65) then gives (5.68)
\end{proof}
\begin{lem}
Let
\begin{align}
x=\frac{1}{\xi}=exp(w),~~w=\log x=-\log\xi,~~\lambda=1
\end{align}
then (5.68) becomes
\begin{align}
\frac{d^{2}w}{dx^{2}}+\left(\frac{5-n}{n-1}\right)\frac{dw}{dz}-
\frac{2(n-3)}{(n-1)^{2}}w+|w|^{n}=0
\end{align}
or
\begin{align}
\frac{d^{2}w}{dx^{2}}
=-2\Xi\left(w,\frac{dw}{dz}\right)
\end{align}
where
\begin{align}
\Xi\left(w,\frac{dw}{dz}\right)=\frac{1}{2}\left(-\frac{n-5}{n-1}\right)
\frac{dw}{dz}+
\frac{2(3-n)}{(n-1)^{2}}w+|w|^{n}
\end{align}
\end{lem}
\begin{cor}
Equation (5.72) can be recast as a first-order system of autonomous ODES, much like a dynamical system so that
\begin{align}
&\frac{dw}{dz}=\mathcal{Q}\\&
\frac{d\mathcal{Q}}{dz}
=-2\Xi\left(w,\frac{dw}{dz}\right)\equiv
-2\Xi\left(w,\mathcal{Q}\right)
\end{align}
and the critical points then lie along
\begin{align}
\Xi(w,0)=0
\end{align}
These can be used to study stability. (REF)
\end{cor}
\subsection{explicit solutions}
The Lane-Emden system $LE_{n}$  can be solved analytically $\bm{[1,2,12]}$ for three values of $\gamma=1+\frac{1}{n}$, namely $n=0,1,5$.
\begin{itemize}
\item If $n=0$ or $\gamma=\infty$ the LE equation is linear and homogenous with solution
$\theta(\xi)=-\tfrac{1}{6}\xi^{2}$ plus any linear combination of $\xi$ and $1$. Then $LE_{0}$ is
\begin{align}
\frac{1}{\xi^{2}}\frac{d}{d\xi}\left(\xi^{2}\frac{d\theta(\xi)}{d\xi}\right)+1=0
\end{align}
Rearranging and integrating once
\begin{align}
\xi^{2}\frac{d\theta(\xi)}{d\xi}=c_{1}-\frac{1}{3}\xi^{3}
\end{align}
Diving through by $\xi^{2}$ and integrating once more
\begin{align}
\theta(\xi)=c_{o}-\frac{c_{1}}{\xi}-\frac{1}{6}\xi^{2}
\end{align}
The condition $\theta(0)=1$ fixes the solution to be $\theta(\xi)=1-\tfrac{1}{6}\xi^{2}$ with $\theta(0)=\sqrt{6}$ and $\xi_{1}^{2}|d\theta(\xi)/d\xi|_{\xi=\xi_{1}}=-3\sqrt{6}$.
So $c_{0}=1,c_{1}=0$.
\item If $\gamma=2$ then $n=1$ so that
\begin{align}
\frac{1}{\xi^{2}}\frac{d}{d\xi}\left(\xi^{2}\frac{d\theta(\xi)}{d\xi}\right)+\theta=0
\end{align}
Using a power series ansatz $\theta(\xi)=\sum_{n=0}^{\infty}a_{n}\xi^{n}$ leads to  recurrence relation $a_{n+2}=-a_{n}/(n+3)(n+2)$. The solution is
\begin{align}
\theta(\xi)=\frac{a}{\xi}\sin\xi+\tfrac{b}{\xi}\cos\xi
\end{align}
    The condition $\theta(0)=1$ fixes the solution to be $\theta(\xi)=(\sin\xi)/\xi$. with $\xi_{1}=\pi$ and $\xi_{1}^{2}|d\theta(\xi)/d\xi|_{\xi=\xi_{1}}=-\pi$
\item If $\gamma=6/5$ or $n=5$ the solution is $\theta(\xi)=(1+\tfrac{1}{3}\xi^{2})^{-1/2}$ with $\theta(\xi)\rightarrow 0$ as $\xi\rightarrow\infty$ so $\xi_{1}=\infty$ and
     $\xi_{1}^{2}|d\theta(\xi)/d\xi|_{\xi=\xi_{1}}\rightarrow-\sqrt{3}$ for $t\rightarrow\infty$. From (5.41) and (5.42) such polytropic stars have infinite radius but finite mass.
For $n=5$, equation (-) reduces to
\begin{align}
\frac{d^{2}w}{dz^{2}}=\frac{1}{4}w(1-w^{4})
\end{align}
Multiply both sides by $dw/dz$ so that
\begin{align}
\frac{1}{2}\frac{d}{dz}\left(\left(\frac{dw}{dz}\right)^{2}\right)=\frac{1}{4}w(1-w^{4})
\frac{dw}{dz}
\end{align}
Integrating
\begin{align}
\left(\frac{dw}{dz}\right)^{2}=\frac{1}{2}\int(w+w^{5})dw
=\frac{1}{4}w^{2}-\frac{1}{12}w^{6}+C
\end{align}
then (choosing the negative so that $w\rightarrow\infty$)
\begin{align}
\frac{dw}{w(1-\frac{1}{3}w^{4})^{1/2}}=-\frac{1}{2}dt
\end{align}
Now make the substitution
\begin{align}
\frac{1}{3}w^{4}=\sin^{2}\zeta
\end{align}
First integrate it so that
\begin{align}
\frac{4}{3}w^{3}dw=2\sin\zeta\cos\zeta d\zeta
\end{align}
then it follows that
\begin{align}
\frac{dw}{w}
=\frac{1}{2} \frac{\cos\zeta}{\sin\zeta}d\zeta
\end{align}
from (5.88) and (5.89) we obtain
\begin{align}
cosec\zeta d\zeta=-dz
\end{align}
Integrating
\begin{align}
\int cosec\zeta d\zeta= \log[\sin(\tfrac{1}{2}\zeta)]-\log[\cos(\tfrac{1}{2}\zeta)]
=\log[\tan(\frac{1}{2}\zeta]=-\int dz
\end{align}
so
\begin{align}
\tan(\frac{1}{2}\zeta)=C\exp(-z)
\end{align}
Then
\begin{align}
\frac{1}{3}w^{4}=
\frac{4\tan^{2}(\frac{1}{2}\zeta)}{1+\tan^{2}(\frac{1}{2}\zeta))^{2}}
\end{align}
which is
\begin{align}
w=\left(\frac{12\tan^{2}(\frac{1}{2}\zeta)}{1+\tan^{2}(\frac{1}{2}\zeta))^{2}}\right)^{1/4}
\end{align}
Using (5.92)
\begin{align}
w= \left(\frac{12C^{2}\exp(-2z)}{1+C^{2}\exp(-2z))^{2}}\right)^{1/4}
\end{align}
then
\begin{align}
\theta(\xi)=\left(\frac{3C^{2}}{1+C^{2}\xi^{2})^{2}}\right)^{1/4}
\end{align}
Since $\theta(-)=1$ it follows that $C^{2}=1/3$ so the solution is
\begin{align}
\theta(\theta)=\left(\frac{1}{1+\frac{1}{3}\xi^{2}}\right)^{1/4}
\end{align}
\item Other values of n or $\gamma$, also the most interesting and astro-physically relevant, require numerical solutions. The case $ \tfrac{1}{2}<\gamma<1$ or $3<\gamma<2$ can model neutron stars. The case $\gamma=5/3$ or $n=3/2$ describes stars with convective cores such as red giants and brown dwarfs and stars with perfect gas compositions.
 \item The case $\gamma=1$ or $n=\infty$ corresponds to an isothermal gas sphere and will be considered separately.
\end{itemize}
\subsection{The isothermal gas sphere again}
The isothermal gas sphere was previously discussed. Here, it is reprised again for the case of a polytropic gas and can be described by a Lane-Emden type equation known as the Chandrasekhar equation $\bm{[1]}$.
\begin{lem}
Let a spherical region $\mathbf{B}(0,R)$ of radius R, centre zero, support a polytropic gas
with pressure $p(r)$ and density $\rho(r)$ and polytropic index $\gamma=1$ or $n=\infty$. Then $p=\mathsf{K}|\rho|^{\gamma}=\mathsf{K}|\rho|^{1+\frac{1}{n}}$. The (Poisson)equation of hydrostatic equilibrium of the isothermal sphere is then given by (2.80) and (2.81) as
\begin{align}
\Delta\log(\rho)=\frac{1}{r^{2}}\frac{d}{dr}\left(r^{2}\frac{d}{dr}(\log\rho)\right)=-\frac{4\pi \mathscr{G}\rho(r)}{\mathsf{K}}
\end{align}
or equivalently
\begin{align}
\Delta\Phi(r)=\frac{1}{r^{2}}\frac{d}{dr}\left(r^{2}\frac{d\Phi(r)}{dr}\right)=
4\pi\mathscr{G}\rho_{c}\exp\left(-\frac{\Phi(r)}{\mathsf{K}}\right)
\end{align}
and there is a singular solution
\begin{align}
\rho(r)=\frac{A}{2\pi \mathscr{G}r^{2}}
\end{align}
\end{lem}
Using the substitutions
\begin{align}
\rho=\lambda e^{-\psi},~~r=a\xi=\left| \frac{A}{4\pi G\lambda}\right|^{1/2}\xi
\end{align}
where $\lambda\in\mathbf{R}$, then the ODE (5.98) reduces to the dimensionless Chandrasekhar equation
\begin{align}
\frac{1}{\xi^{2}}\frac{d}{d\xi}\left(\xi^{2}\frac{d\psi(\xi)}{d\xi}\right)=e^{-\psi(\xi)}
\end{align}
which is the Lane-Emden equation for $\gamma=1$ or $n=\infty$. One can choose $\lambda=\rho_{c}$ with $\lim_{\xi\rightarrow 0}\frac{d\psi(0)}{d\xi}=0$ if $\psi(\xi)$ bounded. The boundary conditions are $\psi(0)=0,\tfrac{d\psi(\xi)}{d\xi}=0,\xi=0$. The analysis of the CE is then similar to that of the LE equation.
The CE can be transformed in various ways $\bm{[1]}$ using 
\begin{align}
X=\frac{1}{\lambda}
\end{align}
then the CE becomes
\begin{align}
X^{4}\frac{d^{2}\psi(X)}{dX^{2}}=\exp(-\psi(X))
\end{align}
which has a solution
\begin{align}
\psi(X)=-\log2-2\log X
\end{align}
The Emden transform then introduces a variable Z
\begin{align}
Z=-\psi-2\log X
\end{align}
so that
\begin{align}
\frac{d\psi(X)}{dX}=-\frac{2}{X}-\frac{dZ}{dX}
\end{align}
\begin{align}
\frac{d^{2}\psi(X)}{dX^{2}}=-\frac{2}{X}^{2}-\frac{d^{2}Z}{dX^{2}}
\end{align}
\begin{align}
X^{2}\frac{d^{2}Z}{dX^{2}}=e^{Z}-2=0
\end{align}
To eliminate $X$ set $X=1/\xi=e^{\beta}$ then
\begin{align}
&\frac{dZ}{d\beta}=e^{-\beta}\frac{dZ}{d\beta}\\&
\frac{d^{2}Z}{d\beta^{2}}=e^{-2\beta}\left(\frac{d^{2}Z}{d\beta^{2}}-\frac{dZ}{d\beta}
\right)
\end{align}
giving
\begin{align}
\frac{d^{2}Z}{d\beta^{2}}-\frac{dZ}{d\beta}+e^{Z}-2=0
\end{align}
\section{The Fundamental Hydrostatic Equilibrium Equation (FHEE) And The Euler-Poisson Equations For A Self-Gravitating Gaseous Star}
The previous derivation and applications of the FEHE were mostly heuristic. In this section, the FEHE is now considered and derived from a more formal mathematical perspective. The fundamental description of a self-gravitating gaseous star are the Euler-Poisson system of nonlinear partial differential equations for a self-gravitating gas/fluid. These are essentially the Euler equations of fluid dynamics (with viscosity set to zero) coupled to a Newtonian gravitational potential. The Euler-Poisson system has been studied extensively in many works. (For example $\bm{[22-26]}$ and references their in.)
\begin{defn}
The scenario describing a self-gravitating gas/fluid described by the Euler-Poisson equations is defined as follows.
\begin{enumerate}
\item Let $\mathbf{D}\subset\mathbf{R}^{3}$. Generally, we use the functions
    $\rho:[0,T)\times\mathbf{R}^{3}\rightarrow\mathbf{R}_{\ge 0},p:[0,T)\times\mathbf{R}^{3}\rightarrow\mathbf{R}_{\ge 0}$ and $\Phi:[0,T)\times\mathbf{R}^{3}\rightarrow\mathbf{R}$ for some temporal interval $[0,T)$, to describe the spatio-temporal distributions of mass density, pressure and Newtonian potential of the gas/fluid. Then  $\rho=\rho(\mathbf{x},t),p=p(\mathbf{x},t),\Phi=\Phi(\mathbf{x},t)$.
\item The self-gravitating gas is supported within a spherical region $\mathbf{D}(\lambda(t))\equiv\mathbf{D}(t)$ defined by
\begin{align}
\mathbf{D}(\lambda(t))=\mathbf{D}(t)=supp~\rho(t,\bullet)=\lbrace(\mathbf{x}\in\mathbf{R}^{3}|
\rho(\mathbf{x},t)>0,\lambda(t)\gtrless \lambda(0)=R\rbrace\subset\mathbf{R}^{3}\nonumber
\end{align}
with varying radius $\lambda(t)\gtrless \lambda(0)=R$ depending on whether the self-gravitating gaseous sphere is collapsing or expanding, or $\lambda(t)= \lambda(0)=R$ in hydrostatic equilibrium. Then $\mathrm{V}(\mathbf{D}(t))=\tfrac{4}{3}\pi|\lambda(t)|^{3}$ is the volume at time t.
\item The (moving) boundary is $\partial\mathbf{D}(t)\subset\mathbf{D}(t)$ with surface area $\mathrm{A}(\mathbf{D}(t))=4\pi|\lambda(t)|^{2}$. The boundary conditions are that the pressure and density vanish here so that $\rho(\mathbf{x},t)=p(\mathbf{x},t)=0$ for all $x\in\partial\mathbf{D}(t)$. If $\mathbf{D}(t)\bigcup\mathbf{D}^{c}(t)=\mathbf{R}^{3}$ then the vacuum region is defined by $\rho(\mathbf{x},t)=0$ in $\mathbf{D}^{c}(t)$ and the matter or hydrodynamic region by $\rho(\mathbf{x},t)>0,x\in\mathbf{D}(t)$. The fluid velocity $\mathbf{u}$ is defined only on $\mathbf{D}(t)$ by $\mathbf{u}\equiv u^{i}:[0,T)\times\mathbf{D}(t)\rightarrow \mathbf{R}^{3}$ and there is no definition in $\mathbf{D}^{c}(t)$. The space $\mathbf{D}(0)$ is taken to be precompact. A diffuse boundary condition can also be defined as (ref)
\begin{align}
\lim_{(\mathbf{x},t)\rightarrow(\overline{\mathbf{x}},\overline{t})}\rho(\mathbf{x},t)\nabla_{i}p(\mathbf{x},t))=0
\end{align}
where $(x,t)\in[0,T)\times\mathbf{D}(t)$ and $(\overline{x},\overline{t})\in[0,T)\times\partial\mathbf{D}(t)$, where again $\partial\mathbf{D}(t)$ is the boundary of the volume $\mathbf{D}(t)$, that is, the moving interface between the fluid/gas matter and the exterior vacuum. This boundary condition characterizes the gradually vanishing fluid pressure per mass near the boundary.
\item The Newtonian potential $\Phi$ exists on $\mathbf{R}^{3}=\mathbf{D}(t)\bigcup\mathbf{D}^{c}(t)$, in both the matter and vacuum regions and obeys the Poisson equation
\begin{align}
\Delta\Phi=4\pi\mathscr{G}\rho,~~~in~~\mathbf{R}^{3}
\end{align}
with the general solution
\begin{align}
\Phi(\mathbf{x},t)=-\mathscr{G}{\int}_{\mathbf{D}}\frac{\rho(\mathbf{y},t)d^{3}\mathbf{y}}{|\mathbf{x}-\mathbf{y}|},~~~~
for~(t,\mathbf{x},\mathbf{y})\in [0,T)\times\mathbf{R}^{3}
\end{align}
\item It is generally assumed that the gas is \emph{isentropic} having constant entropy or entropy per nucleon, and obeys a barytropic pressure law of the form $p=p(\rho)$, typically a polytropic gas law
\begin{align}
p=\mathsf{K}|\rho|^{\gamma},~~\mathsf{x}\in\mathbf{D}(t)\in \mathbf{R}^{3}
\end{align}
(Polytropic stars will be considered in detail in Section (5) The set $\lbrace\rho(x,t),p(x,t),\Phi(x,t)\rbrace$ is also a solution to the Euler-Poisson system of equations but this will be considered formally in Section (6). For a gas in hydrostatic equilibrium $\mathbf{u}$=0 and $\rho:\mathbf{R}^{3}\rightarrow\mathbf{R}_{\ge 0},p:\mathbf{R}^{3}\rightarrow\mathbf{R}_{\ge 0}$ and $\Phi:\mathbf{R}^{3}\rightarrow\mathbf{R}$ so that
$\rho=\rho(\mathbf{x}),p=p(\mathbf{x}),\Phi=\Phi(\mathbf{x})$. The volume of the support is now fixed so that $\mathrm{V}(\mathbf{D}(t))\mathrm{V}(
\mathbf{D}(0))=\tfrac{4}{3}\pi R^{3}\equiv\tfrac{4}{3}\pi |\lambda(0)|^{3}$.
\item For self-gravitating fluid-gas systems like stars in equilibrium at radius $R$, it is also natural to assume spherical symmetry and coordinates so that $\rho=\rho(r), p=p(r),\Phi=\Phi(r)$ with central values $\rho(0),p(0)$ and boundary values $\rho(R)=p(R)=0$. The domain $\mathbf{D}(\lambda(t))=\mathbf{D}(\lambda(0))\equiv \mathbf{D}(0)=\mathbf{B}(\mathbf{x}_{o},R) $ is now a static ball with centre $\mathbf{x}_{o}$ an radius $R$. Then
    \begin{align}
    \mathbf{D}(\lambda(0)=R)=\mathbf{B}(0,R)=
    \lbrace \mathbf{x},\mathbf{y}\in \mathbf{R}^{3}|d(x,y)\le R\rbrace
    \end{align}
    The mass contained within a radius $r$ is or ball $\mathbf{B}(0,r)\subset
    \mathbf{B}(0,R)$ is then $\mathcal{M}(r)$ and $\mathcal{M}(R)=M$, the total mass contained within $\mathbf{B}(0,R)$  It is defined as
\begin{align}
\mathcal{M}(r)=\int_{\mathbf{B}(0,r)}4\pi\widehat{r}^{2}\rho(\widehat{r})d{r}\equiv
\int_{0}^{r}4\pi \widehat{r}^{2}\rho(\widehat{r})d\overline{r}
\end{align}
\end{enumerate}
\end{defn}
\begin{defn}
If the fluid/gas is described by a fundamental Boltzmann equation then the continuity of mass and Navier-Stokes equations can be established as a consequence. In the absence of gravitation they are
\begin{align}
&\partial_{t}\rho+\nabla.(\rho \mathbf{u})=0~~~~x\in\mathbf{D}(t), t>0\\&
\rho(\partial_{t}\mathbf{u}+(\bm{u}.\nabla)\bm{u})+\nabla p=\eta\Delta\bm{u}~~~~x\in\mathbf{D}(t), t>0
\end{align}
and with the following boundary and initial conditions.
\begin{align}
&P=\rho=0~on~\partial\mathbf{D}\\&
\mathcal{U}_{\partial\mathbf{D}}=\bm{u}.\bm{n}(t)~on~\partial\mathbf{D}\\&
(\rho(0,\bullet),\bm{u}(0,\bullet))=(\rho_{o},u_{o}),~in~\mathbf{D}(t)
\end{align}
where $\mathcal{U}_{\partial\mathbf{D}}$ is the normal velocity of the moving domain boundary $\partial\mathbf{D}(t)$ with unit normal $\bm{n}(t)$.
If the gas is now coupled to Newtonian gravitation and becomes a self-gravitating gas then one obtains the Euler-Poisson system system of nonlinear PDEs.
\begin{align}
&\partial_{t}\rho+\nabla.(\rho \mathbf{u})=0,~~x\in\mathbf{D}(t), t>0\\&
\rho(\partial_{t}\mathbf{u}+(\bm{u}.\nabla)\bm{u})+\nabla p=-\rho\nabla {\Phi},~~~~x\in\mathbf{D}(t), t>0\\&
\Delta{\Phi}=4\pi \mathscr{G}\rho,~~~~x\in\mathbf{D}(t), t>0\\&p=\rho=0~on~\partial\mathbf{D}\\&
\mathcal{U}_{\partial\mathbf{D}}=\bm{u}.\bm{n}(t)~on~\partial\mathbf{D}\\&
(\rho(0,\bullet),\bm{u}(0,\bullet))=(\rho_{o},u_{o}),~in~\mathbf{D}(t)
\end{align}
and the gas-vacuum boundary is $\mathbf{R}^{3}\symbol{92}\partial\mathbf{D}$.
One can also introduce a mass function $\mathcal{M}(r,t)$
\begin{align}
\mathcal{M}(\rho)=\int_{\mathbf{D}(t)}\rho(x,t)d\mathrm{V}(x)
\end{align}
where $\int_{\mathbf{B}(t)}d\mathrm{V}(x)=\int_{\mathbf{B}(t)} d^{n}x=\mathrm{V}(\mathbf{B}(t))$.
\end{defn}
The standard notations apply so that
\begin{align}
&\nabla.(\rho\bm{u})=\sum_{k=1}^{3}\frac{\partial}{\partial x^{k}}(\rho u^{k})\equiv
\sum_{k=1}^{3}\nabla_{k}(\rho u^{k}),j,k=1,2,3\\&
(\bm{u}.\nabla)\bm{u}=\sum_{k=1}^{3}u^{k}\frac{\partial u^{j}}{\partial x^{k}}\equiv
\sum_{k=1}^{3}u^{k}\nabla_{k}u^{j}\\&
\nabla Z=(\nabla_{1}Z,\nabla_{2}Z,\nabla_{3}Z),~Z=p,\mathscr{B}\\&
\Delta\Phi=\sum_{k=1}^{3}\frac{\partial^{2}}{\partial x^{k}\partial x^{k}}
=\sum_{k=1}^{3}\partial_{k}\partial^{k}\Phi\equiv\nabla^{2}
\end{align}
The EL system can also be expressed in various coordinate systems $\bm{[22]}$. One can also use Lagrangian derivatives $D/Dt$ whereby
\begin{align}
\frac{D}{Dt}=\partial_{t}+\bm{u}.\nabla\equiv \partial_{t}+\sum_{k=1}^{3}u^{k}\nabla_{k}
\end{align}
so that (6.14) and (6.15) become
\begin{align}
&\frac{D\rho}{Dt}+\rho\nabla.\bm{u}=0\\&
\rho\frac{Dv}{Dt}+\nabla p=-\rho\nabla\Phi
\end{align}
The EL system of PDEs can be closed by an equation of state $p=p(\rho)$ with $dp/d\rho>0$ and the polytropic equation can be utilised so that $p=\mathsf{K}|\rho|^{\gamma}$ with $\gamma\in[1,2]$. We also reprise the enthalpy variable $\mathfrak{E}(\rho)$ and the functional $\psi(\rho)$ so that $\mathfrak{E}(\rho)=\int_{0}^{\rho}dp/\rho$ and $\psi(\rho)=\int_{0}^{\rho}\mathfrak{E}(\rho)d\rho$. The notation $EP_{\gamma}$ will denote an Euler-Poisson system with a polytropic gas with index $\gamma$.
\begin{defn}
A solution $\rho=\rho(x,t),\phi=\phi(x,t)$ is a compactly supported classical solution (CSCS) iff $\forall(\rho,u,\phi)\in C^{1}([0,T]\times\mathbf{R}^{2}),
\phi(t,\bullet)\in C^{2}(\mathbf{R}^{2}),\rho>0$,everywhere and the support of $\rho(t,\bullet)$ is compact in $\mathbf{D}(t)$ for all $t\in[0,T]$.
\end{defn}
For any CSCS the solution of the Poisson equation (6.14) is the Newtonian potential
\begin{align}
\phi(x,t)=-\mathscr{G}\int_{\mathbf{D}} \rho\frac{(x^{\prime},t)}{|x-x^{\prime}|}d\mathbf{V}(x^{\prime})
\end{align}
The EP equations have been studied quite extensively and there are expanding and collapsing solutions.
\begin{defn}
The total mass $\mathcal{M}(\rho)$ was defined. The total energy at any time $t>0$ is defined as
\begin{align}
\mathbf{E}(\rho,\bm{u})=&\int_{\mathbf{D}(t)}\frac{1}{2}\rho|\bm{u}|^{2}d\mathrm{V}(x)
+\int_{\mathbf{D}(t)}\psi(\rho)d\mathrm{V}(x)+\int_{\mathbf{D}(t)}\frac{1}{2}\rho\phi d\mathrm{V}(x)\nonumber\\&
=\int_{\mathbf{D}(t)}\left(\frac{1}{2}\rho|\bm{u}|^{2}d\mathrm{V}(x)
+\psi(\rho)\right)d\mathrm{V}(x)-\frac{1}{2}\mathscr{G}\int\int
\frac{\rho(x,t)\rho(x^{\prime},t)}{|x-x^{\prime}|}d\mathrm{V}(x)d\mathrm{V}(x^{\prime}
\end{align}
For a polytropic gas
\begin{align}
\mathbf{E}(\rho,\bm{u})=&\int_{\mathbf{D}(t)}\frac{1}{2}\rho|\bm{u}|^{2}d\mathrm{V}(x)
+\int_{\mathbf{D}(t)}\frac{\mathsf{K}}{\gamma-1}\rho^{\gamma}d\mathrm{V}(x)+\int_{\mathbf{D}(t)}\frac{1}{2}\rho\phi
d\mathrm{V}(x)\nonumber\\&
=\int_{\mathbf{D}(t)}\left(\frac{1}{2}\rho|\bm{u}|^{2}d\mathrm{V}(x)
+\frac{\mathsf{K}}{\gamma-1}\rho^{\gamma}\right)d\mathrm{V}(x)-\frac{\mathscr{G}}{2}\int\int
\frac{\rho(x,t)\rho(x^{\prime},t)}{|x-x^{\prime}|}d\mathrm{V}(x)d\mu(x^{\prime}
\nonumber\\&=\int_{\mathbf{D}(t)}\left(\frac{1}{2}\rho|\bm{u}|^{2}d\mu(x)
+\frac{p}{\gamma-1}\right)d\mathrm{V}(x)-\frac{\mathscr{G}}{2}\int\!\!\!\int
\frac{\rho(x,t)\rho(x^{\prime},t)}{|x-x^{\prime}|}d\mu(x)d\mathrm{V}(x^{\prime})
\end{align}
\begin{thm}
The total energy $\mathbf{E}(\rho,\bm{u})$ is conserved for any solution $(\rho,\bm{u})$ so that
\begin{align}
\frac{d}{dt}\mathbf{E}(\rho,\bm{u})=0
\end{align}
\end{thm}
\begin{proof}
\begin{align}
\frac{d}{dt}\mathbf{E}(\rho,\bm{u})&
=\frac{d}{dt}\int\frac{1}{2}\rho\bm{u}|^{2}d\mathrm{V}=\frac{1}{2}
\int\partial_{t}(\rho\sum_{k}|u_{k}|^{2}d\mathrm{V}\\&
=\int\rho\sum_{k}u^{k}\partial_{t}u^{k}d\mathrm{V}+\int\frac{1}{2}\partial_{t}\rho\sum_{k}|u_{k}|^{2}
d\mathrm{V}\nonumber\\&
=-\int(\frac{1}{2}\rho\sum_{j,k}u^{j}\nabla)_{j}|u^{k}|^{2}
d\mathrm{V}+\int\sum_{k}u^{k}\nabla_{k}p d\mathrm{V}\nonumber\\&+\int\sum_{k}\rho u^{k}\nabla_{k}\Phi d\mathrm{V}+\int\frac{1}{2}\nabla_{j}|\rho u^{j}|\sum_{k}|u_{k}|^{2}d\mathrm{V}\nonumber\\&
=-\int(\frac{1}{2}\rho\sum_{j,k}|u^{j}\nabla_{j}|u^{k}|^{2}d\mathrm{V}
+\int\sum_{k}u^{k}\nabla_{k}p d\mathrm{V}\nonumber\\&+\int\sum_{k}\rho u^{k}\nabla_{k}\Phi+
\int\frac{1}{2}\nabla_{j}|\rho u^{j})\sum_{k}|u^{k}|^{2} d\mathrm{V}\nonumber\\&
=-\underbrace{\int\frac{1}{2}\sum_{j,k}\nabla_{j}(\rho u^{j}|u^{k}|^{2})-\int\sum_{k}u^{k}(\nabla_{k}p+\rho\nabla_{k}\Phi)d\mathrm{V}}\nonumber\\&
=-\int\sum_{k}\rho u^{k}\nabla_{k}\mathfrak{E}d\mathrm{V}+\int\sum_{k}\rho u^{k}\nabla_{k}\Phi d \mathrm{V}\nonumber\\&
=\int \frac{\partial\Psi}{\partial t} d\mathrm{V}+\int\underbrace{\sum_{k}\nabla_{k}(\rho u^{k})}_{use~EP~eqn~again}\Phi d\mathrm{V}\nonumber\\&
=-\frac{d}{dt}\int\Psi d\mathrm{V}-\int\frac{\partial \rho}{\partial t}\Phi d\mathrm{V}
\end{align}
The last term can be re-expressed as
\begin{align}
\int\partial_{t}\rho\Phi d\mathrm{V}&=-\mathscr{G}\int\!\!\!\int
\frac{\partial_{t}\rho(x,t)\rho(x,t)}{|x-x^{\prime}|}d\mathrm{V}(x)d\mathrm{V}(x^{\prime})
\nonumber\\&=
-\frac{1}{2}\frac{d}{dt}\int\!\!\!\int\frac{\rho(x,t)\rho(x,t)}{|x-x^{\prime}|}d\mathrm{V}(x)
d\mathrm{V}(x^{\prime})
=\frac{1}{2}\frac{d}{dt}\int\rho\Phi d\mathrm{V}
\end{align}
Hence
\begin{align}
\frac{d}{dt}\int\frac{1}{2}\rho|u|^{2} d\mathrm{V}+\frac{d}{dt}\int\Psi d\mathrm{V}+\frac{d}{dt}\int\frac{1}{2}\rho\Phi d\mathrm{V}=0
\end{align}
and so the total energy is conserved along t for any compact supported solution $(\rho,\bm{u})$,
\end{proof}
\end{defn}
Solutions of an $EP_{\gamma}$ system have a self-similar re-scaling property under
$\overline{t}=t/\lambda^{1/(2-\gamma)}$ and $\overline{x}=x/\lambda$. If $(\rho,\bm{u})$ is a solution of an $EP_{\gamma}$ system then so is the pair $(\overline{\rho},\overline{\bm{u}})$ given by $\bm{[32,36]}$.
\begin{align}
&{\rho(x,t)}=\lambda^{-\frac{2}{2-\gamma}}\overline{\rho}
(\frac{t}{\lambda^{\frac{1}{2-\gamma}}},\frac{x}{\lambda})\\&
{\bm{u}(x,t)}=\lambda^{-\frac{\gamma-1}{2-\gamma}}\overline{\bm{u}}(\frac{t}{\lambda^{\frac{1}{2-\gamma}}},\frac{x}{\lambda})
\end{align}
Also, given a solution $(p,\phi)$ for the pressure and Newtonian potential, the rescaled pressure and potential are
\begin{align}
&{p(x,t)}=\lambda^{-\frac{2\gamma}{2-\gamma}}\overline{o}
(\frac{t}{\lambda^{\frac{1}{2-\gamma}}},\frac{x}{\lambda})\\&
{\phi(x,t)}=\lambda^{-\frac{2\gamma-2}{2-\gamma}}\overline{\phi}(\frac{t}{\lambda^{\frac{1}{2-\gamma}}},\frac{x}{\lambda})
\end{align}
The mass $M(\rho)$ and the energy $E(\rho,u)$ are rescaled as
\begin{align}
&M(\rho)=\lambda^{\frac{4-3\gamma}{2-\gamma}}
\overline{M(\overline{\rho})}\\&
E(\rho,\bm{u})=\lambda^{\frac{6-5\gamma}{2-\gamma}}
\overline{E(\overline{\rho},\overline{\bm{u}})}
\end{align}
The mass is invariant under a re-scaling when $\gamma=4/3$ and the energy when $\gamma=5/6$.
so the re-scalings are respectively mass-critical and energy critical.

The $EP_{\gamma}$ system applied to a spherical gaseous star is naturally best expressed in spherical symmetry then $r=\|x\|$. Define a ball $\mathbf{B}(\lambda(t),0)$, centred at the origin, then the 'radius of support of the fluid' is $\lambda(t)$ at time $t$. Then
$u=u(r,t),\rho=\rho(r,t), p=p(r,t),\phi=\phi(r,t)$. The $EL_{\gamma}$ system of PDE are then
\begin{align}
&\partial_{t}p+u\partial_{r}\rho+\rho\partial_{r}u+\frac{2}{r}\rho u=0,~r\le \lambda(t)
\\&
\rho(\partial_{t}u+u\partial_{r}u)+\partial_{r}p=-\rho\partial_{r}\Phi~r\le\lambda(t)\\&
\frac{1}{r^{2}}\partial_{r}(r^{2}\partial_{r}\Phi)=4\pi \mathscr{G}\rho,~r\le\lambda(t)\\&
\lim_{r\rightarrow\infty}\Phi(r,t)=0\\&
p(t,\xi(t))=0\\&
\mathrm{V}_{\partial\mathbf{B}}=\dot{\lambda}(t)=\mathrm{V}(t,\lambda(t))\\&
\lambda(0)=\lambda_{o},\rho(0,r)=\rho_{0}(r),u(0,r)=u_{o}(r)
\end{align}
\begin{cor}
For a static or equilibrium solution of the $EP$ system one must have $\bm{u}=0$ and
$\partial_{k}\bm{u}=\partial_{t}\bm{u}=0$. The spherically symmetric $EP$ system then reduces to
\begin{align}
&\partial_{r}p(r)=-\rho(r)\partial_{r}\Phi(r),~~r\le R\equiv\lambda\\&
\frac{1}{r^{2}}(\partial_{r}(r^{2}\partial_{r}\Phi(r))=4\pi\mathscr{G}\rho(r)
\end{align}
or
\begin{align}
&\frac{dp(r)}{dr}=-\rho(r)\frac{d\Phi(r)}{dr}\\&
\frac{1}{r^{2}}\frac{d}{dr}\left(r^{2}\frac{d\Phi(r)}{dr}\right)=4\pi\mathscr{G} \rho(r)
\end{align}
Writing
\begin{align}
\frac{d}{dr}\left(r^{2}\frac{d\Phi(r)}{dr}\right)=4\pi r^{2} \mathscr{G}\rho(r)
\end{align}
then integrating
\begin{align}
r^{2}\frac{d\Phi(r)}{dr}=\int_{0}^{r} 4\pi \mathscr{G} r^{\prime 2}\rho(r^{\prime})dr^{\prime}
\end{align}
Now using (6.49)
\begin{align}
\frac{d \Phi(r)}{dr}=-\frac{1}{\rho(r)}\frac{d p(r)}{dr}=\frac{\int_{0}^{r} 4\pi \mathscr{G} r^{\prime 2}
\rho(r^{\prime})dr^{\prime}}{r^{2}}\equiv \frac{\mathcal{M}(r)}{r^{2}}
\end{align}
and so using (6.51) the FEHE is recovered once more
\begin{align}
\frac{dp(r)}{dr}=-\frac{\mathscr{G}\mathcal{M}(r)\rho(r)}{r^{2}}\nonumber
\end{align}
\end{cor}
\begin{lem}
For an $EP_{\gamma}$ system, the energy equation (6.28) is equivalent to (2.68) with
$\mathbf{E}_{K}=0$ for equilibrium.
\end{lem}
\begin{proof}
\begin{align}
\mathbf{E}(\rho,\bm{u})&=\int\frac{1}{2}\rho|\bm{u}|^{2}d\mathrm{V}+\int\Psi(\rho)
d\mathrm{V}+\int\frac{1}{2}\rho\Phi d\mathrm{V}\nonumber\\&
=\int\frac{1}{2}\rho|\bm{u}|^{2}d\mathrm{V}+\int\frac{p}{|\gamma-1|}d\mathrm{V}
+\int\frac{1}{2}\rho\Phi d\mathrm{V}\nonumber\\&
=\int\frac{1}{2}\rho|\bm{u}|^{2}d\mathrm{V}+\int\mathcal{E}d\mathrm{V}
+\int\frac{1}{2}\rho\Phi d\mathrm{V}\nonumber\\&
=\int\frac{1}{2}\rho(r)|\bm{u}|^{2}d\mathrm{V}+\int\mathcal{E}(r)d\mathrm{V}
-\int\rho\frac{\mathscr{G}\mathcal{M}(r)}{r}
d\mathrm{V}
\end{align}
In equilibrium the 'kinetic energy' term is zero so that $\mathbf{E}_{K}=0.$ Continuing in spherical symmetry
\begin{align}
\mathbf{E}(\rho,0)&=\int\mathcal{E}(r)d\mathrm{V}
-\int\rho\frac{\mathscr{G}\mathcal{M}(r)}{r}
d\mathrm{V}\nonumber\\&=\int\mathcal{E}(r)4\pi r^{3}dr
-\int\rho\frac{\mathscr{G}\mathcal{M}(r)}{r}
4\pi r^{2}dr\nonumber\\&=\int\mathcal{E}(r)4\pi r^{3}dr
-\int\frac{\mathscr{G}\mathcal{M}(r)d\mathcal{M}(r)}{r}=\mathbb{E}_{T}+\mathbb{E}_{G}
\end{align}
\end{proof}
\begin{prop}
Let an EP system of self-gravitating gas/fluid with compact support on some
$\mathbf{D}(t)$ have total energy $E(\rho,\bm{u})$, mass $M(\rho)$, density $\rho=\rho(x,t)$, pressure $p=p(x,t)$ and Newtonian potential $\Phi=\Phi(x,t)$. Let $Q(\rho,\bm{u})=\tfrac{d}{dt}E(\rho,\bm{u})$. The enthalpy variable $\mathfrak{E}(\rho)$ is
$\mathfrak{E}(\rho)=\int_{0}^{\rho}dp/\rho$ and $\psi(\rho)=\int_{0}^{\rho}\mathfrak{E}(\rho)d\rho$ with
barytropic state equation $p=p(\rho)$. Then the total energy is conserved if
$\mathbf{E}(\rho,\bm{u})=\tfrac{d}{dt}\mathbf{E}(\rho,\bm{u})=0$. If the total energy is conserved
\begin{align}
&\frac{d}{dt}\mathbf{E}(\rho,\bm{u})=\frac{d}{dt}\left|\int\frac{1}{2}\rho|\bm{u}|^{2}+\psi(\rho)
+\frac{1}{2}\rho\Phi)d\mathrm{V}(x)\right|\nonumber\\&
=\frac{d}{dt}\left|\int\frac{1}{2}\rho|\bm{u}|^{2}+
\psi(\rho)-\frac{1}{2}\mathscr{G}\int\int \frac{\rho(x,t)\rho(x^{\prime},t)}{|x-x^{\prime}}|d\mathrm{V}(x)d\mathrm{V}(x^{\prime})
\right|=0
\end{align}
Then the hydrostatic equilibrium equation
\begin{equation}
\nabla_{k}p+\rho\nabla_{k}\Phi=0
\end{equation}
holds iff $u=0$.
\end{prop}
\begin{proof}
Hence $\tfrac{d}{dt}\mathbf{E}(\rho,\bm{u}=0$ and the total energy is conserved. Returning now to the underbraced term in (6.33) this constraint implies that
\begin{align}
\frac{d\mathbf{E}(\rho,\bm{u})}{dt}=\frac{d}{dt}\int\frac{1}{2}\rho|\bm{u}|^{2}d\mathrm{V}
+\int\frac{1}{2}
\sum_{j,k}\nabla_{j}[\rho u^{j}|u^{k}|^{2})d\mathrm{V}
+\int\sum_{k}u^{k}[\nabla_{k}p+\rho\nabla_{k}\Phi]d\mathrm{V}=0
\end{align}
Suppose the term involving the gradients of the pressure and gravitational potential vanishes then
\begin{align}
\sum_{k}(\nabla_{k}p+\rho\nabla_{k}\Phi)=0
\end{align}
Then one must
\begin{align}
\frac{d}{dt}\int\frac{1}{2}\rho|\bm{u}|^{2}d\mathrm{V}=-\int\frac{1}{2}
\sum_{j,k}\nabla_{j}[\rho u^{j}|u^{k}|^{2})d\mathrm{V}
\end{align}
or
\begin{align}
\frac{1}{2}\int(\dot{\rho}|u|^{2}+2\rho{u}_{k}\dot{{u}}^{k})d\mathrm{V}=-\int\frac{1}{2}
\sum_{j,k}\nabla_{j}[\rho u^{j}|u^{k}|^{2})d\mathrm{V}
\end{align}
This can only be satisfied for $\bm{u}=0$ and $\dot{\rho}=0$, that is for static equilibrium. In spherical symmetry the FEHE becomes $\partial_{r}p(r)+\rho\partial_{r}\Phi=0$ or
\begin{align}
\underbrace{\frac{dp(r)}{dr}=-\rho(r)\frac{d\Phi(r)}{dr}=-\frac{\mathscr{G}\rho(r)\mathcal{M}(r)}{r^{2}}}
\end{align}
which is again the FEHE.
\end{proof}
\begin{prop}
The conservation of total energy within $\mathbf{D}(t)$ is equivalent to the 1st variation with respect to the density so that
\begin{align}
\frac{d}{dt}{\mathbf{E}}(t)\delta t=\delta\mathbf{E}(t)
=\frac{\delta\mathbf{E}}{\delta \rho}\delta\rho=0
\end{align}
In full
\begin{align}
&\frac{d}{dt}\int_{\mathbf{R}^{3}}\left(\frac{1}{2}\rho|\bm{u}|^{2}+\frac{p}{|\gamma-1|}
+\frac{1}{2}\rho\Phi\right)d\mathrm{V}(\mathbf{x})=\delta\left|\int_{\mathbf{R}^{3}}\left(\frac{1}{2}\rho|\bm{u}|^{2}+\frac{p}{|\gamma-1|}
+\frac{1}{2}\rho\Phi\right)d\mathrm{V}(\mathbf{x})\right|\nonumber\\&
=\frac{d}{d\rho}\int_{\mathbf{R}^{3}}\left(\frac{1}{2}\rho|\bm{u}|^{2}+\frac{p}{|\gamma-1|}
+\frac{1}{2}\rho\Phi\right)d\mathrm{V}(\mathbf{x})\delta\rho=0
\end{align}
If $\mathbf{u}=0$ or $\mathbf{u}=const$ with $\delta_{i}\mathbf{u}=0$ then
\begin{align}
\delta\left|\int_{\mathbf{R}^{3}}\left(\frac{1}{2}\rho|\bm{u}|^{2}+\frac{p}{|\gamma-1|}
+\frac{1}{2}\rho\Phi\right)d\mathrm{V}(\mathbf{x})\right|
\end{align}
\begin{align}
\nabla_{i}p+\rho\nabla_{i}\Phi=0
\end{align}
which is the FE of hydrostatic equilibrium. Then in spherical symmetry we once
again recover the basic stellar structure equation
\begin{align}
\underbrace{\frac{d p(r)}{dr}=-\frac{\mathscr{G}\mathcal{M}(r)\rho(r)}{r^{2}}}
\end{align}
\end{prop}
\begin{proof}
\begin{align}
\delta\mathbf{E}&=\delta\left|\int_{\mathbf{R}^{3}}\left(\frac{1}{2}\rho|\bm{u}|^{2}+\frac{p}{|\gamma-1|}
+\frac{1}{2}\rho\Phi\right)d\mathrm{V}(\mathbf{x})\right|\nonumber\\&
=\frac{d}{dt}\left|\int_{\mathbf{R}^{3}}\left(\frac{1}{2}\rho|\bm{u}|^{2}+\frac{p}{|\gamma-1|}
+\frac{1}{2}\rho\Phi\right)d\mathrm{V}(\mathbf{x})\right|dt\nonumber\\&
=\frac{d}{dt}\left|\int_{\mathbf{R}^{3}}\left(\frac{1}{2}\rho|\bm{u}|^{2}+\frac{p}{|\gamma-1|}
+\frac{1}{2}\rho\Phi\right)d\mathrm{V}(\mathbf{x})\right|\frac{dt}{d\rho}\delta\rho\nonumber\\&
=\frac{d}{d\rho}\frac{d\rho}{dt}\left|\int_{\mathbf{R}^{3}}\left(\frac{1}{2}\rho|\bm{u}|^{2}+\frac{p}{|\gamma-1|}
+\frac{1}{2}\rho\Phi\right)d\mathrm{V}(\mathbf{x})\right|\frac{dt}{d\rho}\delta\rho\nonumber\\&
=\frac{d}{d\rho}\left|\int_{\mathbf{R}^{3}}\left(\frac{1}{2}\rho|\bm{u}|^{2}
+\frac{p}{|\gamma-1|}+\frac{1}{2}\rho\Phi\right)d\mathrm{V}(\mathbf{x})\right|\delta\rho=0
\end{align}
Now
\begin{align}
&\frac{d}{dt}\int_{\mathbf{D(t)}}\frac{p}{|\gamma-1|}d\mathrm{V}(\mathbf{x})dt
\equiv\frac{d}{d\rho}\int_{\mathbf{D(t)}}
\frac{p}{|\gamma-1|}d\mathrm{V}(\mathbf{x})d\rho=\int_{\mathbf{D(t)}}u_{i}
(\nabla_{i}p)d\mathrm{V}(\mathbf{x})\delta\rho\\&
\frac{d}{dt}\int_{\mathbf{D(t)}}\frac{1}{2}\rho\Phi
d\mathrm{V}(\mathbf{x})dt
\equiv\frac{d}{d\rho}\int_{\mathbf{D(t)}}
\frac{1}{2}\rho\Phi d\mathrm{V}(\mathbf{x})d\rho=\int_{\mathbf{D}}u_{i}\rho
(\nabla_{i}\Phi)d\mathrm{V}(\mathbf{x})\delta\rho
\end{align}
and since
\begin{align}
\frac{d}{d\rho}\frac{1}{2}\int_{\mathbf{R}^{3}}
\left(\frac{1}{2}\rho|\bm{u}|^{2}d\mathrm{V}(\mathbf{x})\right)
\delta\rho \equiv \frac{1}{2}\int_{\mathbf{R}^{3}}
\left(\frac{1}{2}\frac{d}{d\rho}
\rho|\bm{u}|^{2}\right)d\mathrm{V}(\mathbf{x})
\delta\rho
\end{align}
we have
\begin{align}
&\delta \mathrm{E}=\frac{1}{2}\int_{\mathbf{R}^{3}}\left(\frac{d}{d\rho}\big(\rho|u_{i}u^{i}|\big)
+u_{i}\big(\nabla_{i}p+\rho\nabla_{i}\Phi\big)\right)d\mathrm{V}(x)\delta\rho\nonumber\\&
=\frac{1}{2}\int_{\mathbf{R}^{3}}\left(u^{i}u_{i}+2\rho u_{i}\frac{du_{i}}{d\rho}+u_{i}\big(\nabla_{i}p+\rho\nabla_{i}\Phi\big)
\right)d\mathrm{V}(x)\delta\rho=0
\end{align}
assuming $du^{i}/d\rho\ne 0$. However, $\tfrac{du^{i}}{d\rho}=\tfrac{du^{i}}{dt}
\tfrac{dt}{d\rho}=0$ if $u^{i}=const$ or $\rho(\mathbf{x},t)=\rho(\mathbf{x})$. Then
\begin{align}
u^{i}u_{i}+2\rho u_{i}\frac{du_{i}}{d\rho}+u_{i}\big(\nabla_{i}p+\rho\nabla_{i}\Phi\big)=0
\end{align}
which is
\begin{align}
u^{i}+2\rho \frac{du_{i}}{d\rho}+\big(\nabla_{i}p+\rho\nabla_{i}\Phi\big)
\end{align}
For static or equilibrium configurations $u^{i}=0$ so that (6.76) reduces to
\begin{align}
\nabla_{i}p+\rho\nabla_{i}\Phi=0
\end{align}
which is again the fundamental equation of hydrostatic equilibrium. In spherical symmetry
\begin{align}
\underbrace{\frac{d p(r)}{dr}=-\rho(r)\frac{d\Phi(r)}{dr}=-\frac{\rho(r)\mathscr{G}\mathcal{M}(r)}{r^{2}}}
\end{align}
\end{proof}
As before, this can be reduced to the Lane-Emden ODE if one utilises a polytropic gas equation of state, and the results of Section 2 then apply. Of particular interest is the mass- critical case $\gamma=4/3$, discussed in Section and now reconsidered. A polytropic gas with $\gamma=4/3$ is essentially radiation. In supermassive stars, the total pressure is dominated by the radiation pressure contribution $\bm{[1,10]}$. The total pressure from a mixture of gas and radiation is $p=\mathscr{R}\rho\Theta+\frac{1}{3}a\Theta^{4}\sim\rho\Theta+\Theta^{4}$, where $\Theta$ is the temperature of the gas. When $\Theta\sim\rho^{1/2}$ then $p\sim\Theta^{4/3}$. Such stars are then teetering on the brink of catastrophic instability. The case $\gamma=4/3$ for equilibrium configurations has been shown to be nonlinearly unstable even though there are no growing modes in the linearised analysis $\bm{28}$. Special dynamical solutions can be constructed for the $EP_{4/3}$ system due to the self-similarity property. A compact fluid/gas sphere with $\gamma=4/3$ can expand or contract in a self-similar manner by cascading between different scales with $t\rightarrow\lambda(t)$., preserving the total mass. When $\gamma=4/3$ then $1/(2-\gamma)=\frac{3}{2}$. The self-similar expansion/collapse for $EP_{4/3}$ would then have the rates
\begin{align}
\overline{\lambda_{expand}(t)}\sim_{t\uparrow\infty} c_{1}t^{2/3},~~\overline{\lambda_{collapse}(t)}\sim_{t\uparrow\infty} c_{1}|T-t|^{2/3}
\end{align}
for some constants $c_{1},c_{2}>0$, where $ \overline{\lambda_{expand}(t)},\overline{\lambda_{collapse}(t)}$ represent the radii of the expanding or collapsing star. Solutions from the works (REFS) expand and collapse at a linear rate nd are no self-similar with respect to rescaling so that
\begin{align}
\overline{\lambda_{expand}(t)}\sim_{t\uparrow\infty} c_{1}t,~~\overline{\lambda_{collapse}(t)}
\sim_{t\uparrow\infty} c_{1}|T-t|
\end{align}
\subsection{Hydrostatic equilibrium via the virial tensor}
The fundamental equation of hydrostatic equilibrium can also be derived from the conservation of the virial tensor for a self-gravitating gas/fluid. The virial theorem for stars (Theorem 2.18) is also re-derived from this tensor. One first requires the following preliminary lemmas
\begin{lem}
The time derivative of the product of the density and the velocity is
\begin{align}
\partial_{t}(u^{k}\rho)=-\big[\nabla_{i}(\rho u^{i} u^{k})+\delta^{ik}\nabla_{i}p+\rho\nabla^{k}\Phi\big]
\end{align}
\end{lem}
\begin{proof}
The EP equations are
\begin{align}
&u^{k}\partial_{t}\rho+u^{k}\nabla_{i}\rho u^{i}+u^{k}(\nabla_{i}u^{i})\rho=0\\&
\rho\partial_{t}u^{k}+\rho u^{i}\nabla_{i}u^{k}+\delta^{ik}\nabla_{i}p=-
\rho\nabla^{k}\Phi
\end{align}
where the continuity equation has been multiplied through by $u^{i}$. Adding these then gives (6.85)
\end{proof}
\begin{defn}
Let $\mathbf{D}(t)_\subset\mathbf{R}^{3}$ support an $EP$ system with density $\rho=\rho(\bm{x},t)$, and velocity $u^{i}$ then the virial tensor is defined as
\begin{align}
\bm{\mathcal{V}}^{ij}(t)=\int_{\mathbf{D}}\rho(x,t)x^{i}u^{j}dV(x)
\end{align}
and $\bm{\mathcal{V}}(t)=\bm{\mathcal{V}}^{ij}(t)\delta_{ij}$ is the 'virial'.
\end{defn}
\begin{lem}(\underline{Derivative of the Newtonian potential})\newline
The fundamental solutions of the Laplace equation $\Delta\Gamma(x-y)=0$ are (REF-Evans)
\begin{align}
&\Gamma(x-y)=\Gamma(|\bm{x}-\bm{y}|)=\frac{1}{n(2-n)V(\mathbf{B}(0,1))}|\bm{x}-\bm{y}|^{2-n},~~n>2\\&
\Gamma(x-y)=\Gamma(|\bm{x}-\bm{y}|)=\frac{1}{2\pi}\log|\bm{x}-\bm{y}|,~~n=2
\end{align}
where $\Delta=\nabla_{i}\nabla^{i}=\partial^{2}/\partial x_{i}\partial x^{i}$. The following derivative estimates hold
\begin{align}
&\nabla_{i}\Gamma(\bm{x}-\bm{y})\le \frac{1}{nV(\mathbf{B}(0,1))}|\bm{x}-\bm{y}|^{1-n}\\&
\nabla_{i}\nabla_{j}\Gamma(\bm{x}-\bm{y})\le \frac{1}{V(\mathbf{B}(0,1))}|\bm{x}-\bm{y}|^{-n}
\end{align}
Let $\rho:\mathbf{R}^{3}\rightarrow\mathbf{R}$ be a density function with support in $\bm{\mathbf{D}}\in\mathbf{R}^{3}$ and let $\Phi(x)$ be the Newtonian potential
of $\rho$. Then $\Phi\in C^{1}(\mathbf{R})$ for all $x\in\bm{\mathbf{D}}\subset\mathbf{R}^{n}$. The derivatives of $\Phi(x)$ are then
\begin{align}
&\nabla_{i}\Phi(x)=\int_{\mathbf{D}}\nabla_{i}\Gamma(\bm{x}-\bm{y})\rho(y)d^{n}y\\&
\nabla_{i}\nabla_{j}\Phi(x)=\int_{\mathbf{D}}\nabla_{i}\nabla_{j}\Gamma(\bm{x}-\bm{y})\rho(y)d^{n}y
\end{align}
Now let $T\in[0,\infty),\rho\in C^{2}([0,T)\times \mathbf{R}^{3})$ so that the Newtonian potential $\Phi(\bm{x},t)$ at time t is
\begin{align}
\Phi(\bm{x},t)=-\frac{1}{4\pi}\int_{\mathbf{D}}\frac{\rho(y,t)}{|\bm{x}-\bm{y}|}d^{3}y,~x\in\mathbf{R}^{3}
\end{align}
The derivative is then $\bm{[24]}$,
\begin{align}
&\nabla_{i}\Phi(\bm{x},t)=-\frac{1}{4\pi}\int_{\mathbf{D}}\nabla_{i}\frac{\rho(y,t)}{|\bm{x}-\bm{y}|}d^{3}y
=\frac{1}{4\pi}\int_{\mathbf{D}(t)}\rho(\bm{y},t)\frac{(x^{j}-y^{j})}{|x-y|^{3}}dV(y)\\&\equiv \frac{1}{4\pi}\int_{\mathbf{D}(t)}\rho(\bm{y},t)\delta_{kj}\frac{(x^{k}-y^{k})}{|x-y|^{3}}dV(y)
\end{align}
\end{lem}
\begin{thm}(\underline{Fubin-Tonelli theorem})\newline
Let $f:\mathbf{R}^{3}\times\mathbf{R}\rightarrow\mathbf{R}$ be a multivariate function $f=f(x,y)$ for $x\in X\subset\mathbf{R}^{+}$ and
$y\in Y\subset\mathbf{R}^{+}$. Then
\begin{align}
\int_{X}\left(\int_{Y}f(x,y)d^{n}y\right)d^{n}x
=\int_{Y}\left(\int_{X}f(x,y)d^{n}x\right)d^{n}y
\end{align}
with convergence
\begin{align}
\int_{X}\left(\int_{Y}f(x,y)d^{n}y\right)d^{n}x<\infty
\end{align}
\end{thm}
The gravitational energy $\mathbf{E}_{G}$ of the EP system can also be expressed as an integral involving the gradient of the Newtonian potential. This lemma will be required when deriving the virial theorem for stars from the virial tensor $\mathbf{V}^{ij}$.
\begin{lem}
Given $\mathbf{D}(t)\in\mathbf{R}^{3}$ and the density $\rho(x,t)$ and Newtonian potential $\Phi(x,t)$ defined as before then
\begin{align}
\int_{\mathbf{D}(t)}\rho(x,t)\nabla_{j}\Phi(x,t)x^{j}dV(x)=-\frac{1}{2}\int_{\mathbf{D}(t)}\rho(x,t)\Phi(x,t)dV(x)\equiv -\mathbf{E}_{G}
\end{align}
\end{lem}
\begin{proof}
The proof $\bm{[24]}$ utilises the Fubini-Tonelli Theorem and the expression for the gradient $\nabla^{k}\Phi(x,t)$ and the fact that
$x^{i}=(x^{i}-y^{i})+y^{i}$. Then
\begin{align}
&\int_{\mathbf{D}(t)}\rho(x,t)\nabla_{j}\Phi(x,t)x^{j}dV(x)=\frac{1}{4\pi}\int\!\!\!\int_{\mathbf{D}(t)}\rho(\bm{x},t)\rho(\bm{y},t)x^{j}\delta_{kj}
\frac{(x^{k}-y^{k})}{|x-y|^{3}}dV(x)dV(y)\nonumber\\&
=\frac{1}{4\pi}\int\!\!\!\int_{\mathbf{D}(t)}\rho(\bm{x},t)\rho(\bm{y},t)\delta_{kj}
\frac{(x^{j}-y^{j})(x^{k}-y^{k})}{|x-y|^{3}}dV(x)dV(y)\nonumber\\&+\frac{1}{4\pi}
\int\!\!\!\int_{\mathbf{D}(t)}\rho(\bm{x},t)\rho(\bm{y},t)\delta_{kj}
\frac{y^{j}(x^{k}-y^{k})}{|x-y|^{3}}dV(x)dV(y)\nonumber\\&
=\frac{1}{4\pi}\int\!\!\!\int_{\mathbf{D}(t)}\rho(\bm{x},t)\rho(\bm{y},t)
\frac{(x^{j}-y^{j})(x^{j}-y^{j})}{|x-y|^{3}}dV(x)dV(y)\\&+\frac{1}{4\pi}
\int\!\!\!\int_{\mathbf{D}(t)}\rho({x},t)\rho(\bm{y},t)\delta_{kj}
\frac{y^{j}(x^{k}-y^{k})}{|x-y|^{3}}dV(x)dV(y)\nonumber\\&
=\int_{\mathbf{D}(t)}\rho(x,t)\left(\frac{1}{4\pi}\int_{\mathbf{D}(t)}\rho(y,t)\left(\frac{1}{|x-y|}\right)dV(y)\right)dV(x)\nonumber\\&
+\frac{1}{4\pi}\int\!\!\!\int_{\mathbf{D}(t)}\rho({x},t)\rho(\bm{y},t)\delta_{kj}
\frac{y^{j}(x^{k}-y^{k})}{|x-y|^{3}}dV(x)dV(y)\nonumber\\&
=\int_{\mathbf{D}(t)}\rho(x,t)\left(\frac{1}{4\pi}\int_{\mathbf{D}(t)}\rho(y,t)\left(\frac{1}{|x-y|}\right)dV(y)\right)dV(x)\nonumber\\&
-\int_{\mathbf{D}(t)}\rho(y,t)y^{j}\left(\frac{1}{4\pi}\int_{\mathbf{D}(t)}\rho(x,t)\frac{\delta_{kj}(y^{k}-x^{k})}{|x-y|^{3}}dV(x)\right)dV(y)\nonumber\\&
=-\int_{\mathbf{D}(t)}\rho(x,t)\Phi(x,t)dV(x)-\int_{\mathbf{D}(t)}\rho(x,t)y^{j}\nabla_{j}\Phi(y,t)dV(y)
\end{align}
Moving the last term on the RHS to the LHS and changing variables then gives
\begin{align}
2\int_{\mathbf{D}(t)}\rho(x,t)\nabla_{j}\Phi(x,t)x^{j}dV(x)=-\int_{\mathbf{D}(t)}\rho(x,t)\Phi(x,t)dV(x)\equiv -\mathbf{E}_{G}
\end{align}
from which (6.94) follows.
\end{proof}
The hydrostatic equilibrium equation for self-gravitating gaseous stars then follows from conservation of the virial tensor.
\begin{prop}
If
\begin{align}
\left|\frac{d\bm{\mathcal{V}}^{ij}(t)}{dt}\right|_{u=0}=\frac{d}{dt}\left|\int_{\mathbf{D}}
\rho(x,t)x^{i}u^{j}\delta_{ij}dV(x)
\right|_{u=0}=0
\end{align}
or
\begin{align}
\left|\frac{d\bm{\mathcal{V}}^{ij}(t)}{dt}\right|
+\int_{\mathbf{D}(t)}\partial_{t}(\rho u^{i}u^{k})\delta V(x)=0
\end{align}
then the fundamental hydrostatic equilibrium equation follows such that
\begin{align}
\underbrace{\nabla_{i}p+\rho\nabla_{i}\Phi=0}
\end{align}
\end{prop}
\begin{proof}
The time derivative of the virial tensor is
\begin{align}
&\left|\frac{d\bm{\mathcal{V}}^{ij}(t)}{dt}\right|=\int_{\mathbf{D}}\partial_{t}\rho u^{i})x^{j}\delta_{ij}d^{3}x\nonumber\\&=\int_{\mathbf{D}(t)}\big[-\nabla_{i}(\rho u^{i}u^{k})-\delta^{ik}\nabla_{i}p-\rho\nabla^{k}\Phi\big]x^{j}\delta_{kj}dV(x)
\delta_{ij}d^{3}x
\end{align}
using (6.77). Then if $u=0$ and $\left|\frac{d\bm{\mathcal{V}}^{ij}(t)}{dt}\right|=0$
\begin{align}
&\left|\frac{d\bm{\mathcal{V}}^{ij}(t)}{dt}\right|=\int_{\mathbf{D}(t)}
\big[-\delta^{ik}\nabla_{i}p-\rho\nabla^{k}\Phi\big]x^{j}
\delta_{ik}d^{3}x=0
\end{align}
so that once again the following fundamental result is inevitable
\begin{align}
\underbrace{\nabla^{k}p+\rho\nabla^{k}\Phi=0}
\end{align}
Theorem (2.8) gave the Virial Theorem for self-gravitating stars. The same result now follows from the conservation of the virial tensor via an integration by parts and utilisation of Lemma (6.14). Taking $u=0$ as before for static equilibrium gives
\begin{align}
&\left|\frac{d\bm{\mathcal{V}}^{ij}(t)}{dt}\right|=\int_{\mathbf{D}}\partial_{t}\rho u^{i})x^{j}\delta_{ij}dV(x)\nonumber\\&=\int_{\mathbf{D}(t)}\big[-\delta^{ik}\nabla_{i}p-\rho\nabla^{k}\Phi\big]x^{j}\delta_{kj}dV(x)\nonumber\\&
=-\int_{\mathbf{D}(t)}[\delta^{ik}\nabla_{i}px^{j}\delta_{kj}dV(x)
-\int_{\mathbf{D}(t)}\rho\nabla^{k}\Phi x^{j}\delta_{kj}dV(x)\nonumber\\&=-\int_{\mathbf{D}(t)}\underbrace{[\nabla^{k}px_{k}
dV(x)}_{int~by~parts}-\int_{\mathbf{D}(t)}\rho\nabla^{k}\Phi x^{j}\delta_{kj}dV(x)\nonumber\\&
=\int_{\mathbf{D}}p(\nabla_{k}x^{k})dV(x)-\underbrace{\int_{\mathbf{D}(t)}\rho\nabla^{k}\Phi x^{j}\delta_{kj}d^{3}}_{use~Lem~6.14}\nonumber\\&
=3\int_{\mathbf{D}}p(x)d^{3}x+\frac{1}{2}\int_{\mathbf{D}(t)}\rho\phi(x)dV(x)= 3\int_{\mathbf{D}}p(x)d^{3}x+\mathbf{E}_{G}=0
\end{align}
which is exactly equivalent to the stellar Virial Theorem (2.8).
\end{proof}
\subsection{Energy integrals for $EP_{\gamma}$ systems in static equilibrium in all dimensions}
Requires the following two preliminary lemmas.
\begin{lem}
In a domain of n-dimensional Euclidean space, ${\mathbf{D}}\subset\mathbf{R}^{n}$ the Poisson equation is
\begin{align}
\Delta\Phi(\mathbf{x})=\mathcal{Q}(n)\mathscr{G}\rho
\end{align}
The solution is
\begin{align}
\Phi(\mathbf{x})=\int_{\mathbf{D}}\Gamma(\mathbf{x}-\mathbf{y})\rho(\mathbf{y},t)d^{n}\mathbf{y}
\end{align}
where $\mathscr{G}(\mathbf{x}-\mathbf{y})$ is the Greens function or fundamental solution.
\begin{enumerate}
\item If $n=1$ then
\begin{align}
\Gamma(\mathbf{x}-\mathbf{y})=|\mathbf{x}-\mathbf{y}|
\end{align}
\item If $n=2$ then
\begin{align}
\Gamma(\mathbf{x}-\mathbf{y})=\log|\mathbf{x}-\mathbf{y}|
\end{align}
\item If $n\ge 3$ then
\begin{align}
\Gamma(\mathbf{x}-\mathbf{y})=-\frac{1}{|\mathbf{x}-\mathbf{y}|^{n-2}}
\end{align}
\end{enumerate}
\end{lem}
Since the pressure $p(x)=0$ for all $\bm{x}\in\mathbf{D}$ then the following useful integral identity can be derived.
\begin{lem}
Let $\mathbf{D}\subset\mathbf{R}^{n}$ and let $\partial\mathbf{D}$ be the boundary. The domain supports a fluid/gas obeying the Euler-Poisson equations with pressure $p(\mathbf{x})$ or $p(\mathbf{x},t)$, although here only the spatial dependence is relevant. At the boundary $p(\mathbf{x})=0$ for all $\bm{x}\in\partial\mathbf{D}$. The gradient of the pressure is $\nabla p$. If $p(\mathbf{x})=0$ for all $\mathbf{x}\in\partial\mathbf{D}$ then
\begin{align}
n\int_{\mathbf{D}}p(\mathbf{x})d^{n}\mathbf{x}=-\int_{\mathbf{D}}\mathbf{x}.\nabla p(\mathbf{x})d^{n}x
\end{align}
\end{lem}
\begin{proof}
Let $V(\mathbf{x})$ be any vector field or vector in $\mathbf{D}\bigcup\partial\mathbf{D}$ and let $p(x)$ be the pressure. Then the basic identity for the divergence gives
\begin{align}
\bm{\nabla}.(fV)=f(\nabla.V)+V.\bm{\nabla} f\equiv\sum_{i=1}^{n}\frac{\partial V_{i}}{\partial x^{i}}+\sum_{i=1}^{n}V^{i}\frac{\partial}{\partial x_{i}}p(x)
\end{align}
If $V(\mathbf{x})=\mathbf{x}=(x_{1},x_{2},x_{3})$ then
\begin{align}
\nabla.(p(\mathbf{x})\mathbf{x})=p(\mathbf{x})(\nabla.\mathbf{x})+\mathbf{x}.(\nabla f)=p(x)\sum_{i=1}^{n}\frac{\partial x_{i}}{\partial x_{i}}+\sum_{i=1}^{n}
x_{i}\nabla^{i}p(\mathbf{x})=np(\mathbf{x})+\mathbf{x}.\bm{\nabla} p(\mathbf{x})
\end{align}
If $p(x)=0$ on $\partial\mathbf{D}$ then
\begin{align}
\int_{\partial\mathbf{D}}\mathbf{x} p(\mathbf{x})d^{n-1}x=0
\end{align}
Applying the divergence theorem and then (6.118) to this integral
\begin{align}
\int_{\partial\mathbf{D}}\mathbf{x}.p(\mathbf{x})d^{n-1}\mathbf{x}=
\bm{\nabla}.(\mathbf{x}p(\mathbf{x})d^{n}\mathbf{x}=n\int_{\mathbf{D}} p(\mathbf{x})d^{n}\mathbf{x}+\int_{\mathbf{D}}\mathbf{x}.\bm{\nabla} f(\mathbf{x})d^{n}\mathbf{x}=0
\end{align}
so that
\begin{align}
n\int_{\mathbf{D}}p(\mathbf{x})d^{n}\mathbf{x}=-\int_{\mathbf{D}}\mathbf{x}.\nabla p(\mathbf{x}) d^{n}\mathbf{x}
\end{align}
as required.
\end{proof}
The following Lemma then applies the previous lemmas to establish the stationary energy integrals for an $EP_{\gamma}$ system in dimensions $n=2$ and $n\ge3$. solutions to EP system for gaseous stars $\bm{[35]}$
\begin{lem}
Let $EP_{\gamma}$ be an Euler-Poisson system for a polytropic gas with $p=\mathsf{K}|\rho|^{\gamma}$ and with support ${\mathbf{D}}\subset\mathbf{R}^{3}$. On the boundary $p(x)=\rho(x)=0$ for $x\in\partial\mathbf{D}$. As before, the total energy is
\begin{align}
&\mathbf{E}(t)=\mathbf{E}_{K}(t)+\mathbf{E}(t)= \int_{\mathbf{D}}\left(\frac{1}{2}\rho|u|^{2}dV(x)+\frac{p }{|\gamma-1|}+\frac{1}{2}\rho\Phi\right)d^{n}x,~\gamma>1\\&
\mathbf{E}(t)=\mathbf{E}_{K}(t)+\mathbf{E}(t)=\int_{\mathbf{D}}\left(\frac{1}{2}\rho|u|^{2}dV(x)+\mathsf{K}\rho(x)\log\rho(x)+
\frac{1}{2}\rho(x)\Phi(x)\right) dV(x),~\gamma=1
\end{align}
and the total mass is $M=\int_{\mathbf{D}}\rho(x,t)d^{n}x$. The stationary energies are for $u=0$ when $\mathbf{E}_{K}=0$ so that $\mathbf{E}(t)=\mathbf{E}(t)$ and the system is in hydrostatic equilibrium. Then when $\nabla_{i}p+\rho\nabla_{i}\Phi=0$ the stationary energies are given by the following integral formulas.
\begin{enumerate}
\item For $n=2$ or $\mathbf{D}\subset\mathbf{R}^{2}$
\begin{align}
\mathbf{E}(t)=\frac{\mathscr{G}M^{2}}{4|\gamma-1|}+\frac{1}{2} \int\!\!\!\int\rho(x)\rho(y)\log|x-y|
d^{2}xd^{2}y
\end{align}
\item For $n\ge3$ or $\mathbf{D}\subset\mathbf{R}^{n}$
\begin{align}
\mathbf{E}(t)=\frac{2(n-1)-n\gamma}{(n-2)(\gamma-1)}\int_{\mathbf{D}}p(x)dV(x)
\end{align}
\end{enumerate}
\end{lem}
\begin{proof}
To prove (6.118), apply equation (6.115) with $n=2$ so that
\begin{align}
&2\int_{\mathbf{D}}p(x)d^{2}x=-\int_{\mathbf{D}}x.\nabla p(x)d^{2}x
=\int_{\mathbf{D}}x\rho(x)\nabla\Phi(x)d^{2}x\nonumber\\&=\mathscr{G}\int_{\mathbf{D}}\rho(x)\int_{\mathbf{D}}
\nabla_{x}\log|x-y|\rho(y)d^{2}x d^{2}y\nonumber\\&
=\mathscr{G}\int_{\mathbf{D}}\rho(x)\int_{\mathbf{D}}\frac{\rho(y)(x-y).x}{|x-y|^{2}}d^{2}xd^{2}y
\equiv\mathscr{G}\int\!\!\!\int_{\mathbf{D}}\frac{\rho(x)\rho(y)(x-y)^{2}.x}{|x-y|^{2}}d^{2}x d^{2}y=\mathscr{G}\mathbf{I\!H}
\end{align}
The integral $\mathbf{I\!H}$ can be written as
\begin{align}
&\mathbf{I\!H}=\int\!\!\!\int_{\mathbf{D}}\frac{\rho(x)\rho(y)(x-y)^{2}.x}{|x-y|^{2}}d^{2}x d^{2}y\nonumber\\&=
\int\!\!\!\int_{\mathbf{D}}\frac{\rho(x)\rho(y)(x-y)(x-y)}{|x-y|^{2}}d^{2}x d^{2}y+
\int\!\!\!\int_{\mathbf{D}}\frac{\rho(x)\rho(y)(x-y)^{2}.y}{|x-y|^{2}}d^{2}y d^{2}y\nonumber\\&
=\int\!\!\!\int_{\mathbf{D}}\rho(x)\rho(y)d^{2}xd^{2}y+\int\!\!\!\int_{\mathbf{D}}
\frac{\rho(x)\rho(y)(x-y)^{2}.y}{|x-y|^{2}}d^{2}y d^{2}y\nonumber\\&
=\int_{\mathbf{D}}\rho(x)d^{2}x\int_{\mathbf{D}}\rho(y)d^{2}xd^{2}y-\mathbf{I\!H}=M^{2}-\mathbf{I\!H}
\end{align}
so that $\mathbf{I\!H}=\frac{1}{2}M^{2}$. Then the pressure integral becomes
\begin{align}
\int_{\mathbf{D}}p(x)d^{n}x=\frac{1}{4}\mathscr{G}M^{2}
\end{align}
The stationary energy is then
\begin{align}
\mathbf{E}(t)&=\int_{\mathbf{D}}\left(\frac{p}{|\gamma-1|}+\frac{1}{2}\rho\Phi\right)d^{n}x
=\frac{1}{4}\frac{\mathscr{G}M^{2}}{|\gamma-1|}
+\frac{1}{2}\int_{\mathbf{D}}\rho(x)\Phi(x) d^{2}x
\nonumber\\&
=\frac{1}{4}\frac{\mathscr{G}M^{2}}{|\gamma-1|}
+\frac{1}{2}\int\!\!\!\int_{\mathbf{D}}\rho(x)\rho(y)\log|x-y|d^{2}xd^{2}y
\end{align}
since
\begin{align}
\Phi(x)=\int_{\mathbf{D}}\rho(y)\log|x-y|d^{3}y
\end{align}
To prove (6.119) for $n\ge 3$ we must first establish the identity
\begin{align}
n\int_{\mathbf{D}}\mathsf{K}|\rho|^{\gamma}d^{n}x+
n\int_{\mathbf{D}}p(x)d^{n}x=-\frac{1}{2}(n-2)\int_{\mathbf{D}}\rho(x)\Phi(x)d^{n}x
\end{align}
First evaluate $\rho(x)\nabla\Phi(x)$ taking the derivative of the
gravitational potential for $n\ge 3$ giving
\begin{align}
\rho(x)\nabla \Phi(x)=-\mathscr{G}\rho(x)\int_{\mathbf{D}}
\nabla_{x}\frac{1}{|x-y|^{n-2}}\rho(y)d^{n}y=\mathscr{G}\rho(x)\int_{\Omega}
\frac{(n-2)(x-y)\rho(y)}{|x-y|^{n}}
\end{align}
Multiply by $x$ and integrate over $\bm{\Omega}$ so that
\begin{align}
\int_{\mathbf{D}}x.\rho(x)\nabla\Phi(x)d^{n}x
=\mathscr{G}(n-2)\int\!\!\!\int_{\mathbf{D}}\frac{\rho(x)\rho(y)(x-y).x}{|x-y|^{n}}
\end{align}
However, from the FEHE $\nabla p(x)=-\rho\nabla \Phi$ so
\begin{align}
\int_{\mathbf{D}}x.\rho(x)\nabla\Phi(x)d^{n}x
=-\int_{\mathbf{D}}x.\nabla p(x) d^{n}x=\int_{\mathbf{D}}p(x)d^{n}x
\end{align}
and so
\begin{align}
\int_{\mathbf{D}}p(x)d^{n}x=\mathscr{G}(n-2)\int\!\!\!\int_{\mathbf{D}}\frac{\rho(x)
\rho(y)(x-y).x}{|x-y|^{n}}d^{2}x d^{2}y=\mathscr{G}\mathbf{I\!H}
\end{align}
The integral $\mathbf{I\!H}$ can be manipulated as
\begin{align}
\mathbf{I\!H}&=(n-2)\int\!\!\!\int_{\mathbf{D}}\frac{\rho(x)
\rho(y)(x-y).x}{|x-y|^{n}}d^{2}x d^{2}y\\&
=(n-2)\int\!\!\!\int_{\mathbf{D}}\frac{\rho(x)
\rho(y)(x-y).(x-y)}{|x-y|^{n}}d^{2}x d^{2}y+(n-2)\int\!\!\!\int_{\mathbf{D}}\frac{\rho(x)
\rho(y)(x-y).y}{|x-y|^{n}}d^{2}x d^{2}y\nonumber\\&
=(n-2)\int\!\!\!\int_{\mathbf{D}}\frac{\rho(x)
\rho(y)}{|x-y|^{n-2}}d^{2}x d^{2}y+(n-2)\int\!\!\!\int_{\mathbf{D}}\frac{\rho(x)
\rho(y)(x-y).y}{|x-y|^{n}}d^{2}x d^{2}y\nonumber\\&=
(n-2)\int=\!\!\!\int_{\mathbf{D}}\frac{\rho(x)
\rho(y)}{|x-y|^{n-2}}d^{2}x d^{2}y-
\mathbf{I\!H}
\end{align}
so that
\begin{align}
\mathbf{I\!H}=\frac{1}{2}(n-2)\int\!\!\!\int_{\mathbf{D}}\frac{\rho(x)
\rho(y)d^{n}xd^{n}y}{|x-y|^{n-2}}
\end{align}
The pressure integral becomes
\begin{align}
\int_{\mathbf{D}}p(x)d^{n}x&=\frac{1}{2}(n-2)\mathscr{G}
\int\!\!\!\int_{\mathbf{D}}\frac{\rho(x)\rho(y)}{|x-y|^{n-2}}d^{n}x d^{n}y\nonumber\\&\equiv -\frac{1}{2}(n-2)\int_{\mathbf{D}}\rho(x)\Phi(x)d^{n}x
\end{align}
The total stationary energy is then
\begin{align}
\mathbf{E}&=\int_{\mathbf{D}}\frac{p(x)}{|\gamma-1|}d^{n}x+\frac{1}{2}\int_{\mathbf{D}}\rho(x)\Phi(x)d^{n}x\nonumber\\&=
\int_{\mathbf{D}}\frac{p(x)}{|\gamma-1|}d^{n}x-\frac{n}{n-2}\int_{\mathbf{D}}p(x)d^{n}x\nonumber\\&
=\frac{2(n-1)-n\gamma}{|\gamma-1|(n-2)}\int_{\mathbf{D}}p(x)d^{n}x
\end{align}
\end{proof}
\begin{cor}
For $n=3$ the stationary energy is
\begin{align}
\mathbf{E}=\frac{3(\frac{4}{3}-\gamma)}{|\gamma-1|}\int_{\mathbf{D}}p(x)d^{3}x
\end{align}
so that $\mathbf{E}=0$ for $\gamma=4/3$. This once again confirms that a polytropic gaseous star with $\gamma=4/3$ is teetering on instability.
\end{cor}
\subsection{Bounds on the pressure integral}
The following theorem $\bm{[36]}$ establishes bounds on the pressure integral $\int_{\mathbf{R}^{n}}p(x)d^{n}x$. If the star collapsed to zero size then the pressure integral would blow up. Conversely, if the pressure integral is constant, finite and bounded then the star is stable and in equilibrium. From (6.135) $z\mathbf{E}<0$ if $\gamma>4/3$ and $E\ge0$ if $\gamma\le 4/3$ so the star is stable for $\gamma >4/3$. Using the Hardy-Littlewood-Paley inequality once can derive a finite bound on the pressure integral for $\gamma>2(n-1)/n$.

Some preliminary definitions and lemmas are first required.
\begin{defn}
If $f\in L^{1}(\mathbf{R}^{n})$ then the Fourier transform $(\mathscr{F},f)(\xi)$ and its inverse is defined as
\begin{align}
&\widehat{f(\xi)}=\int_{\mathbf{R}^{n}}f(x)e^{-ix\xi}d^{n}x\\&
\widehat{f(x)}=\int_{\mathbf{R}^{n}}f(\xi)e^{ix\xi}d^{n}\xi
\end{align}
with derivatives $(\nabla^{\alpha}f)(\xi)=(i\xi)^{\alpha}f(\xi)$ and $\nabla^{2}\widehat{f(\xi)}=-\xi^{2}\widehat{f(\xi)}$.
\end{defn}
\begin{lem}
If $f\in L^{2}(\mathbf{R}^{n})$ then Plancheral's Thm states that $\|\mathscr{F} f(\bullet)\|_{L_{2}}=\|f(\bullet)\|_{L_{2}}$ or $\int_{\mathbf{R}^{n}}|f(x)|^{2}d^{n}x=
\int_{\mathbf{R}^{n}}|f(\xi)|^{2}d^{\xi}$.
\end{lem}
\begin{lem}
The Fourier transforms of the density $\rho(x)$ and the gravitational potential $\phi(x)$ and their inverses are
\begin{align}
&\widehat{\rho(\xi)}=\int_{\mathbf{R}^{n}}\rho(x)e^{-ix\xi}d^{n}x\\&
\widehat{\rho(x)}=\int_{\mathbf{R}^{n}}\rho(\xi)e^{ix\xi}d^{n}\xi\\&
\widehat{\Phi(\xi)}=\int_{\mathbf{R}^{n}}\Phi(x)e^{-ix\xi}d^{n}x\\&
\widehat{\Phi(x)}=\int_{\mathbf{R}^{n}}\Phi(\xi)e^{ix\xi}d^{n}\xi
\end{align}
so that
\begin{align}
\int_{\mathbf{R}^{n}}\rho(x)\Phi(x)d^{n}x=\int_{\mathbf{R}^{n}}\widehat{\rho(\xi)}
\widehat{\Phi(\xi)}d^{n}\xi
\end{align}
or
\begin{align}
\|\rho(\bullet)\Phi(\bullet)\|_{L_{1}}=\|\widehat{\rho(\bullet)}\widehat{\Phi(\bullet)}\|_{L_{1}}
\end{align}
\end{lem}
\begin{proof}
Since $\Phi(x)$ is real then $\Phi^{*}(x)=\Phi(x)$. The integral (6.142) becomes
\begin{align}
&\int_{\mathbf{R}^{n}}\rho(x)\Phi(x)d^{n}x
=\int_{\mathbf{R}^{n}}\rho(x)\Phi^{*}(x)d^{n}x\nonumber\\&=\int_{\mathbf{R}^{n}}
\left|\int\!\!int_{\mathbf{R}^{n}}\rho(\xi)\Phi^{*}(\eta)e^{ix(\xi-\eta)}d^{n}\xi d^{n}\eta
\right|d^{n}x\nonumber\\&=
\int\!\!\!\int_{\mathbf{R}^{n}}\rho(\xi)\Phi^{*}(\eta)\delta^{n}(\xi-\eta)d^{n}\xi d^{n}
\nonumber\\&=\int_{\mathbf{R}^{n}}\rho(\xi)\Phi^{*}(\xi)d^{n}\xi =\int_{\mathbf{R}^{n}}\rho(\xi)\Phi(\xi)d^{n}\xi
\end{align}
One an also use the Cauchy Schwartz inequality, taking the equality, and the Plancharel Theorem so that
\begin{align}
\int_{\mathbf{R}^{n}}\rho(x)\Phi(x)d^{n}x&=\left( \int_{\mathbf{R}^{n}}|\rho(x)|^{2}d^{n}x\right)^{1/2}\left( \int_{\mathbf{R}^{n}}|\Phi(x)|^{2}d^{n}x\right)^{1/2}\nonumber\\&
=\left( \int_{\mathbf{R}^{n}}|\rho(\xi)|^{2}d^{n}\xi\right)^{1/2}\left( \int_{\mathbf{R}^{n}}|\Phi(\xi)|^{2}d^{n}\xi\right)^{1/2}\nonumber\\&
= \int_{\mathbf{R}^{n}}\rho(\xi)\Phi(\xi)d^{n}\xi
\end{align}
\end{proof}
The final preliminary lemma is the Hardy Littlewood-Paley inequality.
\begin{lem}
Let $1<\gamma\le 2$ and $f\in L^{\gamma}(\mathbf{R}^{n})$ then $\exists$ a
Marcinkiewicz interpolation constant $\bm{\mathfrak{M}}_{\gamma}$ such that
\begin{align}
\left(\int_{\mathbf{R}^{n}}|\widehat{f(\xi)}|^{\gamma}|\gamma|^{n(\gamma-2)}d^{n}\xi\right)^{\frac{1}{\gamma}}
\le \bm{\mathfrak{M}}_{\gamma}\|f(\bullet)\|_{\gamma}\equiv
\bm{\mathfrak{M}}_{\gamma}\int_{\mathbf{R}^{n}}|f(x)|^{\gamma}d^{n}x
\end{align}
If $f(\xi)=\rho(\xi)$ then
\begin{align}
\left(\int_{\mathbf{R}^{n}}|\widehat{\rho(\xi)}|^{\gamma}|\gamma|^{n(\gamma-2)}d^{n}
\xi\right)^{\frac{1}{\gamma}}\le \bm{\mathfrak{M}}_{\gamma}\|\rho(\bullet)\|_{\gamma}\equiv
\bm{\mathfrak{M}}_{\gamma}\int_{\mathbf{R}^{n}}|\rho(x)|^{\gamma}d^{n}x
\end{align}
\end{lem}
The main theorem on the pressure integral bound is then as follows $\bm{[36]}$.
\begin{thm}
Let $p=\mathsf{K}\rho^{\gamma}$ be a polytropic gas equation of state and let $(p,\rho,u)$ be a solution of an $EP_{\gamma}$ system with support $\Omega$ with mass $M$ and total energy $\mathrm{E}$.
\begin{enumerate}
\item Then for $\gamma\ge 2(n-1)/n$ there exists constants
$C_{1}(M,E,\gamma)$ and $C_{2}(M,E,\gamma)$ such that
\begin{align}
\int_{\mathbf{R}^{n}}p(x)d^{n}x\le C_{1}(M,E,\gamma)=2(\gamma-1)\left[E-\frac{1}{2}n^{2}(n-1)
V(\mathbf{B}_{n}(1))\mathscr{G} M^{2}\gamma^{n-2}\right]<\infty
\end{align}
where $V(\mathbf{B}_{n}(1))$ is the volume of the unit ball.
\item For $2(n-1)/n<\gamma<2$
\begin{align}
&\int_{\mathbf{R}^{n}}p(x)d^{n}x\le C_{2}(M,E,\gamma)\le \frac{\mathsf{K}}{|\gamma-1|}
\int_{\mathbf{R}^{n}}|\rho(x)|^{\gamma}d^{n}x
+\frac{1}{\lambda^{2}}n(n-2)\mathscr{G}|\mathrm{V}(\mathbf{B}_{n}(1)|
\left|\frac{|\gamma-2}{\gamma-1}\right|M\nonumber\\&+n^{2}(n-2)|\mathrm{V}(\mathbf{B}_{n}(1)|^{2}
\mathscr{G}M^{2}\lambda^{n-2}=C_{2}(M,E,\gamma)<\infty
\end{align}
\end{enumerate}
\end{thm}
\begin{proof}
We consider the case $2(n-1)/n<\gamma<2$ then the case $\gamma>2$. The following identities will be useful in establishing the proof.
\begin{align}
&\xi^{2}=\xi^{2(n-\gamma)}\xi^{2+n(\gamma-2)}\\&
\xi^{2}=|\xi|^{n-1}|\xi|^{3-n}\\&
\frac{1}{|\xi|^{n-1}}\le \frac{n}{\lambda^{n-1}},\lambda>0\\&
\int_{\xi\le |\lambda|}d^{n}\xi=|V(\mathbf{B}_{n}(1)|\lambda^{n}\\&
|\widehat{\rho(\xi)}|\le \|\rho(\bullet)\|_{L_{1}}=M=\int_{\mathbf{R}^{n}}\rho(x)d^{n}x
\end{align}
\begin{enumerate}
\item The case $2(n-1)/n<\gamma<2$. First take the Fourier transform of the Poisson equation in n dimensions so that
\begin{align}
\mathscr{F}(\Delta\Phi(x))=n(n-2)|V(\mathbf{B}_{n}(1))|\mathscr{G}\rho(x)
\end{align}
which is
\begin{align}
\Delta\widehat{\Phi(\xi)}=-\xi^{2}\Phi(\xi)=n(n-2)|V(\mathbf{B}_{n}(1)|\mathscr{G}
\widehat{\rho(\xi)}
\end{align}
so that
\begin{align}
\widehat{\Phi(\xi)}=-n(n-2)|V(\mathbf{B}_{n}(1)|\mathscr{G}
\frac{\widehat{\rho(\xi)}}{|\xi|^{2}}
\end{align}
Multiplying by $\widehat{\rho(\xi)}$ and integrating gives
\begin{align}
\int_{\mathbf{R}^{n}}\rho(\xi)\Phi(\xi)d^{n}\xi=-n(n-2)|V(\mathbf{B}_{n}(1)|\mathscr{G}
\int_{\mathbf{R}^{n}}\frac{\widehat{\rho(\xi)}}{|\xi|^{2}}d^{n}\xi
\end{align}
Then
\begin{align}
&-\int_{\mathbf{R}^{n}}\rho(x)\Phi(x)d^{n}x=n(n-2)|V(\mathbf{B}_{n}(1)|\mathscr{G}
\int_{\mathbf{R}^{n}}\frac{|\widehat{\rho(\xi)}|^{2}}{|\xi|^{2}}d^{n}\xi\nonumber\\&=
n(n-2)|V(\mathbf{B}_{n}(1)|\mathscr{G}
\int_{|\xi|>\lambda}\frac{|\widehat{\rho(\xi)}|^{2}}{|\xi|^{2}}d^{n}\xi\nonumber\\&+n(n-2)|V(\mathbf{B}_{n}(1)|
\mathscr{G}
\int_{|\xi|\le \lambda}\frac{|\widehat{\rho(\xi)}|^{2}}{|\xi|^{2}}d^{n}\xi\nonumber\\&=
n(n-2)|V(\mathbf{B}_{n}(1)|\mathscr{G}
\int_{|\xi|>\lambda}\frac{|\widehat{\rho(\xi)}|^{2-\gamma}|\rho(\xi)|^{\gamma}}{|\xi|^{2}}d^{n}\xi
\nonumber\\&+n(n-2)|V(\mathbf{B}_{n}(1)|\mathscr{G}
\int_{|\xi|\le \lambda}\frac{|\widehat{\rho(\xi)}|^{2}}{|\xi|^{2}}d^{n}\xi\nonumber\\&\le
n(n-2)|V(\mathbf{B}_{n}(1)|\mathscr{G}\int_{|\xi|>\lambda}
\frac{\|\widehat{\rho(\bullet)}\|_{L_{1}}^{2-\gamma}|\rho(\xi)|^{\gamma}}{|\xi|^{2}}d^{n}\xi\nonumber\\&-n(n-2)|V(\mathbf{B}_{n}(1)|
\mathscr{G}\int_{|\xi|\le \lambda}\frac{\|\widehat{\rho(\bullet)}\|_{L_{1}}^{2}}{|\xi|^{2}}d^{n}\xi\nonumber\\&\le
-n(n-2)|V(\mathbf{B}_{n}(1)M^{2-\gamma}|\mathscr{G}
\int_{|\xi|>\lambda}\frac{\rho(\xi)|^{\gamma}}{|\xi|^{2}}d^{n}\xi\nonumber\\&
-n(n-2)|V(\mathbf{B}_{n}(1)|\mathscr{G}
\int_{|\xi|\le \lambda}\frac{\|\widehat{\rho(\bullet)}\|^{2}}{|\xi|^{2}}d^{n}\xi\nonumber\\&\le
-n(n-2)|V(\mathbf{B}_{n}(1)M^{2-\gamma}|\mathscr{G}
\int_{|\xi|>\lambda}\frac{\rho(\xi)|^{\gamma}}{|\xi|^{2}}d^{n}\xi\nonumber\\&-n(n-2)
|V(\mathbf{B}_{n}(1)|\mathscr{G}
M^{2}\int_{|\xi|\le \lambda}\frac{d^{n}\xi }{|\xi|^{2}}d^{n}\xi\nonumber\\&=
-n(n-2)|V(\mathbf{B}_{n}(1)M^{2-\gamma}|\mathscr{G}
\int_{|\xi|>\lambda}\frac{\rho(\xi)|^{\gamma}}{|\xi|^{n(2-\gamma})\xi^{2+n(\gamma-2)}}
d^{n}\xi\nonumber\\&-n(n-2)|V(\mathbf{B}_{n}(1)|\mathscr{G}
M^{2}\int_{|\xi|\le \lambda}\frac{d^{n}\xi }{|\xi|^{n-1}|\xi|^{3-n}}d^{n}\xi\nonumber\\&\le
-\frac{n(n-2)|V(\mathbf{B}_{n}(1)M^{2-\gamma}|\mathscr{G}}{\lambda^{2+n(\gamma-2)}}
\underbrace{\int_{|\xi|>\lambda}\frac{|\rho(\xi)|^{\gamma}}{|\xi|^{n(2-\gamma})}
d^{n}}_{use~HLP~ineq.}\xi\nonumber\\&-\frac{n(n-2)|V(\mathbf{B}_{n}(1)|\mathscr{G}
M^{2}}{\lambda^{3-n}}\int_{|\xi|\le \lambda}\frac{d^{n}\xi }{|\xi|^{n-1}}d^{n}\xi\nonumber\\&\le
-\frac{n(n-2)|V(\mathbf{B}_{n}(1)M^{2-\gamma}|\mathscr{G}}{\lambda^{2+n(\gamma-2)}}
|\bm{\mathfrak{M}}|^{\gamma}\|\rho(\bullet)\|_{L_{\gamma}}^{\gamma}
\nonumber\\&-\frac{n(n-2)|V(\mathbf{B}_{n}(1)|\mathscr{G}
M^{2}}{\lambda^{3-n}}\int_{|\xi|\le \lambda}\frac{d^{n}\xi }{|\xi|^{n-1}}d^{n}\xi\nonumber\\&\le
-\frac{n(n-2)|V(\mathbf{B}_{n}(1)M^{2-\gamma}|\mathscr{G}}{\lambda^{2+n(\gamma-2)}}
|\bm{\mathfrak{M}}|^{\gamma}\int_{\mathbf{R}^{n}}|\rho(x)|^{\gamma}d^{n}x
\nonumber\\&-\frac{n(n-2)|V(\mathbf{B}_{n}(1)|\mathscr{G}
M^{2}}{\lambda^{3-n}}\int_{|\xi|\le \lambda}\frac{d^{n}\xi }{|\xi|^{n-1}}d^{n}\xi\nonumber\\&\le
-\frac{n(n-2)|V(\mathbf{B}_{n}(1)M^{2-\gamma}|\mathscr{G}}{\lambda^{2+n(\gamma-2)}}
|\bm{\mathfrak{M}}|^{\gamma}\int_{\mathbf{R}^{n}}|\rho(x)|^{\gamma}d^{n}x
\nonumber\\&
-n^{2}(n-2)|V(\mathbf{B}_{n}(1))|\mathscr{G}
M^{2}\left|\frac{\lambda^{n-2}}{\lambda^{n}}\right|\int_{|\xi|\le \lambda}d^{n}\xi\nonumber\\&=
-\frac{n(n-2)|V(\mathbf{B}_{n}(1)M^{2-\gamma}|\mathscr{G}}{\lambda^{2+n(\gamma-2)}}
|\bm{\mathfrak{M}}|^{\gamma}\int_{\mathbf{R}^{n}}|\rho(x)|^{\gamma}d^{n}x
\nonumber\\&
-n^{2}(n-2)|V(\mathbf{B}_{n}(1))|\mathscr{G}
M^{2}\left|\frac{\lambda^{n-2}}{\lambda^{n}}\right|V(\mathbf{B}_{n}(1)|\lambda^{n}\nonumber
\\&=\frac{n(n-2)|V(\mathbf{B}_{n}(1)M^{2-\gamma}|\mathscr{G}}{\lambda^{2+n(\gamma-2)}}
|\bm{\mathfrak{M}}|^{\gamma}\int_{\mathbf{R}^{n}}|\rho(x)|^{\gamma}d^{n}x
\nonumber\\&
+n^{2}(n-2)|V(\mathbf{B}_{n}(1))|^{2}\mathscr{G}
M^{2}\left|{\lambda^{n-2}}\right|V(\mathbf{B}_{n}(1)|
\end{align}
\end{enumerate}
Next, recall that the total energy is constant so that $\mathbf{E}(t)=\mathbf{E}(0)=\mathbf{E}$ where
\begin{align}
\mathbf{E}=\frac{\mathsf{K}}{|\gamma-1|}\int_{\mathbf{R}^{n}}|\rho(x)|^{\gamma}d^{n}x
+\frac{1}{2}\int_{\mathbf{R}^{n}}\rho(x)\Phi(x)d^{n}x
\end{align}
so that
\begin{align}
&\frac{K}{|\gamma-1|}\int_{\mathbf{R}^{n}}|\rho(x)|^{\gamma}d^{n}x\le
E-\frac{1}{2}\int_{\mathbf{R}^{n}}\rho(x)\Phi(x)d^{n}x\nonumber\\&
\le \mathbf{E}+n(n-2)|V(\mathbf{B}_{n}(1))|\mathscr{G}M^{2-\gamma}
|\bm{\mathfrak{M}}|_{\gamma}^{\gamma}\int_{\mathbf{R}^{n}}|\rho(x)|^{\gamma}d^{n}x\nonumber\\&
+n^{2}(n-1)|V(\mathbf{B}_{n}(1))|^{2}M^{2}\lambda^{n-2}\nonumber\\&\le \mathbf{E}+\frac{\mathsf{K}}{2|\gamma-1|}
\int_{\mathbf{R}^{n}}|\rho(x)|^{\gamma}d^{n}x\nonumber\\&
+n^{2}(n-1)|V(\mathbf{B}_{n}(1))|^{2}M^{2}\lambda^{n-2}
\end{align}
where
\begin{align}
\frac{n(n-2)|V(\mathbf{B}_{n}(1))|M^{2-\gamma}|\bm{\mathfrak{M}}|_{\gamma}^{\gamma}}
{\lambda^{2+n(\lambda-2)}}\le \frac{\mathsf{K}}{|\gamma-1|}
\end{align}
\end{proof}
\section{Perturbations And (Linear) Stability Criteria For Static Gaseous Stars In Hydrostatic Equilibrium}
In this section, Eulerian and Lagrangian perturbations of a self-gravitating fluid/gas configuration in hydrostatic equilibrium are considered. (This section closely follows Chapter of Teukolsky and Shapiro $\bm{[37]}$. Establishing the Existence and uniqueness of solutions of the Lane-Emden or Chandrasekhar equation, and  derivation of actual solutions, is still insufficient--one must also establish the perturbative stability/instability, and hence physical existence, of steady state or equilibrium fluid/gas configurations. However, despite extensive studies of the LE equation and polytropes going back a century, no complete rigorous mathematical theory or proofs of nonlinear stability/instability exist; or at best exist only partially. It is well known that compactly supported equilibrium solutions exist for $\gamma\in[\tfrac{6}{5},2)$. Classical linear stability analyses establish the following dichotomy centred at $\gamma=\tfrac{4}{3}$:
\begin{align}
&If~\gamma\in (\tfrac{6}{5},\tfrac{4}{3}),~equilibrium~solutions~are~lineary~unstable;\nonumber\\&
If~\gamma\in [\tfrac{4}{3},2),~equilibrium~solutions~are~lineary~stable\nonumber
\end{align}
The critical case $\gamma=\tfrac{4}{3}$ admits $0$ as the 1st eigenvalue and is generally considered to be linearly stable or neutrally stable but teetering on the brink of instability. Nonlinear stability of the range $[\tfrac{4}{3},2)$ has been (partially) shown in $\bm{[32]}$, and nonlinear instability in the range $(\tfrac{6}{5},\tfrac{4}{3}]$ has been established in $\bm{[36]}$. Here, instability is induced by a growing mode in a linearised operator. The critical case $\gamma=\tfrac{4}{3}$ has been shown to be nonlinearly unstable despite an absence of growing modes in the linearised analysis. 

The basic methods of linear stability of fluid/gas spheres in hydrostatic equilibrium are reviewed. As already established, polytropic stars with $\gamma=4/3$ are on the brink of (catastrophic) instability. For example, in (6.135) one also has $\mathbf{E}=0$ if $\gamma=4/3$ and $\mathbf{E}<0$ for stability. The following theorem establishes that a homologous expansion or contraction of a polytropic star, leaves it in hydrostatic equilibrium iff $\gamma=4/3$.
\begin{thm}
Consider a polytropic gaseous star in hydrostatic equilibrium with pressure $p(r)$ and density $\rho(r)$ and mass $M$. If the star is subject to homologous contractions or expansions that leave the star in hydrostatic equilibrium, with mass M remaining constant, then this is only possible for a polytrope with $\gamma =4/3$.
\end{thm}
\begin{proof}
For a polytropic star in equilibrium the usual hydrostatic equilibrium equations are
\begin{align}
&\frac{dp(r)}{dr}=-\frac{\mathscr{G}\mathcal{M}(r)\rho(r)}{r^{2}}\\&
\mathcal{M}(r)=\int_{0}^{r}4\pi r^{\prime 2}\rho(r^{\prime}dr^{\prime}
\end{align}
Let $\mathfrak{H}:(p(r),\rho(r),r)\rightarrow(\overline{p(r)},\overline{\rho(r)},\overline{r})$ be a homologous transformation of the pressure and density that leaves invariant the equations of HE such that
\begin{align}
&\overline{p(r)}=\mathcal{A}p(r)\\&
\overline{\rho(r)}=\mathcal{B}\rho(r)\\&
\overline{r}=\mathcal{C}r
\end{align}
The FEHE is left invariant by the homologous transforms if
\begin{align}
\frac{dp(r)}{dr}=-\frac{\mathscr{G}\mathcal{M}(r)\rho(r)}{r^{2}}=\frac{d\overline{p(r)}}{d\overline{r}}=
-\frac{\mathscr{G}\overline{\mathcal{M}(r)}\overline{\rho(r)}}
{\overline{r}^{2}}
\end{align}
and $\overline{\mathcal{M}(\overline{r})}=\mathcal{M}(r)$. Hence
\begin{align}
\frac{\mathcal{A}}{\mathcal{C}}\frac{dp(r)}{dr}=-\frac{\mathscr{G}\mathcal{BC}^{2}\mathcal{M}(r)\rho(r)}{\mathcal{C}^{2}r}
=-\frac{\mathscr{G}\mathcal{B}\rho(r)}{r^{2}}
\end{align}
so that
\begin{align}
\mathcal{AC}=\mathcal{B}
\end{align}
Next, the mass M and $\mathcal{M}(r)$ must be left invariant so that
\begin{align}
\int_{0}^{r}4\pi r^{2}\rho(r^{\prime}dr^{\prime}=\int_{0}^{Cr}4\pi r^{\prime 2} \mathcal{C}^{3}\mathcal{B}\rho(r^{\prime})dr^{\prime}
\end{align}
or
\begin{align}
\mathcal{BC}^{3}=1
\end{align}
Eliminating $\mathcal{C}$ between (7.8) and (7.10) then gives
\begin{align}
\mathcal{A}=\mathcal{B}^{4/3}
\end{align}
so that the pressure and density transform as
\begin{align}
&\overline{p(r)}=\mathcal{A}p(r)\\&
\overline{\rho(r)}=\mathcal{B}\rho(r)=\mathcal{A}^{3/4}\rho(r)
\end{align}
Hence
\begin{align}
&\left|\frac{\overline{p(r)}}{p(r)}\right|=\mathcal{A}^{3/4}\\&
\left|\frac{\overline{\rho(r)}}{\rho(r)}\right|^{4/3}=\mathcal{A}
\end{align}
Finally
\begin{align}
\left|\frac{\overline{p(r)}}{p(r)}\right|=\left|\frac{\overline{\rho(r)}}{\rho(r)}\right|^{4/3}=\left|\frac{\mathsf{K}\overline{\rho(r)}}{\mathsf{K} \rho(r)}\right|^{4/3}
\end{align}
so that $\gamma=4/3$.
\end{proof}
\begin{defn}
The 'macroscopic' or Eulerian approach  considers variations in fluid variables at a specific point in space.
Perturbations $\delta\rho,\delta p,\delta \bm{u}$ are Eulerian perturbations or variations. If $\mathfrak{S}(x,t)$ is any property/attribute
of a perturbed flow and $\mathfrak{S}_{*}(x,t)$ the unperturbed flow then
\begin{align}
\delta\mathfrak{S}(x,t)=|\mathfrak{S}(x,t)-\mathfrak{S}_{*}(x,t)|
\end{align}
\end{defn}
\begin{defn}
The 'microscopic' or Lagrangian approach considers Lagrangian displacements $\xi(x,t)$, connecting fluid elements in the unperturbed state to elements in the
perturbed state. The Lagrangian change is defined by $\Delta \mathfrak{S}$. Then
\begin{align}
\Delta\mathfrak{S}=|\mathfrak{S}(x+\xi(x,t),t)-\mathfrak{S}_{*}(x,t)|
\end{align}
An element at x is displace to $x+\xi$. The Eulerian and Lagrangian variations are related by
\begin{align}
\Delta=\delta+\bm{\xi}.\bm{\nabla}\equiv \delta+\xi_{i}\nabla^{i}
\end{align}
\end{defn}
Note that in this section $\Delta$ will denote the Lagrangian variation and not the Laplacian operator, which here will be denoted by
$\nabla^{2}$.
The Lagrangian change or variation in the fluid velocity $\Delta \bm{u}$ is the velocity of the perturbed flow $u+\xi(x,t)$ relative to the velocity of the same
element at $x$ in the unperturbed flow so that
\begin{align}
\Delta\bm{u}=\frac{d}{dt}(\bm{x}+\xi)-\frac{d\bm{x}}{dt}=\frac{d\xi}{dt}
\end{align}
\begin{lem}
The following commutation relations hold
\begin{align}
&\delta \frac{d}{dt}-\frac{d}{dt}\delta=\left[\delta \frac{d}{dt},\frac{d}{dt}\delta\right]=0\\&
\delta \nabla_{i}-\nabla_{i}\delta=[\delta\nabla_{i},\nabla_{i},\delta]=0\\&
\Delta\frac{d}{dt}-\frac{d}{dt}\Delta=\left[\Delta\frac{d}{dt},\frac{d}{dt}\Delta\right]=-\partial_{t}\bm{\xi}.\bm{\nabla}\\&
\Delta\nabla_{i}-\nabla_{i}\Delta=[\Delta,\nabla_{i}]=\nabla_{i}\xi^{j}\nabla_{j}\\&
\Delta\frac{d}{dt}-\frac{d}{dt}\Delta=\left[\Delta\frac{d}{dt},\frac{d}{dt}\Delta\right]=0\\&
\delta\frac{d}{dt}-\frac{d}{dt}\Delta=\left[\delta,\frac{d}{dt},\frac{d}{dt}\Delta\right]=-(\xi,\nabla)\frac{d}{dt}
\end{align}
\end{lem}
It is important to  consider perturbations of integral quantities.
\begin{lem}
Let a fluid/gas have support in some $\mathbf{D}\in\mathbf{R}^{3}$ with boundary $\delta\mathbf{D}$. Let $\mathfrak{S}_{*}$ denote an attribute
of the unperturbed flow and consider the integral
\begin{align}
I=\int_{\mathbf{D}}\mathfrak{S}_{*}(x,t)dV(x)
\end{align}
The integral with respect to the perturbed flow is
\begin{align}
I=\int_{\mathbf{D}+\Delta\mathbf{D}}\mathfrak{S}(x,t)dV(x)
\end{align}
Under the perturbation, $\mathbf{D}\rightarrow\mathbf{D}+\Delta\mathbf{D}$ by subjecting the boundary $\partial\mathbf{D}$ to a displacement $\xi$.
The 1st variation of I is then
\begin{align}
\delta I=\int_{\mathbf{D}+\Delta\mathbf{D}}\mathfrak{S}(x,t)dV(x)-\int_{\mathbf{D}}\mathfrak{S}_{*}(x,t)dV(x)
\end{align}
It can be shown that the 1st variation has the integral representation
\begin{align}
\delta I=\int_{\mathbf{D}}(\Delta\mathfrak{S}+\mathfrak{S}\bm{\nabla}.\bm{\xi})dV(x)
=\int_{\mathbf{D}}(\Delta\mathfrak{S}+\mathfrak{S}\bm{\nabla}_{i}\bm{\xi})^{i}dV(x)
\end{align}
\end{lem}
\begin{proof}
Let $x^{\prime}=x-\xi(x,t)$ so that $x=x^{\prime}+\xi(x,t)$. The Jacobian is then
\begin{align}
\mathcal{J}(x,x^{\prime})= \frac{\partial(x)}{\partial(x)}=\frac{\partial(x^{\prime}+\xi)}{\partial(x^{\prime})}
=\frac{\partial(x^{\prime})}{\partial(x^{\prime})}+\frac{\partial{\xi}}{\partial(x^{\prime})}=1+\bm{\nabla}^{\prime}.\bm{\xi}
\end{align}
so that (7.30) becomes
\begin{align}
&\delta I=\int_{\mathbf{D}}\mathfrak{S}(x^{\prime}+\xi,t)\mathcal{J}(x,x^{\prime})
dV(x)-\int_{\mathbf{D}}\mathfrak{S}_{*}(x,t)dV(x)\\&
=\int_{\mathbf{D}}\mathfrak{S}(x^{\prime}+\xi,t)(1+\bm{\nabla}^{\prime}.\bm{\xi})dV(x)-\int_{\mathbf{D}}\mathfrak{S}_{*}(x,t)dV(x)\\&
=\int_{\mathbf{D}}\mathfrak{S}(x^{\prime}+\xi,t)dV(x)+\int_{\mathbf{D}}\mathfrak{S}(x^{\prime}+\xi,t)\bm{\nabla}^{\prime}.\bm{\xi})dV(x)
-\int_{\mathbf{D}}\mathfrak{S}_{*}(x,t)dV(x)\\&
=\int_{\mathbf{D}}\mathfrak{S}(x,t)dV(x)+\int_{\mathbf{D}}\mathfrak{S}(x,t)\bm{\nabla}^{\prime}.\bm{\xi})dV(x)
-\int_{\mathbf{D}}\mathfrak{S}_{*}(x,t)dV(x)\\&
=\int_{\mathbf{D}}(\Delta \mathfrak{S}+\mathfrak{S}\bm{\nabla}.\bm{\xi})dV(x)
\end{align}
which establishes (7.30).
\end{proof}
\begin{lem}
\begin{align}
\delta \int_{\mathbf{D}}\mathfrak{S}\rho d^{3}x=\int_{\mathbf{D}}\rho \Delta \mathfrak{S} d\mathrm{V}x
\end{align}
\end{lem}
\begin{lem}
The total mass $M$ of a static self-gravitating system in equilibrium is
\begin{align}
M=\int_{\Omega}\rho(x)dV(x)
\end{align}
Conservation of mass $\delta M=0$ then implies that
\begin{align}
&\Delta\rho=-\rho(\bm{\nabla}.\bm{\xi})\equiv -\rho\nabla_{i}\xi^{i}\\&
\delta\rho=-\bm{\nabla}.(\rho\bm{\xi})\equiv -\nabla_{i}(\rho\xi^{i})
\end{align}
\end{lem}
\begin{proof}
Using the integral relation
\begin{align}
\delta M=\int_{\mathbf{D}}(\Delta\rho+\rho(\bm{\nabla}.\bm{\xi})d^{3}x=0
\end{align}
so that (7.39) immediately follows. Using $\Delta=\delta+\bm{\xi}.\bm{\nabla}$ then gives
(7.40).
\end{proof}
\begin{lem}
Let $\mathscr{J}$ be the internal thermal energy per unit mass of the fluid/gas with an equation of state $p(\rho,s)$ for density $\rho$ and entropy density $s$, then the Lagrangian perturbation $\Delta\mathscr{J}$ is
\begin{align}
\Delta\mathscr{J}=\frac{p}{\rho^{2}}\Delta\rho
\end{align}
for adiabatic perturbations with $\Delta s=0$.
\end{lem}
\begin{proof}
The perturbation is
\begin{align}
\Delta\mathscr{J}=\frac{\partial\mathscr{J}}{\partial \rho}\Delta\rho+\frac{\partial\mathscr{J}}{\partial s}\Delta s
\end{align}
Using the 1st law of thermodynamics
\begin{align}
d\mathscr{J}=-pd\left|\frac{1}{\rho}\right|+\Theta ds
=\frac{p}{\rho^{2}}d\rho+\Theta ds
\end{align}
If $ds=0$ then
\begin{align}
\left|\frac{d\mathscr{J}}{d\rho}\right|_{s}=\frac{p}{\rho^{2}}
\end{align}
so that
\begin{align}
\Delta\mathscr{J}=\frac{\partial\mathscr{J}}{\partial \rho}\Delta\rho=\frac{p}{\rho^{2}}\Delta\rho
\end{align}
\end{proof}
\begin{lem}
The perturbation of the gravitational potential is
\begin{align}
\delta \phi(x)=\mathscr{G}\int_{\mathbf{D}}\frac{\bm{\nabla}.(\rho(x^{\prime},t)\bm{\xi}
(x^{\prime},t))d\mathrm{V}(x^{\prime})}{|x-x^{\prime}|}
\end{align}
\end{lem}
\begin{proof}
The perturbed potential obeys the perturbed Poisson equation so that $
\nabla^{2}\delta\phi=4\pi \mathscr{G}\rho$. Using $ \delta\rho=-\bm{\nabla}.(\rho\bm{\xi})$
\begin{align}
\delta\phi(x)=-\mathscr{G}\int_{\mathbf{D}}\frac{\delta\rho(x,t)d^{3}x}{|x-x^{\prime}|}=\delta \phi(x)=\mathscr{G}\int_{\mathbf{D}}\frac{\bm{\nabla}.(\rho(x^{\prime},t)\bm{\xi}
(x^{\prime},t))d^{3}x^{\prime}}{|x-x^{\prime}|}
\end{align}
\end{proof}
\subsection{Equilibrium as a critical point of the total energy}
In Section 4, the FEHE was derived via the 1st variation of the total energy. Here an alternative
and more elegant derivation is given via perturbation theory and the lemmas already presented in this section. The total energy $\mathbf{E}(t)$ has been presented several times but these are all equivalent in that
\begin{align}
&\mathbf{E}=\mathbf{E}_{K}+\mathbf{E}_{T}+\mathbf{E}_{G}=\int_{\mathbf{D}}\rho|\bm{u}|^{2}dV(x)+\int_{\mathbf{D}}\frac{p}{|\gamma-1|}dV(x)
+\int_{\mathbf{D}}\frac{1}{2}\Phi\rho dV(x)\\&
\equiv\int_{\mathbf{D}}\rho|\bm{u}|^{2}dV(x)+\int_{\mathbf{D}}\mathcal{E}d^{3}x
+\int_{\mathbf{D}}\frac{1}{2}\Phi\rho dV(x)\\&\equiv \int_{\mathbf{D}}\rho|\bm{u}|^{2}dV(x)+\int_{\mathbf{D}}\mathfrak{T}\rho
dV(x)+\int_{\mathbf{D}}\frac{1}{2}\Phi\rho dV(x)
\end{align}
where $\mathcal{E}$ is the internal energy per unit mass. For perturbations away from hydrostatic equilibrium $\bm{u}=0$ so that $\mathbf{E}_{K}=0$ and
\begin{align}
\mathbf{E}=\mathbf{E}_{T}+\mathbf{E}_{G}=\int_{\mathbf{D}}\mathfrak{T}\rho
d^{3}x+\int_{\mathbf{D}}\frac{1}{2}\Phi\rho dV(x)
\end{align}
\begin{thm}
The variation of the total energy $\delta \mathbf{E}$ with respect to the perturbation $\bm{\xi}$ gives the fundamental equation of hydrostatic equilibrium. If
\begin{align}
\delta\mathbf{E}=\delta\mathcal{T}+\delta\mathcal{W}=\delta\left(\int_{\mathbf{D}}\mathfrak{T}\rho
dV(x)+\int_{\mathbf{D}}\frac{1}{2}\Phi\rho dV(x)\right)=0
\end{align}
then
\begin{align}
\underbrace{\bm{\nabla}p + \rho\bm{\nabla}\Phi=0}
\end{align}
which is again the FEHE.  Therefore, static equilibrium is a critical point of the total energy $\mathbf{E}$ of the star.
\end{thm}
\begin{proof}
The 1st variation is
\begin{align}
\delta\mathbf{E}=\delta\mathcal{T}+\delta\mathcal{W}=\delta\int_{\mathbf{D}}\mathfrak{T}\rho
dV(x)+\delta \int_{\mathbf{D}}\frac{1}{2}\Phi\rho dV(x)=0
\end{align}
We make use of the following identities which have already been derived/proved
\begin{align}
&\int_{\mathbf{D}}\mathfrak{S}\rho d^{3}x=\int_{\mathbf{D}}\Delta\mathfrak{S}\rho dV(x)\\&
\Delta u=\frac{p}{\rho^{2}}\Delta\rho\\&
\Delta\rho=-\rho\bm{\nabla}.\bm{\xi}=-\rho\nabla_{j}\xi^{j}\\&
\Delta=\delta+\bm{\xi}.\bm{\nabla}\\&
\delta \rho=-\nabla.(\rho\bm{\xi})\\&
\delta\phi(x)=-\mathscr{G}\int_{\mathbf{D}}\frac{\delta\rho(x^{\prime},t)d^{3}x^{\prime}}{|x-x^{\prime}|} =\mathscr{G}\int_{\mathbf{D}}\frac{\bm{\nabla}.(\rho(x^{\prime},t)\bm{\xi}
(x^{\prime},t))dV(x^{\prime})}{|x-x^{\prime}|}
\end{align}
The variations $\delta\mathbf{E}_{T}$ and $\delta\mathbf{E}_{G}$ are computed separately.
\begin{itemize}
\item The variation of the internal thermal energy is
\begin{align}
&\delta\mathbf{E}_{T}=\delta\int_{\mathbf{D}}\mathfrak{T}\rho dV(x)=\int_{\mathbf{D}}\Delta \mathfrak{T}\rho dV(x)\nonumber\\&
=\int_{\mathbf{D}}\frac{p}{\rho}\Delta\rho dV(x)=-\int_{\mathbf{D}}p(\nabla.\xi)dV(x)\equiv
\int_{\mathbf{D}}p\nabla_{j}\xi^{j}dV(x)\equiv \int_{\Omega}\nabla_{j}p\xi^{j}dV(x)
\end{align}
where the last term is via integration by parts.
\item The variation in the gravitational energy is
\begin{align}
&\delta\mathbf{E}_{G}=\delta\int_{\mathbf{D}}\frac{1}{2}\rho\Phi dV(x)=\frac{1}{2}\int_{\mathbf{D}}\Delta\phi\rho dV(x)\nonumber\\&
=\frac{1}{2}\int_{\mathbf{D}}\big((\delta\Phi)\rho+\rho(\xi^{j}\nabla_{j})\Phi\big)d^{3}x
=\frac{1}{2}\int_{\mathbf{D}}\left(-\mathscr{G}\int_{\mathbf{D}}\frac{\delta\rho^{\prime}\rho dV(x^{\prime}}){|x-x^{\prime}|}+\rho(\xi^{j}\nabla_{j})\Phi\right)d^{3}x\nonumber\\&
=-\frac{1}{2}\int_{\mathbf{D}}\Phi\delta\rho dV(x)+\frac{1}{2}\int_{\mathbf{D}}\rho(\xi^{j}\nabla_{j})\Phi dV(x)\nonumber\\&=\frac{1}{2}\int_{\mathbf{D}}\rho\nabla_{j}\Phi\xi^{j}dV(x)+
\frac{1}{2}\int_{\mathbf{D}}\rho(\xi^{j}\nabla_{j})\Phi dV(x)=\int_{\mathbf{D}}\rho\nabla_{j}\phi\xi^{j}dV(x)
\end{align}
\item Adding the variations of the internal and gravitational energies then gives
\begin{align}
\delta\mathbf{E}=\delta\mathbf{E}_{T}+\delta\mathbf{E}_{G}=\int_{\mathbf{D}}(\nabla_{j}p+\rho\nabla_{j}\Phi)dV(x)=0
\end{align}
so that the fundamental equation of hydrostatic equilibrium follows
\begin{align}
\underbrace{\nabla_{j}p+\rho\nabla_{j}\Phi=0}
\end{align}
which again in spherical symmetry is equivalent to
\begin{align}
\underbrace{\frac{dp(r)}{dr}=-\frac{\mathscr{G}\rho(r)\mathcal{M}(r)}{r^{2}}}
\end{align}
\end{itemize}
\end{proof}
\subsection {}
In Section 3, deviations from hydrostatic equilibrium were considered in relation to the free-fall or hydrodynamic time. Here, perturbations of the Euler-Poisson equation describing an equilibrium fluid-gas configuration, are considered more rigorously. We consider the general (non-static) case of the EP system
\begin{align}
\frac{du^{i}}{dt}+\frac{1}{\rho}\nabla_{i}p+\nabla_{i}\phi=0
\end{align}
and consider Lagrangian perturbations of the form
\begin{align}
\Delta\left(\frac{du^{i}}{dt}+\frac{1}{\rho}\nabla_{i}p+\nabla_{i}\phi\right)=0
\end{align}
\begin{lem}
The perturbed EL system
\begin{align}
\Delta\left(\frac{du^{i}}{dt}+\frac{1}{\rho}\nabla_{i}p+\nabla_{i}\phi\right)=0
\end{align}
gives a 2nd-order ODE for the linear perturbations of the form
\begin{align}
\rho\frac{d^{2}\xi_{i}}{dt^{2}}-\frac{\Delta \rho}{\rho}\nabla_{i}p+\nabla_{i}\Delta p+\rho\nabla_{i}\Delta\Phi=0
\end{align}
\end{lem}
\begin{proof}
The following identities are required
\begin{align}
&\frac{D}{Dt}=\frac{d}{dt},~\bm{u}=0\\&
\Delta\frac{d}{dt}=\frac{d}{dt}\Delta\\&
\Delta\nabla_{i}=\nabla_{i}\Delta-\nabla_{i}\xi^{j}\nabla_{j}\\&
\Delta u^{i}=\frac{d\xi^{i}}{dt}\\&
\nabla_{i}p+\rho\nabla_{i}\phi=0
\end{align}
then the perturbed EL equation (7.69) becomes
\begin{align}
&\Delta\left(\rho\frac{du^{i}}{dt}+\nabla_{i}p+\rho\nabla_{i}\phi\right)=\Delta\left(\rho\frac{du^{i}}{dt}\right)
+\Delta\nabla_{i}p+\Delta(\rho\nabla_{i}\phi)\nonumber\\&
=\Delta\left(\rho\frac{du^{i}}{dt}\right)
+\Delta\nabla_{i}p+(\Delta\rho)(\nabla_{i}\phi)+\rho\Delta\nabla_{i}\phi\nonumber\\&
=(\Delta\rho)\frac{du^{i}}{dt}+\rho\frac{d\Delta u^{i}}{dt}
+\Delta\nabla_{i}p+(\Delta\rho)(\nabla_{i}\phi)+\rho\Delta\nabla_{i}\phi\nonumber\\&
=(\Delta\rho)\frac{du^{i}}{dt}+\rho\frac{d\Delta u^{i}}{dt}+\nabla_{i}\Delta p
-\nabla_{i}\xi^{j}\nabla_{j}p+\Delta\rho\nabla\phi+\rho\nabla_{i}\Delta\phi
-\underbrace{\rho\nabla_{i}\xi^{j}\nabla_{j}\phi}_{use~HE~eqn.}\nonumber\\&
=(\Delta\rho)\frac{du^{i}}{dt}+\rho\frac{d\Delta u^{i}}{dt}+\nabla_{i}\Delta p
-\nabla_{i}\xi^{j}\nabla_{j}p+\Delta\rho\nabla\phi+\rho\nabla_{i}\Delta\phi+\nabla_{i}\xi^{j}\nabla_{j}p
\nonumber\\&
\nonumber\\&=(\Delta\rho)\frac{du^{i}}{dt}+\rho\frac{d\Delta u^{i}}{dt}+\nabla_{i}\Delta p
+\underbrace{\Delta\rho\nabla\phi}_{Use~HE~eqn.}+\rho\nabla_{i}\Delta\phi\nonumber\\&=(\Delta\rho)\frac{du^{i}}{dt}+
\rho\frac{d\Delta u^{i}}{dt}+\nabla_{i}\Delta p
-\frac{\Delta\rho}{\rho}\nabla_{i}p
+\rho\nabla_{i}\Delta\phi\nonumber\\&
=\rho\frac{d\Delta u^{i}}{dt}+\nabla_{i}\Delta p
+\underbrace{\Delta\rho\nabla\phi}_{Use~HE~eqn.}+\rho\nabla_{i}\Delta\phi\nonumber\\&
=\rho\frac{d\Delta u^{i}}{dt}+\nabla_{i}\Delta p
-\frac{\Delta\rho}{\rho}\nabla_{i}p
+\rho\nabla_{i}\Delta\phi\nonumber\\&
=\rho\frac{d^{2} \xi^{i}}{dt^{2}}+\nabla_{i}\Delta p
-\frac{\Delta\rho}{\rho}\nabla_{i}p
+\rho\nabla_{i}\Delta\phi
\end{align}
which is (7.70).
\end{proof}
\begin{thm}
The differential equation for the perturbations $\xi^{i}$ can be expressed as a linear eigenvalue problem of the form
\begin{align}
\rho\partial_{t}^{2}\xi^{i}=\Xi_{ij}\xi^{j}
\end{align}
where
\begin{align}
\Xi_{ij}\xi^{j}=\nabla_{i}(\gamma_{1}p\nabla_{j}\xi^{j})-(\partial_{j}\xi^{j})
\nabla_{i}p+(\nabla_{i}\xi^{j})\nabla_{j}p+\rho\xi^{\j}\nabla_{i}\nabla_{j}\Phi
-\rho\nabla_{i}\delta \Phi
\end{align}
If $\xi_{i}(x,t)=\xi_{i}\exp(i\omega t)$ then
\begin{align}
-\rho\omega^{2}\xi_{i}=\Xi_{ij}\xi^{j}
\end{align}
\end{thm}
\begin{proof}
Beginning with equation (7.70)
\begin{align}
\rho\frac{d^{2}\xi^{i}}{dt^{2}}-\frac{\Delta\rho}{\rho}\nabla_{i}p+\nabla_{i}\Delta p+\rho\nabla_{i}\Delta\Phi=0
\end{align}
apply the following equations
\begin{align}
&\Delta\rho=-\rho\bm{\nabla}.\bm{\xi}=-\rho\nabla_{j}\xi^{j}\\&
\Delta=\delta+\bm{\xi}.\bm{\nabla}=\delta+\xi^{j}\nabla_{j}\\&
\Delta p=\gamma_{1}p\frac{\Delta\rho}{\rho}\\&
\Delta\Phi=\delta\Phi+\xi^{j}\nabla_{j}\Phi\\&
\nabla_{i}p=-\rho\nabla_{i}\Phi
\end{align}
Then
\begin{align}
&\rho\frac{d^{2}\xi^{i}}{dt^{2}}-\frac{\Delta\rho}{\rho}\nabla_{i}\rho+\nabla_{i}\Delta p+\rho\nabla_{i}\Delta\Phi\nonumber\\&
=\rho\frac{d^{2}\xi^{i}}{dt^{2}}
+\frac{\rho\nabla_{j}\xi^{j}}{\rho}\nabla_{i}p+
\nabla_{i}\left(\gamma_{1}p\frac{\Delta\rho}{\rho}\right)+\rho(\delta\Phi+\xi^{j}\nabla_{j}\Phi)
\nonumber\\&
=\rho\frac{d^{2}\xi^{i}}{dt^{2}}+(\nabla_{j}\xi^{j})\nabla_{i}p
-\nabla_{i}(\gamma_{1}p\nabla_{j}\xi^{j})+
\rho\nabla_{i}\delta\Phi+\rho\nabla_{i}(\xi^{j}\nabla_{j}\Phi)\nonumber\\&
=\rho\frac{d^{2}\xi^{i}}{dt^{2}}+(\nabla_{j}\xi^{j})\nabla_{i}p
-\nabla_{i}(\gamma_{1}p\nabla_{j}\xi^{j})+
\rho\nabla_{i}\delta\Phi+\rho(\nabla_{i}\xi^{j})\underbrace{\nabla_{j}\Phi}_{use~(6.8)}+
\rho\xi^{j}\nabla_{i}\nabla_{j}\Phi\nonumber\\&
=\rho\frac{d^{2}\xi^{i}}{dt^{2}}+(\nabla_{j}\xi^{j})\nabla_{i}p
-\nabla_{i}(\gamma_{1}p\nabla_{j}\xi^{j})+\rho\nabla_{i}\delta\Phi-(\nabla_{i}\xi^{j})(\nabla_{j}p)+
\rho\xi^{j}\nabla_{i}\nabla_{j}\Phi
\end{align}
which establishes (7.77) and (7.78). The normal modes
$\xi_{i}(x,t)=\xi_{i}\exp(i\omega t)$ then give the linear eigenvalue equation
\begin{align}
\rho\frac{d^{2}\xi_{i}}{dt^{2}}=-\rho\omega^{2}\xi_{i}=\Xi_{ij}\xi^{j}
\end{align}
as required.
\end{proof}
The eigenvalue-eigenvector equation (7.77) for the normal modes can also be derived form a variation principle.
\begin{lem}
Give the equation
\begin{align}
\rho\partial_{t}^{2}\xi^{i}(r)=\Xi_{ij}~\xi^{j}(r)
\end{align}
with $\Xi_{ij}=\Xi_{ji}$, then (7.88) follows from the variation
\begin{align}
\delta\mathbf{S}=\delta\int\delta \mathbf{L}dt=0
\end{align}
where
\begin{align}
\mathbf{L}=\mathbf{K}-\mathbf{V}
\end{align}
with
\begin{align}
&\mathbf{K}=\frac{1}{2}\int_{\mathbf{D}}\rho(\partial_{t}\xi^{i})^{2}dV\\&
\mathbf{V}=\frac{1}{2}\int_{\mathbf{D}}\xi^{j}\Xi_{ij}\xi^{j}dV=\frac{1}{2}\int[\gamma p(\nabla_{i}\xi^{i})^{2}+2(\nabla_{j}\xi^{j})\xi^{i}\nabla p+\xi^{i}\xi^{j}]
(\nabla_{i}\nabla_{j} p]dV\nonumber\\&+\int[\rho(\nabla_{i}\nabla_{j}\phi)+\rho\xi^{i}\nabla_{i}\delta \phi]dV
\end{align}
\end{lem}
\begin{proof}
\begin{align}
&\delta\mathbf{S}=\delta\int \mathbf{L}dt=\delta \frac{1}{2}\int[\rho(\partial_{t}\xi^{i})^{2}+\xi^{i}\Xi_{ij}\xi^{j}]dVdt\\&
=\int[\underbrace{\rho(\partial_{t}\xi^{i})(\partial_{t}\delta\xi^{i})}_{int.by~parts}+\delta\xi^{i}\Xi_{ij}\xi^{j}]dVdt\\&
=\int[-\rho\partial_{t}^{2}\xi^{i}+\Xi_{ij}\xi^{j}]\delta\xi^{i}dVdt=0
\end{align}
and so (7.92) follows.
\end{proof}
Equation (7.88) gives the gives the linear eigenvalue problem for the oscillations of the gaseous-fluid star away from hydrostatic equilibrium. In spherical symmetry, the purely radial modes satisfy a Sturm-Liouville eigenvalue equation known as the linear adiabatic wave equation. This differential equation is derived as follows.
\begin{lem}
In spherical symmetry equation (7.70) can be reduced to the Sturm-Liouville eigenvalue equation
\begin{align}
\frac{d}{dr}\left(\Gamma_{1}p\frac{1}{r}\frac{d}{dr}(r^{2}\xi(r)\right)-\frac{4}{r}\xi(r)\frac{dp(r)}{dr}+\rho\omega^{2}\xi(r)=0
\end{align}
with boundary conditions $\xi(r)=r$ at $r=0$,$\Delta p=0$ at $r=R$ and $\xi(R)<\infty$.
\end{lem}
\begin{proof}
Starting with
\begin{align}
&\nabla_{i}(\Gamma_{1}p\nabla_{j}\xi^{j})-(\nabla_{j}\xi^{j})\nabla_{i}p
+(\nabla_{i}\xi^{j})\nabla_{j}p-\rho\xi^{j}
\nabla_{j}\nabla_{i}\phi-\rho\nabla_{i}\delta\phi+\rho\omega^{2}\xi^{i}\nonumber\\&
=\nabla_{i}(\Gamma_{1}p\nabla_{j}\xi^{j})-\underbrace{2(\nabla_{j}\xi^{j})\nabla_{i}p+(\nabla_{j}\xi^{j})\nabla_{i}p
+(\nabla_{i}\xi^{j})\nabla_{j}p}-\rho\xi^{j}
\nabla_{j}\nabla_{i}\phi-\rho\nabla_{i}\delta\phi+\rho\omega^{2}\xi^{i}\nonumber\\&
=\nabla_{i}(\Gamma_{1}p\nabla_{j}\xi^{j})-\underbrace{2(\nabla_{j}\xi^{j})\nabla_{i}p
+(\nabla_{i}\xi^{j})\delta^{ij}\nabla_{i}p
+(\nabla_{i}\xi^{j})\nabla_{j}p}-\rho\xi^{j}
\nabla_{j}\nabla_{i}\phi-\rho\nabla_{i}\delta\phi+\rho\omega^{2}\xi^{i}\nonumber\\&
=\nabla_{i}(\Gamma_{1}p\nabla_{j}\xi^{j})-\underbrace{2(\nabla_{j}\xi^{j})\nabla_{i}p
+(\nabla_{i}\xi^{j})\nabla^{j}p
+(\nabla_{i}\xi^{j})\nabla_{j}p}-\rho\xi^{j}
\nabla_{j}\nabla_{i}\phi-\rho\nabla_{i}\delta\phi+\rho\omega^{2}\xi^{i}\nonumber\\&
=\nabla_{i}(\Gamma_{1}p\nabla_{j}\xi^{j})-\underbrace{2(\nabla_{j}\xi^{j})\nabla_{i}p
+2(\nabla_{i}\xi^{j})\nabla^{j}p}-\rho\xi^{j}
\nabla_{j}\nabla_{i}\phi-\rho\nabla_{i}\delta\phi+\rho\omega
^{2}\xi^{i}\nonumber\\&
=\nabla_{i}(\Gamma_{1}p\nabla_{j}\xi^{j})-\underbrace{2(\nabla_{j}\xi^{j})\nabla_{i}p
+2(\nabla_{i}\xi^{j})\nabla^{j}p}-\overbrace{\rho\xi^{j}\delta_{ij}
\nabla_{i}\nabla^{i}\phi-\rho\nabla_{i}\delta\phi}+\rho\omega^{2}\xi^{i}\nonumber\\&
=\nabla_{i}(\Gamma_{1}p\nabla_{j}\xi^{j})-\underbrace{2(\nabla_{j}\xi^{j})\nabla_{i}p
+2(\nabla_{i}\xi^{j})\nabla^{j}p}-\overbrace{\rho\xi_{i}
\nabla^{2}\phi-\rho\nabla_{i}\delta\phi}+\rho\omega^{2}\xi^{i}
\end{align}
For purely radial perturbations $\xi{i}=\xi(r,\theta,\varphi)=\xi(r)\chi(\theta,\varphi)$
and all angular derivatives are dropped. For any vector $V^{i}$ or $\mathbf{V}$ and any scalar $f$, the gradients, divergences and Laplacian in spherical coordinates are
$\mathbf{\nabla}.\mathbf{V}=\nabla_{j}V^{j}=\frac{1}{r^{2}}\frac{d}{dr}(r^{2}V(r))$,$\nabla f(r)=\tfrac{df}{dr},\mathbf{\nabla V}=\nabla_{i}V^{j}=
(\nabla V)_{i}^{j}=\tfrac{dV_{r}(r)}{dr}=\tfrac{dV(r)}{dr}$. The Laplacian is $\frac{1}{r^{2}}\frac{d}{dr}(r^{2}\frac{df}{dr})$. We also require the Poisson equation
$\nabla^{2}\phi=4\pi\mathscr{G}\rho$ and the expression for the perturbations of the Newtonian gravitational potential in spherical coordinates
\begin{align}
\nabla_{i}\delta\phi(r)=-4\pi \mathscr{G}\rho \xi(r)
\end{align}
Continuing, (7.97) now becomes
\begin{align}
&\frac{d}{dr}(\Gamma_{1} p\frac{1}{r^{2}}\frac{d}{dr}(r^{2}\xi(r))
-\frac{2}{r^{2}}\frac{d}{dr}(r^{2}\xi(r))\frac{dp(r)}{dr}+
2\frac{dp(r)}{dr}\frac{d\xi(r)}{dr}\nonumber\\&-\rho\xi(r) 4\pi\frac{1}{r^{2}}\frac{d}{dr}(r^{2}\frac{d\Phi}{dr})+4\pi\mathscr{G}\rho^{2}\xi(r)+\rho\omega^{2}\xi(r)
\nonumber\\&
=\frac{d}{dr}(\Gamma_{1} p\frac{1}{r^{2}}\frac{d}{dr}(r^{2}\xi(r))
-\frac{2}{r^{2}}\frac{d}{dr}(r^{2}\xi(r))\frac{dp(r)}{dr}+
2\frac{dp(r)}{dr}\frac{d\xi(r)}{dr}\nonumber\\&\underbrace{
-4\pi\mathscr{G}\rho^{2}\xi(r)+4\pi\mathscr{G}\rho^{2}\xi(r)}+\rho\omega^{2}\xi(r)\nonumber\\&
=\frac{d}{dr}(\Gamma_{1} p\frac{1}{r^{2}}\frac{d}{dr}(r^{2}\xi(r))
-\frac{2}{r^{2}}\frac{d}{dr}(r^{2}\xi(r))\frac{dp(r)}{dr}+
2\frac{dp(r)}{dr}\frac{d\xi(r)}{dr}+\rho\omega^{2}\xi(r)\nonumber\\&
=\frac{d}{dr}(\Gamma_{1} p\frac{1}{r^{2}}\frac{d}{dr}(r^{2}\xi(r))-\frac{2}{r^{2}}
\left(r^{2}\frac{d\xi(r)}{dr}+2r\xi(r)\right)\frac{dp(r)}{dr}+2\frac{d\xi(r)}{dr}\frac{dp(r)}{dr}+\rho\omega^{2}\xi(r)
\nonumber\\&
=\frac{d}{dr}(\Gamma_{1} p\frac{1}{r^{2}}\frac{d}{dr}(r^{2}\xi(r))-\frac{4}{r}\xi(r)\frac{dp(r)}{dr}
\underbrace{-2\frac{d\xi(r)}{dr}\frac{dp(r)}{dr}+2\frac{d\xi(r)}{dr}\frac{dp(r)}{dr}}+\rho\omega^{2}\xi(r)\nonumber\\&
=\frac{d}{dr}(\Gamma_{1} p\frac{1}{r^{2}}\frac{d}{dr}(r^{2}\xi(r))-\frac{4}{r}\xi(r)\frac{dp(r)}{dr}
+\rho\omega^{2}\xi(r)=0
\end{align}
\end{proof}
A basic solution of the LAWE is the fundamental mode with $\omega^{2}=0$, corresponding to a self-similar homologous deformation of the equilibrium configuration, but again only for a polytropic gas index of $\gamma=4/3$.
\begin{lem}
A basic solution of the LAWE is the fundamental mode
\begin{align}
\xi(r)=\mathcal{C} r,~~~\omega^{2}=0
\end{align}
iff $\gamma=4/3$.
\end{lem}
\begin{proof}
If $\omega^{2}=0$ then
\begin{align}
\frac{d}{dr}(\gamma p(r)\frac{1}{r^{2}}\frac{d}{dr}(r^{2}\xi(r))-\frac{4}{r}\frac{dp(r)}{dr}\xi(r)=0
\end{align}
If $\xi(r)=\mathcal{C} r $ is a solution then
\begin{align}
&\frac{d}{dr}(\gamma p(r)\frac{1}{r^{2}}\frac{d}{dr}(\mathcal{C}r^{3})-\frac{4}{r}\frac{dp(r)}{dr}\mathcal{C}r\\&
=\frac{d}{dr}(3\gamma p(r))-4\frac{dp(r)}{dr}=3\gamma\frac{d p(r)}{dr}-4\frac{dp(r)}{dr}=0
\end{align}
so that again $\gamma=4/3$.
\end{proof}
\clearpage
\section{Relativistic Stars And The Tolman-Oppenheimer-Snyder Equation}
For relativistic stars, the relativistic extension of the hydrostatic equilibrium equation is the Tolman-Oppenheimer-Volkoff equation (TOVE). This has been studied quite extensively in a number of works beginning with the original paper of Oppenheimer and Volkoff $\bm{[38]}$. and is particularly relevant to the study of neutron stars. In this section, the emphasis is on the various possible mathematical derivations based on constrained optimisational-variational methods and extremization of entropy, rather than physical considerations for which there is considerable literature. $\bm{[39,40]}$ and references their in.
\begin{thm}
Let $\bm{T}^{\mu\nu}$ be the energy momentum tensor for a perfect fluid/gas with pressure p, density $\rho$ and 4-velocity $u^{\mu}=(u^{1},u^{1},u^{2},u^{3})$. The EM tensor is then given by $\mathbf{T}^{\mu\nu}=p\mathbf{g}^{\mu\nu}+(p+\rho)
u^{\mu}u^{\nu}$. Then conservation of energy via the covariant derivative $\mathbf{D}_{\nu}\mathbf{T}^{\mu\nu}=0$ yields the following equation when the fluid
is static and in hydrostatic equilibrium
\begin{align}
\partial_{\lambda}p=-\frac{1}{2}(p+\rho)\frac{1}{g_{oo}}\partial_{\lambda}\bm{g}_{oo}
\end{align}
For spherical symmetry it is
\begin{align}
\frac{dp(r)}{dr}=-\frac{1}{2}(\rho(r)+p(r))\frac{1}{\bm{g}_{oo}(r)}
\frac{d\mathbf{g}_{oo}(r)}{dr}
\end{align}
\end{thm}
\begin{proof}
The covariant derivative is
\begin{align}
\mathbf{D}_{\nu}\mathbf{T}^{\mu\nu}=\partial_{\nu}p\bm{g}_{\nu\mu}
+\partial_{\nu}(p+\rho)u^{\mu}u^{\nu}+\Gamma_{\nu\lambda}^{\mu}(p+\rho)u^{\nu}u^{\lambda}
\end{align}
Since $\eta_{\alpha\beta}u^{\alpha}u^{\beta}=-1$ in the absence of gravitation then
$\eta_{\alpha\beta}u^{\alpha}u^{\beta}=-1$ in the presence of gravitation. For a static fluid $u^{i}=0$ and $u^{o}=\sqrt{-\bm{g}_{oo}}$ and all temporal derivatives of $p,\rho,\bm{g}_{\mu\nu}$ vanish. Also
\begin{align}
&\bm{\Gamma}_{oo}^{\mu}=-\frac{1}{2}\bm{g}^{\mu\nu}\partial_{\nu}\bm{g}_{oo}|u^{o}|^{2}=
-\frac{1}{2}\bm{g}^{\mu\nu}\frac{1}{\bm{g}_{oo}}\partial_{\nu}\bm{g}_{oo}\\&
\partial_{\nu}[(p+\rho)u^{\mu}u^{\nu}]=0
\end{align}
Then
\begin{align}
(\partial_{\nu}p)\bm{g}^{\mu\nu}=-\frac{1}{2}(p+\rho)\bm{g}^{\mu\nu}\frac{1}{\bm{g}_{oo}}
\partial_{\nu}\bm{g}_{oo}
\end{align}
Multiplying  through by $\bm{g}_{\mu\lambda}$ gives 
\begin{align}
\partial_{\lambda}p=-\frac{1}{2}(p+\rho)\frac{1}{\bm{g}_{oo}}\partial_{\lambda}\bm{g}_{oo}
\end{align}
and (8.2) follows for spherical symmetry.
\end{proof}
The FEHE now follows (in the non-relativistic limit) from the conservation of energy and momentum via the covariant derivative of the energy momentum tensor for a self-gravitating perfect fluid, with an interior metric given by the interior Schwartzchild metric.
\subsection{Einstein equations and their solution for a spherically symmetric perfect fluid/gas}
The formulation and derivations of the Einstein equations for a spherically symmetric self-gravitating sphere of perfect fuid/gas (of constant density) is a standard exercise and can be found in many texts on general relativity $\bm{[10,40]}$. The solution is the well-known interior Schwarzchild metric. The general solution of the Einstein equations is beyond human ability but solutions can be founds by imposing a high degree of symmetry. For a spherical star, one assumes the standard metric form so that
\begin{align}
ds^{2}=Y(r)dt^{2}-X(r)dr^{2}-r^{2}d\theta^{2}
-r^{2}\sin^{2}\theta d\phi^{2}
\end{align}
so that $\bm{g}_{tt}=-Y(r)$ and $\bm{g}_{rr}=-Y^{-1}(r)=X(r)$,$\bm{g}_{\theta\theta}=r^{2},
\bm{g}_{\varphi\varphi}=r^{2}\sin^{2}\theta$. The Einstein equation are then coupled to the energy momentum tensor of a perfect fluid.
\begin{align}
\mathbf{Ric}_{\nu\mu}-\frac{1}{2}\bm{g}_{\mu\nu}\mathbf{R}=8\pi\bm{\mathscr{G}}\mathbf{T}_{\mu\nu}
\end{align}
where as usual $\mathbf{T}_{\mu\nu}=p\bm{g}_{\mu\nu}+(p+\rho)u^{\mu}u^{\nu}$.
\begin{thm}
The solution is then the interior Schwarzchild solution with metric components
\begin{align}
X(r)=\left(1-\frac{2\bm{\mathscr{G}}\mathcal{M}(r)}{r}\right),~~~~Y(r)
=\left(1-\frac{2\bm{\mathscr{G}}\mathcal{M}(r)}{r}\right)^{-1}
\end{align}
where $\mathcal{M}(r)$ has the usual definition. At $r=R$ then
$M=\mathcal{M}(R)$ and the solution matches the standard exterior Schwarzchild solution for $r>R$.
\end{thm}
\begin{proof}
Using a prime to denote $d/dr$, the Einstein equations are
\begin{align}
&\mathbf{R}_{rr}= \frac{\mathbf{Y}^{\prime\prime}}{ 2\mathbf{Y}}-\frac{\mathbf{X}^{\prime}}{4\mathbf{Y}}\left(\frac{\mathbf{X}^{\prime}}{\mathbf{X}}+\frac{Y^{\prime}}{Y}
\right)-\frac{X^{\prime}}{rX}=-4\pi\bm{\mathscr{G}}(\rho-p)X\\&
\mathbf{R}_{\theta\theta}=-1+\frac{2}{2X}\left(-\frac{X^{\prime}}{X}+\frac{Y^{\prime}}{Y}
\right)-\frac{1}{X}=-4\pi\bm{\mathscr{G}}(\rho-p)r^{2}\\&
\mathbf{R}_{tt}=-\frac{Y^{\prime\prime}}{2X}+
\frac{Y^{\prime}}{4X}\left(\frac{X^{\prime}}{X}+\frac{Y^{\prime}}{Y}
\right)-\frac{Y^{\prime}}{rX}=-4\pi\bm{\mathscr{G}}(\rho+3p)Y
\end{align}
with $\mathbf{R}_{\theta\theta}=\mathbf{R}_{\varphi\varphi}$ and
$\mathbf{R}_{\mu\nu}=0$ if $\mu\ne\nu$. These equations can be combined to give a single equation for the field $X(r)$
\begin{align}
\frac{\mathbf{Ric}_{rr}}{2X}+\frac{\mathbf{Ric}_{\theta\theta}}{r^{2}}
+\frac{\mathbf{Ric}_{tt}}{2Y}=-\frac{X^{\prime}}{r^{2}X^{2}}-\frac{1}{r^{2}}+\frac{1}{Xr^{2}}
=-8\pi\bm{\mathscr{G}}\rho
\end{align}
This can be expressed as
\begin{align}
\frac{d}{dr}\left|\frac{r}{X(r)}
\right|=1-8\pi\bm{\mathscr{G}}\rho(r)
\end{align}
Integrating
\begin{align}
\left|\frac{r}{X(r)}
\right|=\int_{0}^{r}(1-8\pi\bm{\mathscr{G}}\rho(\overline{r}))d\overline{r}=r\left(1-\frac{2\mathcal{M}(r)}{r}\right)
\end{align}
If $|X(0)|<\infty$ then the solution is
\begin{align}
X(r)=\left(1-\frac{2\bm{\mathscr{G}}\mathcal{M}(r)}{r}\right)^{-1}
\end{align}
\end{proof}
\begin{thm}
If $\mathbf{D}_{\nu}\mathbf{T}^{\mu\nu}=0$ and $\mathbf{g}_{oo}(r)=(1-\tfrac{2\bm{\mathcal{G}}\mathcal{M}(r)}{r})$ then for $p\ll\rho$ and weak gravitational fields then the FEHE follows as
\end{thm}
\begin{proof}
The constraint $\mathbf{D}_{\nu}\mathbf{T}^{\mu\nu}=0$ gives (-) which in spherical symmetry is
\begin{align}
\frac{dp(r)}{dr}
=-\frac{1}{2}(p(r)+\rho(r))\frac{1}{\bm{g}_{oo}(r)}\frac{d}{dr}|\bm{g}_{oo}(r)|
\equiv-\frac{1}{2}(p(r)+\rho(r))\frac{Y^{\prime}(r)}{Y(r)}
\end{align}
For the interior metric (8.10) which solves the Einstein equations for a spherically symmetric self-gravitating perfect fluid/gas then $\mathbf{g}_{oo}(r)=(1-\tfrac{2G\mathcal{M}(r)}{r})$ so that(8.18)becomes
\begin{align}
\frac{dp(r)}{dr}=-\frac{1}{2}(p(r)+\rho(r)\frac{d}{dr}
\left\lbrace\left(1-\frac{2\bm{\mathscr{G}}\mathcal{M}(r)}{r}
\right)\right\rbrace\left(1--\frac{2\bm{\mathscr{G}}\mathcal{M}(r)}{r}
\right)^{-1}
\end{align}
When the gravitational fields are weak and when the fluid is non-relativistic then $p\ll \rho$. The last term can also be expanded to 1st order out so that
\begin{align}
&\frac{dp(r)}{dr}=-\frac{1}{2}\rho(r)\frac{d}{dr}
\left\lbrace\left(1-\frac{2\bm{\mathscr{G}}\mathcal{M}(r)}{r}
\right)\right\rbrace\left(1+\frac{2\bm{\mathscr{G}}\mathcal{M}r)}{r}+...\right)\nonumber\\&
=-\frac{1}{2}\rho(r)\left\lbrace\frac{d}{dr}(1-2\Phi(r))
\right\rbrace\left(1+2\Phi(r)+...\right)=-\frac{1}{2}\rho(r)\left\lbrace\frac{d}{dr}(1-2\Phi(r))
\right\rbrace\nonumber\\&
=\rho(r)\frac{d\Phi(r)}{dr}=-\rho(r)\frac{d}{dr}
\int_{r}^{\infty}\frac{\bm{\mathscr{G}}
\mathcal{M}(\overline{r})}{\overline{r}^{2}}dr
\end{align}
Hence the FEHE follows as
\begin{align}
\underbrace{\frac{dp(r)}{dr}=-\frac{\bm{\mathscr{G}}\mathcal{M}(r)\rho(r)}{r^{2}}}
\end{align}
\end{proof}
\begin{thm}
Given the results of Thm (-), then the pressure gradient $\tfrac{dp(r)}{dr}$ through the star is given by the Tolman-Oppenheinmer-Volkoff equation of relativistic hydrostatic equilibrium
\begin{align}
-r^{2}\frac{dp(r)}{dr}
=\bm{\mathscr{G}}\mathcal{M}(r)\rho(r)\left(1+\frac{p(r)}{\rho(r)}\right)
\left(1+\frac{4\pi\bm{\mathscr{G}}p(r)r^{3}}{\mathcal{M}(r)}\right)
\left(1-\frac{2\bm{\mathscr{G}}\mathcal{M}(r)}{r}\right)^{-1}
\end{align}
\end{thm}
\begin{proof}
The solution for $\mathbf{X}(r)$ and the hydrostatic equilibrium condition in the form
\begin{align}
\frac{\mathbf{Y}^{\prime}(r)}{\mathbf{Y}(r)}=-\frac{2p^{\prime}}{|\rho(r)+p(r)|}
\end{align}
can be used to eliminate the fields $X(r),Y(r)$ from the Einstein equation.
First
\begin{align}
\frac{r}{2\mathbf{X}(r)}=\frac{1}{2r}\left(1-\frac{2\bm{\mathscr{G}}\mathcal{M}(r)}{r}\right)
\end{align}
and
\begin{align}
&-\mathbf{X}^{\prime}(r)=-\left(1-\frac{2\bm{\mathscr{G}}\mathcal{M}(r)}{r}\right)^{-2}
\frac{d}{dr}\left(\frac{2\bm{\mathscr{G}}\mathcal{M}(r)}{r}\right)\nonumber\\&
=-\left(1-\frac{2\bm{\mathscr{G}}\mathcal{M}(r)}{r}\right)^{-2}
\left(-\frac{2\bm{\mathscr{G}}\mathcal{M}(r)}{r^{2}}-2\bm{\mathscr{G}}
\frac{d\mathcal{M}(r)}{dr}\right)\nonumber\\&
=-\left(1-\frac{2\bm{\mathscr{G}}\mathcal{M}(r)}{r}\right)^{-2}
\left(-\frac{2\bm{\mathscr{G}}\mathcal{M}(r)}{r^{2}}-8\pi\bm{\mathscr{G}}r^{2}\rho(r)
\right)
\end{align}
then
\begin{align}
-\frac{\mathbf{X}^{\prime}(r)}{\mathbf{X}(r)}=-\left(1-\frac{2\bm{\mathscr{G}}\mathcal{M}(r)}{r}
\right)^{-1}
\left(-\frac{2\bm{\mathscr{G}}\mathcal{M}(r)}{r^{2}}-8\pi\bm{\mathscr{G}}r^{2}\rho(r)
\right)
\end{align}
Equation (8.12) can then be written as
\begin{align}
&\mathbf{Ric}_{\theta\theta}=-1+\left(\frac{\bm{\mathscr{G}}\mathcal{M}(r)}{r}
-4\pi\bm{\mathscr{G}}\rho(r)r^{3}\right)+\left(1-\frac{2\bm{\mathscr{G}}\mathcal{M}(r)}{r}\right)
\left(-\frac{rp^{\prime}}{p(r)+\rho(r)|}\right)+\left(1-\frac{2\bm{\mathscr{G}}\mathcal{M}(r)}{r}\right)
\nonumber\\&=-1+\left(1-\frac{2\mathscr{G}\mathcal{M}(r)}{r}\right)\left(1-\frac{rp^{\prime}}{p(r)+\rho(r)|}
\right)+\frac{\bm{\mathscr{G}}\mathcal{M}(r)}{r}-4\pi \bm{\mathscr{G}}\rho(r)r^{2}=4\pi\bm{\mathscr{G}}(p(r)-\rho(r))r^{2}
\end{align}
Now using
\begin{align}
\frac{\bm{\mathscr{G}}\mathcal{M}(r)}{r}=2\frac{2\bm{\mathscr{G}}\mathcal{M}(r)}{r}
-\frac{\bm{\mathscr{G}}
\mathcal{M}(r)}{r}
\end{align}
equation (8.27)becomes
\begin{align}
-1+\left(1-\frac{2\bm{\mathscr{G}}\mathcal{M}(r)}{r}\right)\left(1-\frac{rp^{\prime}}{p(r)+\rho(r)|}
\right)+\frac{2\bm{\mathscr{G}}\mathcal{M}(r)}{r}-\frac{\bm{\mathscr{G}}
\mathcal{M}(r)}{r}
=4\pi\bm{\mathscr{G}}p(r)
\end{align}
or
\begin{align}
&-\left(1-\frac{2\bm{\mathscr{G}}\mathcal{M}(r)}{r}\right)+\left(1-\frac{2\bm{\mathscr{G}}\mathcal{M}(r)}{r}\right)
\left(1-\frac{rp^{\prime}}{p(r)+\rho(r)|}
\right)-\frac{\bm{\mathscr{G}}\mathcal{M}(r)}{r}\nonumber\\&
=\left(1-\frac{2\bm{\mathscr{G}}\mathcal{M}(r)}{r}\right)\left(\frac{-rp^{\prime}}{p(r)
+\rho(r)|}
\right)-\frac{\bm{\mathscr{G}}\mathcal{M}(r)}{r}=4\pi\bm{\mathscr{G}}p(r)
\end{align}
so that
\begin{align}
\left(1-\frac{2\bm{\mathscr{G}}\mathcal{M}(r)}{r}\right)\left(\frac{-rp^{\prime}(r)}{p(r)+\rho(r)|}
\right)=4\pi\bm{\mathscr{G}}p(r)+\frac{\bm{\mathscr{G}}\mathcal{M}(r)}{r}
\end{align}
This becomes
\begin{align}
-r\frac{p(r)}{dr}=\bm{\mathscr{G}}\mathcal{M}(r)(p(r)+\rho(r))\left(\frac{1}{r}+\frac{4\pi \bm{\mathscr{G}}\rho(r)r^{2}}{\mathcal{M}(r)}\right)
\left(1-\frac{2\bm{\mathscr{G}}\mathcal{M}(r)}{r}
\right)^{-1}
\end{align}
or
\begin{align}
-r^{2}\frac{dp(r)}{dr}
=\bm{\mathscr{G}}\mathcal{M}(r)\rho(r)\left(1+\frac{p(r)}{\rho(r)}\right)
\left(1+\frac{4\pi\bm{\mathscr{G}}p(r)r^{3}}{\mathcal{M}(r)}\right)
\left(1-\frac{2\bm{\mathscr{G}}\mathcal{M}(r)}{r}\right)^{-1}
\end{align}
or equivalently
\begin{align}
\underbrace{\frac{dp(r)}{dr}
=-\frac{\bm{\mathscr{G}}\mathcal{M}(r)\rho(r)}{r^{2}}}\left(1+\frac{p(r)}{\rho(r)}\right)
\left(1+\frac{4\pi\bm{\mathscr{G}}p(r)r^{3}}{\mathcal{M}(r)}\right)
\left(1-\frac{2\bm{\mathscr{G}}\mathcal{M}(r)}{r}\right)^{-1}
\end{align}
which is the Tolman-Oppenheimer-Volkoff equation for hydrostatic equilibrium of self-gravitating relativistic fluid/gas configurations in general relativity
\end{proof}
\begin{cor}
In the non-relativistic limit, the last three terms vanish and the TOV equation reduces to the hydrostatic equilibrium equation.
\end{cor}
\begin{cor}
For self-gravitating radiation or a photon gas of density $\rho(r)$ , the TOVE has the form
\begin{align}
\frac{1}{3}\frac{d\rho(r)}{dr}=-\frac{[\rho(r)+\tfrac{1}{3}\rho][\mathscr{G}\mathcal{M}(r)+\tfrac{4}{3}\pi r^{3} \rho]}{r(r-2\mathscr{G}\mathcal{M}(r))}
\end{align}
\end{cor}
\begin{proof}
For radiation $p(r)=\frac{1}{3}\rho(r)$ so that the TOVE becomes
\begin{align}
-\frac{1}{3}r^{2}\frac{d\rho(r)}{dr}=\frac{4}{3}\mathscr{G}\mathcal{M}(r)\rho(r)
\left(1+\frac{4\pi}{3}\frac{\mathscr{G}\rho(r)r^{3}}{\mathcal{M}(r)}\right)
\left(1-\frac{2\mathscr{G}\mathcal{M}(r)}{r}\right)^{-1}
\end{align}
or
\begin{align}
&-\frac{1}{3}r(r-2\mathscr{G}\mathcal{M}(r))\frac{d\rho(r)}{dr}=\frac{4}{3}\mathscr{G}
\mathcal{M}(r)\rho(r)+\frac{16}{3}
\pi\mathscr{G}
\rho(r)\rho(r)/3\nonumber\\&
=\mathscr{G}\mathcal{M}(r)\rho(r)+\frac{1}{3}\mathscr{G}\mathcal{M}(r)\rho(r)+4\pi\mathscr{G}\rho(r)\rho(r)
+\frac{4}{3}\pi\mathscr{G}\rho(r)\rho(r)\nonumber\\&
\equiv \left[\rho(r)+\frac{1}{3}\rho(r)/3\right][\mathscr{G}\mathcal{M}(r)+4\pi r^{3}\rho(r)/3]
\end{align}
which then gives (8.36) as required.
\end{proof}
\subsection{Some properties of the TOVE}
The foremost property of the TOVE is that is reduces to the FEHE in the non-relativistic limit. The TOVE can also be expressed as a Riccati equation $\bm{[43]}$
\begin{lem}
The TOVE can be expressed as a Riccati equation of the form
\begin{align}
\frac{dp(r)}{dr}=\alpha(r)+\beta(r)p(r)+\gamma(r)|p(r)|^{2}
\end{align}
where
\begin{align}
&f(r)=-\frac{G}{r^{2}}\left(1-\frac{2\mathscr{G}M}{r}\right)^{-1}\\&
\alpha(r)=f(r)\rho(r)\mathcal{M}(r)\\&
\beta(r)=f(r)(4\pi r^{3} \rho(r)+\mathcal{M}(r))\\&
\gamma(r)=f(r) 4\pi r^{3}
\end{align}
\end{lem}
\begin{proof}
\begin{align}
&\frac{dp(r)}{dr}=-\frac{\mathscr{G}\mathcal{M}(r)}{r^{2}}[\rho(r)+p(r)]
[\mathcal{M}(r)+4\pi r^{3} p(r)]\left(1-\frac{2\mathscr{G}\mathcal{M}(r)}{r}
\right)^{-1}\nonumber\\&
=f(r)[\rho(r)\mathcal{M}(r)+p(r)(4\pi r^{3}\rho(r)+\mathcal{M}(r))+4\pi r^{3}|p(r)|^{2})\nonumber
\\&=\alpha(r)+\beta(r)p(r)+\gamma(r)|p(r)|^{2}
\end{align}
\end{proof}
No complete general solutions exists from first principles. However, if there exists one solution $p_{o}(r)$ then the 1-parameter general solution parametrized by a real number $\xi$ is
\begin{align}
p(r)=\frac{p_{o}(r)+\xi\exp\left\lbrace\int_{0}^{r}[2\gamma(r)p_{o}(r)+\beta(r)]dr\right\rbrace
}{\left|1-\xi\int_{0}^{r}\gamma(r)\exp\left\lbrace \int_{0}^{r}2\gamma(r)p_{o}(r)+\beta(r)dr\right\rbrace dr\right|}
\end{align}
\section{Derivation Of The TOVE By Constrained Optimisational-Variational Methods}
The TOVE can also be derived by constrained optimisational-variatonal methods.
\begin{defn}
Let $M=\mathcal{M}(R)$ and let $M_{\infty}$ be the mass/energy of the matter comprising the star if it were dispersed to infinity. If $m_{\mathcal{N}}$ is the mass of a nucleon then the number of nucleons in the star is
\begin{align}
\mathcal{N}=(\mathbf{g})^{1/2}\int_{0}^{R}\mathcal{J}_{\mathcal{N}}^{0}d\mu
=\int_{0}^{R}4\pi r^{2}(X(r)Y(r))^{1/2}\mathcal{J}_{\mathcal{N}}^{0}(r)dr
\end{align}
where $\mathcal{J}_{\mathcal{N}}^{0}$ is the conserved nucleon number current. The internal energy of the star is then
\begin{align}
E=M-M_{\infty}=\mathrm{M}-m_{n}\mathcal{N}
\end{align}
The nucleon number density is denoted ${\eta}$ and is measured in a locally inertial frame at rest in the star so that $u_{r}=u_{\theta}(r)=u_{\varphi}=0$ and $u_{t}=u_{0}=-(Y(t)^{1/2}=-1$. The nucleon number density is $\mathfrak{D}=-u_{\mu}\mathcal{N}^{\mu}=-\sqrt{Y(r)}\mathcal{J}_{\mathcal{N}}^{0}$.
\begin{align}
\mathcal{N}=\int_{0}^{R}4\pi r^{2}(X(r))^{1/2}\mathfrak{D}(r)dr
=\int_{0}^{R}4\pi r^{2}\left(1-\frac{2\mathscr{G}\mathcal{M}(r)}{r}\right)^{-1/2}\mathfrak{D}(r)dr
\end{align}
\end{defn}
\begin{rem}
The proper number density $\mathfrak{D}$ is in general a function of the proper density $\rho(r)$, the chemical composition and the entropy per nucleon $\mathscr{S}$. Once $\rho(0)$ is chosen then $\mathfrak{D}$ and $\mathcal{N}$ are fixed for a star with constant $\mathscr{S}$.
\end{rem}
\begin{defn}
The proper internal material energy density $\mathcal{E}$ is defined as
\begin{align}
\mathcal{E}(r)=\rho(r)-m_{n}\mathcal{N}
\end{align}
Then (9.2) can be written as
\begin{align}
\mathbf{E}=\mathbf{E}_{T}+\mathbf{E}_{G}
\end{align}
where $\mathbf{E}$ and $\mathbf{E}$ are the relativistic thermal and gravitational energies respectively so that
\begin{align}
&\mathbf{E}_{T}=\int_{0}^{R}4\pi r^{2}\left(1-\frac{2\mathscr{G}\mathcal{M}(r)}{r}\right)^{-1/2}\mathcal{E}(r)dr\\&
\mathbf{E}_{G}=\int_{0}^{R}4\pi r^{2}\left(1-\left(1-\frac{2\mathscr{G}\mathcal{M}(r)}{r}\right)^{-1/2}\right)\rho(r)dr
\end{align}
Expanding out the $(1-2\mathscr{G}\mathcal{M}(r)/r)^{-1/2}$ terms
\begin{align}
&\mathbf{E}_{T}=\int_{0}^{R}4\pi r^{2}\left(1+\frac{\mathscr{G}\mathcal{M}(r)}{r}+...\right)\mathcal{E}(r)dr\\&
\mathbf{E}_{G}=\int_{0}^{R}4\pi r^{2}\left(\frac{2\mathscr{G}\mathcal{M}(r)}{r}
+\frac{3\mathscr{G}^{2}|\mathcal{M}(r)|^{2}}{2r^{2}}+...\right)\rho(r)dr
\end{align}
then the first-order terms are just the Newtonian thermal and gravitational energies of the star.
\end{defn}
The TOV equation can now be derived via a constrained optimization problem. The following result appears in $[10]$.
\begin{thm}
A star of mass/energy $M$, uniform entropy per nucleon $\mathscr{S}$, nucleon number $\mathcal{N}$, nucleon number density $\mathfrak{D}(r)$ and chemical composition $\mathscr{C}$ will satisfy the Tolman-Oppenheimer-Volkoff equation of hydrostatic equilibrium iff M defined as
\begin{align}
M=\mathcal{M}(R)=\int_{0}^{R}4\pi r^{2}\rho(r)dr
\end{align}
is stationary, that is $\partial M=0$, with respect to all density variations $\delta \rho(r)=0$ that leave $\mathcal{N}$ invariant or unchanged where
\begin{align}
\mathcal{N}=\int_{0}^{R}4\pi r^{2}\left(1-\frac{2\mathscr{G}\mathcal{M}(r)}{r}\right)^{-1/2}\mathfrak{D}(r)dr
\end{align}
and which leaves $\mathscr{S}$ and $\mathscr{C}$ unchanged, that is $\delta\mathscr{S}=0$ and $\delta\mathscr{C}=0$. This is then a constrained optimization-variational problem $\delta\mathcal{U}=|\delta M-\bm{\chi} \delta\mathcal{N}|=0 $, where $\bm{\chi}\in\mathbf{R}$ is a Lagrange multiplier. The hydrostatic equilibrium is then stable with respect to radial oscillations/perturbations if $M$ or $E$ is an extremum (minimum) to all such variations. The constrained optimizational-variational problem to be solved is
\begin{align}
&\partial\mathrm{U}=\delta M-\bm{\chi}\delta N
=\frac{\partial M}{\partial\rho(r)}\delta\rho(r)+\frac{\partial M}{\partial\mathfrak{D}(r)}\delta\mathscr{D}(r)
-\bm{\chi}\frac{\partial \mathcal{N}}{\partial\rho(r)}\delta\rho(r)-\bm{\chi}\frac{\partial\mathcal{N}}{\partial\mathfrak{D}(r)}\delta\mathfrak{D}(r)=0\\&
\delta\mathscr{S}=\frac{\partial\mathscr{S}}{\partial \rho(r)}\delta\rho(r)+\frac{\partial\mathscr{S}}{\partial \mathfrak{D}(r)}\delta\mathfrak{D}(r)=0
\end{align}
\end{thm}
\begin{proof}
Using (9.6) and (9.7) the variation is computed as follows. Note the integrals can be taken over $[0,\infty)$ but vanish outside $R+\delta R$.
\begin{align}
&\partial\mathrm{U}=\delta M-\bm{\chi}\delta\mathcal{N}
=\frac{\partial M}{\partial\rho(r)}\delta\rho(r)+\frac{\partial M}{\partial\mathcal{N}(r)}\delta\mathcal{N}(r)
-\bm{\chi}\frac{\partial \mathcal{N}}{\partial\rho(r)}\delta\rho(r)-\bm{\chi}\frac{\partial \mathcal{N}}{\partial\mathcal{N}(r)}\delta\mathcal{N}(r)\nonumber\\&
=\frac{\partial}{\partial \rho(r)}\left|\int_{0}^{\infty}4\pi r^{2}\rho(r)dr
\right|\delta\rho(r)+\frac{\partial}{\partial \mathfrak{D}(r)}\left|\int_{0}^{\infty}4\pi r^{2}\rho(r)dr
\right|\delta\mathfrak{D}(r)\nonumber\\&
-\bm{\chi}\frac{\partial}{\partial \rho(r)}\left|\int_{0}^{\infty}4\pi r^{2}\left(1-\frac{2\mathscr{G}\mathcal{M}(r)}{r}\right)^{-1/2}\mathfrak{D}(r)dr\right|
\delta\rho(r)\nonumber\\&-\bm{\chi}\frac{\partial}{\partial \mathfrak{D}(r)}\left|\int_{0}^{\infty}4\pi r^{2}\left(1-\frac{2\mathscr{G}\mathcal{M}(r)}{r}\right)^{-1/2}\mathfrak{D}(r)dr\right|
\delta\mathfrak{D}(r)\nonumber\\&
=\int_{0}^{\infty}4\pi r^{2}\delta\rho(r)dr-\bm{\chi}\int_{0}^{\infty}4\pi r^{2}\frac{\partial}{\partial\rho(r)}\left(1-\frac{2\mathscr{G}\mathcal{M}(r)}{r^{2}}\right)^{-1/2}
\mathfrak{D}(r)\delta\rho(r) dr\nonumber\\&-\bm{\chi}\int_{0}^{\infty}4\pi r^{2}\frac{\partial}{\partial\mathfrak{D}(r)}
\left(1-\frac{2\mathscr{G}\mathcal{M}(r)}{r^{2}}\right)^{-1/2}\mathfrak{D}(r)\delta\mathfrak{D}(r)dr\nonumber\\&
=\int_{0}^{\infty}4\pi r^{2}\delta\rho(r)dr-\bm{\chi}\int_{0}^{\infty}4\pi r^{2}\left(-\frac{1}{2}\right)\left(1-\frac{2\mathscr{G}\mathcal{M}(r)}{r^{2}}\right)^{-3/2}
\frac{\partial}{\partial\rho(r)}\left(-\frac{2\mathscr{G}\mathcal{M}(r)}{r^{2}}\right)\mathfrak{D}(r)\delta\rho(r) dr\nonumber\\&
-\bm{\chi}\int_{0}^{\infty}4\pi r^{2}\left(1-\frac{2\mathscr{G}\mathcal{M}(r)}{r^{2}}\right)^{-1/2}
\delta\mathfrak{D}(r)dr=\int_{0}^{\infty}4\pi r^{2}\delta\rho(r)dr\nonumber\\&-\bm{\chi}\int_{0}^{\infty}4\pi r^{2}\left(1-\frac{2\mathscr{G}\mathcal{M}(r)}{r^{2}}\right)^{-1/2}\delta\mathfrak{D}(r)dr
-\mathbb{L}\int_{0}^{\infty}4\pi r^{2}\mathfrak{D}(r)\left(1-\frac{2\mathscr{G}\mathcal{M}(r)}{r^{2}}\right)^{-3/2}
\delta\mathcal{M}(r)dr
\end{align}
The variations $\delta\mathscr{D}(r)$ and $\delta\rho(r)$ do not change the entropy per nucleon $\mathscr{S}$ so that from (9.13)
\begin{align}
0=&\delta\left|\frac{\rho(r)}{\mathfrak{D}(r)}\right|+p(r)\delta
\left|\frac{1}{\mathfrak{D}(r)}\right|
=\frac{\partial}{\partial \rho(r)}\bigg|\rho(r)\mathfrak{D}^{-1}(r)\bigg|\delta\rho(r)
+\frac{\partial}{\partial \mathfrak{D}(r)}\bigg|\rho(r)\mathfrak{D}^{-1}(r)\bigg|\delta
\mathfrak{D}(r)\nonumber\\&=\mathfrak{D}^{-1}(r)\delta\rho(r)-\rho(r)\mathfrak{D}^{-2}(r)
\delta\mathfrak{D}(r)-p(r)\mathfrak{D}^{-2}(r)\delta\mathfrak{D}(r)\\&
=\delta\rho(r)-(p(r)\rho(r))\mathcal{N}^{-1}(r)\delta\mathfrak{D}(r)
\end{align}
then
\begin{align}
\delta\mathfrak{D}(r)=\mathfrak{D}(r)|p(r)+\rho(r)|^{-1}\delta\rho(r)
\end{align}
Also
\begin{align}
\delta\mathcal{M}(r)=\int_{0}^{r}4\pi r^{\prime 2}\delta\rho(r)dr^{\prime}
\end{align}
Using (9.17) and (9.18) and then interchanging the integrals $\int dr
\rightarrow \int dr^{\prime}$
\begin{align}
&\delta M-\bm{\chi}\mathcal{N}=\int_{0}^{\infty}4\pi r^{2}\delta\rho(r)dr
-\bm{\chi}\int_{0}^{\infty}4\pi r^{2}\left(1-\frac{2\mathscr{G}\mathcal{M}(r)}{r}
\right)^{-1/2}\delta\mathfrak{D}(r)dr\nonumber\\&-\bm{\chi}\mathscr{G}
\int_{0}^{\infty}\int_{r}^{\infty}4\pi r
\left(1-\frac{2\mathscr{G}\mathcal{M}(r)}{r}\right)^{-3/2}\mathfrak{D}(r) 4\pi r^{\prime 2}\delta\rho(r)dr^{\prime}dr\nonumber\\&
=\int_{0}^{\infty}4\pi r^{2}\delta\rho(r)dr
-\bm{\chi}\int_{0}^{\infty}4\pi r^{2}\left(1-\frac{2\mathscr{G}\mathcal{M}(r)}{r}
\right)^{-1/2}\delta\mathfrak{D}(r)dr\nonumber\\&
-\bm{\chi}\left\lbrace\mathscr{G}\int_{0}^{\infty}4\pi r^{\prime 2}\left(1-\frac{2\mathscr{G}\mathcal{M}(r^{\prime})}{r^{\prime}}\right)^{-3/2}dr^{\prime}
\right\rbrace\delta\rho(r)dr\nonumber\\&
=\int_{0}^{\infty}4\pi r^{2}\left\lbrace 1-\bm{\chi}\frac{\mathfrak{D}(r)}{|p(r)-\rho(r)|}
\left(1-\frac{2\mathscr{G}\mathcal{M}(r)}{r}\right)^{-1/2}
-\bm{\chi}\mathscr{G}\int_{0}^{\infty} 4 \pi r^{\prime}\mathfrak{D}(r^{\prime})
\left(1-\frac{2\mathscr{G}\mathcal{M}(r^{\prime})}{r^{\prime}}\right)^{-3/2}dr^{\prime}\right\rbrace\nonumber\\&
=\int_{0}^{\infty}4\pi r^{2}\left\lbrace \bm{\chi}^{-1}-
\frac{\mathfrak{D}(r)}{|p(r)-\rho(r)|}
\left(1-\frac{2\mathscr{G}\mathcal{M}(r)}{r}\right)^{-1/2}
-\mathscr{G}\int_{0}^{\infty} 4 \pi r^{\prime}\mathfrak{D}(r^{\prime})
\left(1-\frac{2\mathscr{G}\mathcal{M}(r^{\prime})}{r^{\prime}}\right)^{-3/2}dr^{\prime}\right\rbrace=0
\end{align}
Then $\delta M-\bm{\chi}\mathcal{N}=0$ iff the inverse of the Lagrange multiplier is
\begin{align}
\bm{\chi}^{-1}=
\frac{\mathfrak{D}(r)}{|p(r)-\rho(r)|}
\left(1-\frac{2\mathscr{G}\mathcal{M}(r)}{r}\right)^{-1/2}
+\mathscr{G}\int_{0}^{\infty} 4 \pi r^{\prime}\mathfrak{D}(r^{\prime})
\left(1-\frac{2\mathscr{G}\mathcal{M}(r^{\prime})}{r^{\prime}}\right)^{-3/2}dr^{\prime}
\end{align}
This is only possible if the rhs is independent of r so that
\begin{align}
\frac{d\bm{\chi}^{-1}}{dr}&=
\frac{d}{dr}\left|\frac{\mathfrak{D}(r)}{|p(r)-\rho(r)|}
\left(1-\frac{2\mathscr{G}\mathcal{M}(r)}{r}\right)^{-1/2}
+\mathscr{G}\int_{0}^{\infty} 4 \pi r^{\prime}\mathfrak{D}(r^{\prime})
\left(1-\frac{2\mathscr{G}\mathcal{M}(r^{\prime})}{r^{\prime}}\right)^{-3/2}dr^{\prime}\right|\nonumber\\&
=\frac{d}{dr}\left|\frac{\mathfrak{D}(r)}{|p(r)-\rho(r)|}
\left(1-\frac{2\mathscr{G}\mathcal{M}(r)}{r}\right)^{-1/2}\right|
+4\mathscr{G} \pi r^{\prime}\mathfrak{D}(r^{\prime})
\left(1-\frac{2\mathscr{G}\mathcal{M}(r^{\prime})}{r^{\prime}}\right)^{-3/2}\nonumber\\&
=\left\lbrace\frac{\mathfrak{D}^{\prime}(r)}{ |p(r)+\rho(r)|}
-\frac{\mathfrak{D}(r)|p^{\prime}(r)+\rho^{\prime}(r))|}{|\rho(r)+p(r)|^{2}}\right\rbrace
\left(1-\frac{2 \mathscr{G}\mathcal{M}(r)}{r}\right)^{-1/2}\nonumber\\&-\frac{1}{2}
\frac{\mathfrak{D}}{|p(r)+\rho(r)|}\left(1-\frac{2\mathscr{G}\mathcal{M}(r)}{r}\right)^{-3/2}
\frac{d}{dr}\left(-\frac{2\mathscr{G}\mathcal{M}(r)}{r}\right)-4\pi\mathscr{G}\mathfrak{D}(r)
\left(1-\frac{2\mathscr{G}\mathcal{M}(r)}{r}\right)^{-3/2}
\nonumber\\&
=\left\lbrace\frac{\mathfrak{D}^{\prime}(r)}{ |p(r)+\rho(r)|}
-\frac{\mathfrak{D}(r)|p^{\prime}(r)+\rho^{\prime}(r))|}{|\rho(r)+p(r)|^{2}}\right\rbrace
\left(1-\frac{2 \mathscr{G}\mathcal{M}(r)}{r}\right)^{-1/2}\nonumber\\&
+\frac{\mathfrak{D}(r)\mathscr{G}}{|p(r)+\rho(r)|}
\left(1-\frac{2\mathscr{G}\mathcal{M}(r)}{r}\right)^{-3/2}
\left(\frac{1}{r}\frac{d\mathcal{M}(r)}{dr}-\frac{\mathcal{M}(r)}{r^{2}}\right)-4\pi\mathscr{G}\mathfrak{D}(r)
\left(1-\frac{2\mathscr{G}\mathcal{M}(r)}{r}\right)^{-3/2}\nonumber\\&
=\left\lbrace\frac{\mathfrak{D}^{\prime}(r)}{ |p(r)+\rho(r)|}
-\frac{\mathfrak{D}(r)|p^{\prime}(r)+\rho^{\prime}(r))|}{|\rho(r)+p(r)|^{2}}\right\rbrace
\left(1-\frac{2 \mathscr{G}\mathcal{M}(r)}{r}\right)^{-1/2}\nonumber\\&
+\frac{\mathfrak{D}(r)\mathscr{G}}{|p(r)+\rho(r)|}
\left(1-\frac{2\mathscr{G}\mathcal{M}(r)}{r}\right)^{-3/2}
\left(4\pi r\rho(r)-\frac{\mathcal{M}(r)}{r^{2}}\right)-4\pi\mathscr{G}\mathfrak{D}(r)
\left(1-\frac{2\mathscr{G}\mathcal{M}(r)}{r}\right)^{-3/2}=0
\end{align}
The variations leave invariant the entropy per nucleon $\mathscr{S}$ so that
\begin{align}
&\frac{d}{dr}|\rho(r)\mathfrak{D}^{-1}(r)|+p(r)\frac{d}{dr}|\mathfrak{D}^{-1}(r)|=\rho(r)+\frac{d}{dr}\mathfrak{D}^{-1}(r)+\mathfrak{D}^{-1}
\frac{d\rho(r)}{dr}+p(r)\frac{d}{dr}|\mathfrak{D}^{-1}(r)|\nonumber\\&
=-\rho(r)\mathfrak{D}^{-2}(r)\mathfrak{D}^{\prime}(r)+\mathfrak{D}^{-1}(r)
\rho^{\prime}(r)-p(r)\mathfrak{D}^{-2}(r)\mathfrak{D}^{\prime}(r)
\end{align}
so that
\begin{align}
\rho^{\prime}(r)=|p(r)+\rho(r)|\mathfrak{D}^{-1}(r)\mathfrak{D}^{\prime}(r)
\end{align}
or
\begin{align}
\mathfrak{D}^{\prime}(r)=\frac{\mathfrak{D}(r)\rho^{\prime}(r)}{|p(r)+\rho(r)|}
\end{align}
Using (9.24) in (9.21) and cancelling terms
\begin{align}
&-\frac{dp(r)}{dr}\frac{1}{|p(r)+\rho(r)|}+\mathscr{G}\left\lbrace 4 \pi r \rho(r)-\frac{\mathcal{M}(r)}{r}-4\pi r(p(r)+\rho(r))\right\rbrace\left(1-\frac{2\mathscr{G}\mathcal{M}(r)}{r}\right)^{-1}
\nonumber\\&=-\frac{dp(r)}{dr}\frac{1}{|p(r)+\rho(r)|}-\left(4\pi\mathscr{G}p(r)
-\mathscr{G}\frac{\mathcal{M}(r)}{ r}\right)\left(1-\frac{2\mathscr{G}\mathcal{M}(r)}{r}\right)^{-1}\\&-4\pi \mathscr{G} r \rho(r)\left(1-\frac{2\mathscr{G}\mathcal{M}(r)}{r}\right)^{-1}-4\pi \mathscr{G} r p(r)
\left(1-\frac{2\mathscr{G}\mathcal{M}(r)}{r}\right)^{-1/2}\nonumber\\&=-\frac{dp(r)}{dr}\frac{1}{|p(r)+\rho(r)|}+
\left(\mathscr{G}\mathcal{M}(r) r^{2} + 4\pi\mathscr{G} r p(r)\right)\left(1-\frac{2 \mathscr{G}\mathcal{M}(r)}{r}\right)^{-1}=0
\end{align}
Multiplying through by $r^{2}/(p(r)+\rho(r))$.
\begin{align}
-r^{2}\frac{dp(r)}{dr}&=\mathscr{G}(4\pi r^{3}p(r)+\mathcal{M}(r))(p(r)+\rho(r))\left(1-\frac{2 \mathscr{G}\mathcal{M}(r)}{r}\right)^{-1}\nonumber\\&
\equiv \mathscr{G}\mathcal{M}(r)\rho(r)\left(1+\frac{4\pi r^{3}p(r)}{\mathcal{M}(r)}
\right)\left(1+\frac{p(r)}{\rho(r)}\right)\left(1-\frac{2 \mathscr{G}\mathcal{M}(r)}{r}\right)^{-1}
\end{align}
The Tolman-Oppenheimer-Volkoff equation is then recovered which again is essentially the Euler hydrostatic equilibrium equation with three relativistic correction terms
\begin{align}
\frac{dp(r)}{dr}&=
-\frac{\mathscr{G}\mathcal{M}(r)\rho(r)}{r^{2}}\left(1+\frac{4\pi r^{3}p(r)}{\mathcal{M}(r)}
\right)\left(1+\frac{p(r)}{\rho(r)}\right)\left(1-\frac{2 \mathscr{G}\mathcal{M}(r)}{r}\right)^{-1}\nonumber\\&
=-\frac{\mathscr{G}\mathcal{M}(r)\rho(r)}{r^{2}}\mathcal{R}el_{1}(r)\mathcal{R}el_{2}(r)
\mathcal{R}el_{3}(r)
\end{align}
\end{proof}
\section{Maximum Entropy Derivations Of The Tolman-Oppenheimer-Volkoff Equations For Self-Gravitating Radiation And Relativistic Fluids/Gases}
The following theorem also appears in the physics literature $\bm{[46-50]}$. Here it is expanded and proved in greater mathematical detail. Again, the emphasis is on the mathematical derivation using a constrained optimisational-variational method. There are deep connections between thermodynamics and gravitation $\bm{[44,45]}$ although this not discussed here. But the TOVE can arise as a critical point of an entropy functional of a self-gravitating system of matter or radiation, and again is an exercise in variational calculus.
\subsection{Maximum entropy derivation the TOVE for self-gravitating radiation}
\begin{thm}
Let $\bm{\Omega}=\mathbb{B}(0,R)\subset\mathbb{R}^{3}$ or $\mathbb{B}(0,R)\subset\Sigma$ be a domain/ball of radius R and centred at zero, and $\Sigma$ is a Cauchy hypersurface. The domain $\mathbb{B}(0,R)$ gives support to a static gas of photons or radiation with density $\rho(r)$ and with boundary condition $\rho(R)=0$. The total mass/energy of the gas is $M$ and it has the energy momentum tensor
\begin{align}
\mathbf{T}_{\mu\nu}=\rho u_{\mu}u_{\nu}+\frac{1}{3}
\rho(g_{\mu\nu}+u_{\alpha}u_{\beta})
\end{align}
The energy density $\rho(r)$, the entropy density $\mathfrak{S}(r)$ and the entropy density 4-vector are given by
\begin{align}
&\rho(r)=b\Theta^{4}(r)\\&
\mathfrak{S}(r)=\frac{4}{3}b\Theta^{3}(r)\\&
\mathfrak{S}^{\mu}=\mathfrak{S}u^{\mu}
\end{align}
where $\Theta$ is the temperature of the gas. The entropy density in terms of the matter density is
\begin{align}
\mathfrak{S}=\frac{4}{3}\alpha\rho(r)
\end{align}
where $\alpha=\tfrac{4}{3}b^{1/4}$. The static spherically symmetric metric or solution of the Einstein equations is taken to be the interior Schwarzschild metric (with $\mathscr{G}=1$)
\begin{align}
ds^{2}=-\left(1-\frac{2|\mathcal{M}(r)|}{r}\right)dt^{2}
+\left((1-\frac{2|\mathcal{M}(r)|}{r}\right)dr^{2}+r^{2}d\Omega^{2}
\end{align}
where $\mathcal{M}(r)$ is as usual
\begin{align}
\mathcal{M}(r)=\int_{0}^{r}4\pi \bar{r}^{2}\rho(\bar{r})d\bar{r}
\end{align}
and $\mathcal{M}^{\prime}(r)=4\pi r^{2}\rho(r)$ with $\mathcal{M}(R)=M$ and $\mathcal{0}=0$. The total entropy of the photon gas within $\mathbb{B}(0,R)\subset\Sigma$ is
\begin{align}
&\mathcal{S}(r)=\int_{\mathbb{B}(0,R)}\mathfrak{S}^{\mu}n_{\mu}d\Sigma
=4\pi\alpha\int_{0}^{R}|\rho(r)|^{3/4}\left(1-\frac{2|\mathcal{M}(r)|}{r}
\right)^{-1/2}r^{2}dr\nonumber\\&
=(4\pi)^{1/4}\alpha \int_{0}^{R}\left(\frac{1}{r^{2}}\frac{d\mathcal{M}(r)}{dr}\right)^{3/4}
\left(1-\frac{2|\mathcal{M}(r)|}{r}
\right)^{-1/2}r^{2}dr\equiv (4\pi)^{1/4}\alpha I(\mathcal{M}(r^{\prime}),\mathcal{M}(r))
\end{align}
Then the 1st variation of the integral
\begin{align}
\delta I(\mathcal{M}(r^{\prime}),\mathcal{M}(r))=\delta\left| \int_{0}^{R}\left(\frac{1}{r^{2}}\frac{d\mathcal{M}(r)}{dr}\right)^{3/4}
\left(1-\frac{2|\mathcal{M}(r)|}{r}
\right)^{-1/2}r^{2}dr\right|\equiv
\int_{0}^{R}\mathcal{L}(\mathcal{M}^{\prime}(r),\mathcal{M}(r))=0
\end{align}
subject to the 'endpoint' constraints $\delta\mathcal{M}(0)=\delta\mathcal{M}(R)=0$
gives the TOVE equation for hydrostatic equilibrium of radiation
\begin{align}
\frac{1}{3}\frac{d\rho(r)}{dr}=-\frac{[\rho(r)+\tfrac{1}{3}\rho][\mathcal{M}(r)+\tfrac{4}{3}\pi r^{3} \rho]}{r(r-2\mathcal{M}(r))}
\end{align}
The static hydrostatic equilibrium configuration is therefore the 'critical point' of the entropy integral.
\end{thm}
\begin{proof}
The 1st variation $\delta I$ is equivalent to solving the Euler-Lagrange equations with the Lagrangian
\begin{align}
L=L(\mathcal{M}^{\prime}(r),\mathcal{M}(r))
=|\mathcal{M}^{\prime}(r)|^{3/4}\left(1-\frac{2|\mathcal{M}(r)|}{tr}\right)^{-1/2}r^{1/2}
\end{align}
Each term in the EL equations
\begin{align}
\frac{d}{dr}\left(\frac{\partial L}{\partial \mathcal{M}^{\prime}(x)}
\right)-\frac{\partial L}{\partial\mathcal{M}(r)}=0
\end{align}
is evaluated separately so that
\begin{align}
&\frac{\partial L}{\partial\mathcal{M}(r)}
=|\mathcal{M}^{\prime}(r)|^{3/4}\left(1-\frac{2|\mathcal{M}(r)|}{r}\right)^{-3/2} \frac{1}{r^{1/2}}\\&
\frac{\partial L}{\partial\mathcal{M}^{\prime}(r)}
=\frac{3}{4}|\mathcal{M}(r)|^{-1/4}\left(1-\frac{2|\mathcal{M}(r)|}{r}\right)^{-1/2}r^{1/2}
\\&
\frac{d}{dr}\left(\frac{\partial L}{\partial\mathcal{M}^{\prime}(r)}\right)=
\frac{-3r^{1/2}\left(\frac{2\mathcal{M}}{r^{2}}-\frac{2|\mathcal{M}^{\prime}(r)|}{r}
\right)}{ 8\left(1-\frac{2|\mathcal{M}(r)}{r}|\right)^{3/2}|\mathcal{M}^{\prime}|^{1/4}}+\frac{3}{8}
\frac{1}{r^{1/2}\left(1-\frac{2|\mathcal{M}(r)|}{r}\right)^{1/2}
|\mathcal{M}^{\prime}(r)|^{1/4}}\nonumber \\&-\frac{3}{16}\frac{r^{1/2}
\mathcal{M}^{\prime\prime}(r)}{\left(1-\frac{2|\mathcal{M}(r)|}{r}\right)^{1/2}
|\mathcal{M}^{\prime}|^{5/4}}
\end{align}
The EL equations are
\begin{align}
\frac{d}{dr}\left(\frac{\partial L}{\partial \mathcal{M}^{\prime}(r)}
\right)&-\frac{\partial L}{\partial\mathcal{M}(r)}=\frac{-3r^{1/2}\left(\frac{2\mathcal{M}(r)}
{r^{2}}-\frac{2|\mathcal{M}^{\prime}(r)|}{r}
\right)}{8\left(1-\frac{2|\mathcal{M}(r)}{r}|\right)^{3/2}|\mathcal{M}^{\prime}|^{1/4}}+\frac{3}{8}
\frac{1}{r^{1/2}\left(1-\frac{2|\mathcal{M}(r)|}{r}\right)^{1/2}
|\mathcal{M}^{\prime}(r)|^{1/4}}\nonumber \\&-\frac{3}{16}\frac{r^{1/2}
\mathcal{M}^{\prime\prime}(r)}{\left(1-\frac{2|\mathcal{M}(r)|}{r}\right)^{1/2}
|\mathcal{M}^{\prime}|^{5/4}}-\frac{|\mathcal{M}^{\prime}(r)|^{3/4}}
{\left(1-\frac{2|\mathcal{M}(r)|}{r}\right)^{3/2}r^{1/2}}\nonumber\\&\equiv
\frac{-3r^{1/2}\left(\frac{2\mathcal{M}(r)}{r^{2}}-\frac{2|\mathcal{M}^{\prime}(r)|}{r}
\right)}{8\left(1-\frac{2|\mathcal{M}(r)}{r}|\right)^{3/2}|\mathcal{M}^{\prime}|^{1/4}}
+\frac{3}{8}
\frac{\left(1-\frac{2|\mathcal{M}(r)|}{r}\right)}{r^{1/2}\left(1-\frac{2|\mathcal{M}(r)|}{r}
\right)^{3/2}
|\mathcal{M}^{\prime}(r)|^{1/4}}\nonumber \\&-\frac{3}{16}\frac{r^{1/2}
\mathcal{M}^{\prime\prime}(r)\left(1-\frac{2|\mathcal{M}(r)|}{r}\right)}
{\left(1-\frac{2|\mathcal{M}(r)|}{r}\right)^{3/2}
|\mathcal{M}^{\prime}(r)|^{5/4}}-\frac{|\mathcal{M}^{\prime}(r)|^{3/4}}
{\left(1-\frac{2|\mathcal{M}(r)|}{r}\right)^{3/2}r^{1/2}}\nonumber\\&\equiv
\frac{-3r^{1/2}\left(\frac{2\mathcal{M}(r)}{r^{2}}-\frac{2|\mathcal{M}^{\prime}(r)|}{r}
\right)}{|\mathcal{M}^{\prime}(r)|^{1/4}}
+\frac{3}{8}
\frac{\left(1-\frac{2|\mathcal{M}(r)|}{r}\right)}{r^{1/2}
|\mathcal{M}^{\prime}(r)|^{1/4}}\nonumber \\&-\frac{3}{16}\frac{r^{1/2}
\mathcal{M}^{\prime\prime}(r)\left(1-\frac{2|\mathcal{M}(r)|}{r}\right)}
{|\mathcal{M}^{\prime}(r)|^{5/4}}-\frac{|\mathcal{M}^{\prime}(r)|^{3/4}}
{r^{1/2}}=0
\end{align}
We now proceed carefully in a sequence of steps
\begin{enumerate}
\item Multiply through by $r^{1/2}$:
\begin{align}
\frac{-3r\left(\frac{2\mathcal{M}(r)}{r^{2}}-\frac{2|\mathcal{M}^{\prime}(r)|}{r}
\right)}{|\mathcal{M}^{\prime}(r)|^{1/4}}
+\frac{3}{8}
\frac{\left(1-\frac{2|\mathcal{M}(r)|}{r}\right)}{
|\mathcal{M}^{\prime}(r)|^{1/4}}-\frac{3}{16}\frac{r
\mathcal{M}^{\prime\prime}(r)\left(1-\frac{2|\mathcal{M}(r)|}{r}\right)}
{|\mathcal{M}^{\prime}(r)|^{5/4}}-{|\mathcal{M}^{\prime}(r)|^{3/4}}
=0
\end{align}
\item Multiply through by $|\mathcal{M}^{\prime}(r)|^{1/4}$ and expand the numerators of the first three terms:
\begin{align}
-\frac{3}{8}\frac{2|\mathcal{M}(r)|}{r}+\frac{3}{8} 2\mathcal{M}^{\prime}(r)+\frac{3}{8}-\frac{3}{4}\frac{|\mathcal{M}(r)|}{r}-\frac{3}{16}
\frac{r\mathcal{M}^{\prime\prime}(r)}{\mathcal{M}^{\prime}(r)}+\frac{3}{8}
\frac{\mathcal{M}^{\prime\prime}(r)|\mathcal{M}(r)|}{\mathcal{M}^{\prime}(r)}
-|\mathcal{M}(r)|=0
\end{align}
\item Multiply through by $\mathcal{M}^{\prime}(r)$
\begin{align}
-\frac{3}{4}\frac{\mathcal{M}(r)\mathcal{M}^{\prime}(r)}{r}
+\frac{3}{4}|\mathcal{M}^{\prime}(r)|^{2}+\frac{3}{8}|\mathcal{M}^{\prime}(r)
-\frac{3}{4}\frac{\mathcal{M}^{\prime}(r)\mathcal{M}(r)}{r}-\frac{3}{16}r\mathcal{M}
^{\prime\prime}(r)
+\frac{3}{8}\mathcal{M}^{\prime\prime}(r)\mathcal{M}(r)-|\mathcal{M}^{\prime}(r)|^{2}=0
\end{align}
\item Multiply through by $r$ and collect like terms:
\begin{align}
&-\overbrace{\frac{3}{4}\mathcal{M}(r)\mathcal{M}^{\prime}(r)}
+\frac{3}{4}|\mathcal{M}^{\prime}(r)|^{2}
+\frac{3}{8}\mathcal{M}^{\prime}(r)r-\overbrace{\frac{3}{4}\mathcal{M}^{\prime}(r)
\mathcal{M}(r)}
-\underbrace{\frac{3}{16}r^{2}\mathcal{M}^{\prime\prime}(r)}
+\underbrace{\frac{3}{8}r\mathcal{M}^{\prime\prime}(r)\mathcal{M}(r)}-r|\mathcal{M}(r)|^{2}
\nonumber\\&
=-\frac{3}{16}\mathcal{M}^{\prime}(r)r^{2}
+\frac{3}{8}\mathcal{M}^{\prime\prime}(r)\mathcal{M}(r)r+\frac{3}{8}\mathcal{M}^{\prime}(r)r
-\frac{1}{4}|\mathcal{M}^{\prime}(r)|^{2}r
-\frac{3}{2}\mathcal{M}^{\prime}(r)\mathcal{M}(r)=0
\end{align}
\end{enumerate}
Next, the derivatives $\mathcal{M}^{\prime}(r)$ and $\mathcal{M}^{\prime\prime}(r)$ are eliminated using $\mathcal{M}^{\prime}(r)=4\pi r^{2}\rho(r)$. Each term in (10.20) is carefully evaluated
\begin{align}
&\mathcal{M}^{\prime\prime}(r)=4\pi r^{2}\rho^{\prime}(r)+8\pi r\rho(r)
\nonumber\\&
\frac{3}{16}\mathcal{M}^{\prime\prime}(r)r^{2}=\frac{3}{16}.4\pi r^{4}\rho(r)+\frac{3}{16} 8\pi r^{3}\rho(r)\nonumber\\&
\frac{3}{8}\mathcal{M}^{\prime\prime}(r)\mathcal{M}(r)r=\frac{3}{8}4\pi r^{3}\rho^{\prime}(r)\mathcal{M}(r)+\frac{3}{8} 6\pi r\rho(r)\mathcal{M}(r) r\\&
\frac{3}{8}\mathcal{M}^{\prime}(r)r=\frac{3}{8} 4\pi r^{3}\rho(r)\\&
\frac{1}{4}\mathcal{M}^{\prime}(r)\mathcal{M}^{\prime}(r)r=\frac{1}{4} 4\pi r^{2}\rho(r) 4\pi r^{2}\rho(r)\\&
\frac{3}{2}\mathcal{M}^{\prime}(r)\mathcal{M}^{\prime}(r)=\frac{3}{2} 4\pi r^{2}\rho(r)\mathcal{M}(r)
\end{align}
Equation (10.20) then becomes
\begin{align}
&-\frac{3}{16} 4\pi r^{2}\rho^{\prime}(r)r^{2}-\underbrace{\frac{3}{2}\pi r^{3}\rho(r)}+\frac{3}{8} 4\pi r^{3}\rho^{\prime}(r)\mathcal{M}(r)+3 \pi r \rho(r)\mathcal{M}(r) r\nonumber\\& +\underbrace{\frac{3}{2}\pi r^{3}\rho(r)}-\pi r^{2}\rho(r) 4 \pi r^{3}\rho(r)-\frac{3}{2} 4 \pi r^{2}\rho(r)\mathcal{M}(r)=0
\end{align}
The following sequence of steps are then carried out.
\begin{enumerate}
\item  Divide out by $\pi r^{2}$
\begin{align}
&(-\frac{3}{4}\rho^{\prime}(r)r^{2}+\frac{3}{2} r\rho^{\prime}(r)\mathcal{M}(r))+3 \rho(r)\mathcal{M}(r)-6\rho(r)\mathcal{M}(r)-\rho(r) 4\pi r^{3}\rho(r)\nonumber\\&
=\big(-\frac{3}{4}\rho^{\prime}(r)r^{2}+\frac{3}{2}r\rho^{\prime}(r)\mathcal{M}(r)\big)-
3\rho(r)\mathcal{M}(r)-\rho(r) 4\pi r^{3}\rho(r)=0
\end{align}
\item Divide out by 3
\begin{align}
\big(-\frac{1}{4}\rho^{\prime}(r)r^{2}+\frac{1}{2}r\rho^{\prime}(r)
\mathcal{M}(r)\big)-\rho(r)\mathcal{M}(r)-\frac{4}{3}\pi r^{3}\rho(r)=0
\end{align}
\item Divide out by 3 again and move terms to the rhs, keeping derivative terms on the lhs so that
    \begin{align}
    \frac{1}{3}(-\frac{1}{4}\rho^{\prime}(r)r^{2}+\frac{1}{2}r\rho^{\prime}(r)
\mathcal{M}(r)\big)=\frac{1}{3}\rho(r)\mathcal{M}(r)+\frac{4}{3}\pi r^{3}\rho(r)(\rho(r)/3)
    \end{align}
which is
\begin{align}
-\frac{1}{3}.\frac{1}{4}r(r-2\mathcal{M}(r))\frac{d\rho(r)}{dr}=\frac{1}{3}\rho(r)
\mathcal{M}(r)+\frac{4}{3}\pi r^{3}\rho(r)(\rho(r)/3)
\end{align}
\item Multiply by $-4$ and factorise:
\begin{align}
&-\frac{1}{3}r\big(r-2\mathcal{M}(r)\big)\frac{d\rho(r)}{dr}
=\frac{4}{3}\rho(r)\mathcal{M}(r)+\frac{16}{3}\pi r^{3}\rho(r)(\rho(r)/3)\nonumber\\&
\equiv [\rho(r)\mathcal{M}(r)+\frac{1}{3}
\rho(r)\mathcal{M}(r)+4\pi r^{3}\rho(r)(\rho(r)/3)+\frac{4}{3}\pi r^{3}\rho(r)(\rho(r)/3)\nonumber\\&
=-\left[\rho(r)+\frac{1}{3}\rho(r)\right]\left[|\mathcal{M}(r)|+\frac{4}{3}\pi r^{3}(\rho(r)/3)\right]
\end{align}
\end{enumerate}
This is then the Tolman-Oppenheimer-Volkoff equation (-) for a self-gravitating photon gas
of density $\rho(r)$
\begin{align}
\underbrace{\frac{1}{3}\frac{d\rho(r)}{dr}=\frac{d}{dr}\left(\frac{\rho(r)}{3}\right)=\frac{-\left[\rho(r)+\frac{1}{3}\rho(r)\right]\left[|\mathcal{M}(r)|+\frac{4}{3}\pi r^{3}(\rho(r)/3)\right]}{r\big(r-2\mathcal{M}(r)\big)}}
\end{align}
\end{proof}
\subsection{Thermodynamic derivation of the TOVE for a perfect fluid/gas as a constrained optimisational-variational problem}
The following appears in $\bm{[47,50]}$
\begin{defn}
Let $\mathscr{S}$ be the entropy density, $\rho$ the mass density, $p$ the pressure, $\mathscr{C}$ the chemical potential, $\mathscr{B}$ the baryon number density and $\Theta$ the temperature of a perfect fluid/gas. The 1st law of thermodynamics is then
\begin{align}
d\mathscr{S}=\frac{d\rho}{\Theta}-\frac{\mathscr{C}}
{\Theta}d\mathscr{B}
\end{align}
The Gibbs-Duhem relation is
\begin{align}
\mathscr{S}=\frac{p+\rho-\mathscr{C}\mathscr{B}}{\Theta}
\end{align}
\end{defn}
The following preliminary lemma will be required.
\begin{lem}
If the Gibbs-Duhem relation and the 1st law of thermodynamics hold then
\begin{align}
p^{\prime}(r)=\mathscr{S}(r)\Theta^{\prime}(r)+\mathscr{B}\mathscr{C}^{\prime}(r)
\end{align}
\end{lem}
\begin{proof}
In spherical symmetry, the differential form of the Gibbs-Duhem relation is
\begin{align}
\Theta(r)d\mathscr{S}(r)+\mathscr{S}(r)d\Theta(r)=d\rho(r)+dp(r)-\mathscr{C}(r) d\mathscr{B}(r)-
\mathscr{B}(r)d\mathscr{C}(r)
\end{align}
If $d\mathscr{B}=0$ then
\begin{align}
\Theta d\mathscr{S}(r)+\mathscr{S}(r)d\Theta(r)=d\rho(r)+dp(r)-
\mathscr{B}(r)d\mathscr{C}(r)
\end{align}
From the first law, $\Theta d\mathscr{S}(r)=d\rho(r)$ so that
\begin{align}
d\rho(r)+\mathscr{S}(r) d\Theta(r)=d\rho(r)+dp(r)-
\mathscr{B}(r)d\mathscr{C}(r)
\end{align}
or
\begin{align}
\mathscr{S}(r)d\Theta(r)=dp(r)-\mathscr{B}(r)d\mathscr{C}(r)
\end{align}
It follows that
\begin{align}
\mathscr{S}(r)\frac{d\Theta(r)}{dr}=\frac{dp(r)}{dr}-\mathscr{B}(r)\frac{d\mathscr{C}(r)}{dr}
\end{align}
or equivalently
\begin{align}
p^{\prime}(r)=\mathscr{S}(r)\Theta^{\prime}(r)+\mathscr{B}\mathscr{C}^{\prime}(r)
\end{align}
\end{proof}
\begin{thm}
Let a perfect fluid/gas have support $\bm{\Omega}$ in a (Lorentzian) spacetime $(\mathcal{M}^{4},\mathbf{g})$. The following are assumed
\begin{enumerate}
\item The fluid/gas obeys the 1st  law and the Gibs-Duhem relation (-).
\item The system is spherically symmetric and the entropy maxima correspond to static equilibrium configurations.  The initial value Einstein constraint
equation for time symmetric data holds such that
\begin{align}
^{(3)}\mathbf{R}=\frac{2}{r^{2}}\frac{d}{dr}\left\lbrace r(1-\mathbf{g}^{-1}_{rr}\right\rbrace=16\pi \mathscr{G}\rho
\end{align}
so that the metric is the interior Schwarzchild metric
\begin{align}
ds^{2}=-\left(1-\frac{2 G|\mathcal{M}(r)|}{r}\right)dt^{2}
+\left((1-\frac{2 G|\mathcal{M}(r)|}{r}\right)^{-1}dr^{2}+r^{2}d\Omega^{2}
\end{align}
with $\mathbf{g}_{rr}=\left(1-\tfrac{2\mathscr{G}\mathcal{M}(r)}{r}\right)$ and $\mathcal{M}(\bar{r})=\int_{0}^{r}\rho(\bar{r})\bar{r}^{2}d\bar{r}$
\item At the boundary $p(R)=\rho(r)=0$ and $M=\mathcal{M}(R)$.
\item The entropy, baryon number and total mass are defined as
\begin{align}
&{S}=\int_{0}^{R}\mathscr{S}(r)|\mathbf{g}_{rr}(r)|^{1/2}4 \pi r^{2}dr=
\int_{0}^{R}\mathscr{S}(r)\left((1-\frac{2 G|\mathcal{M}(r)|}{r}\right)^{-1/2}4 \pi r^{2}dr\\&
{B}=\int_{0}^{R}\mathscr{B}(r)|\mathbf{g}_{rr}(r)|^{1/2}4 \pi r^{2}dr=
\int_{0}^{R}\mathscr{B}(r)\left((1-\frac{2 G|\mathcal{M}(r)|}{r}\right)^{-1/2}4 \pi r^{2}dr\\&
M=\int_{0}^{R}\rho(r)4\pi r^{2}dr
\end{align}
\end{enumerate}
Let $(\mathfrak{L}_{1},\mathfrak{L}_{2})$ be a pair of Lagrange multipliers. The TOVE for hydrostatic equilibrium then follows from the constrained optimisational-variational problem
\begin{align}
\delta S-\mathfrak{L}_{2}\delta M+\mathfrak{L}_{1}\delta\mathrm{B}=0
\end{align}
\end{thm}
\begin{proof}
First
\begin{align}
\delta|\mathbf{g}_{rr}|^{1/2}=\frac{\partial}{\partial r} |\mathbf{g}_{rr}|^{1/2}dr=
\mathscr{G}|\mathbf{g}_{rr}|^{3/2}\delta \mathcal{M}(r) 4\pi r dr
\end{align}
Then
\begin{align}
&\delta S-\mathfrak{L}_{2}\delta M+\mathfrak{L}_{1}\delta\mathrm{B}\nonumber\\&
=\int_{0}^{R}|\delta\mathscr{S}(r)|\mathbf{g}_{rr}|^{1/2} 4 \pi r^{2}dr
+\int_{0}^{R}\mathscr{S}\delta|\mathbf{g}_{rr}|^{1/2} 4\pi r^{2}dr\nonumber\\&
-\mathfrak{L}_{2}\int_{0}^{R}\delta\rho(r) r^{2}dr+\mathfrak{L}_{1}
\int_{0}^{R}\delta\mathscr{B}(r)|\mathbf{g}_{rr}|^{1/2}r^{2}dr+
\mathfrak{L}_{1}\int_{0}^{R}\mathscr{L}(r)\delta|\mathbf{g}_{rr}|^{1/2}\nonumber\\&
=\int_{0}^{R}|\delta\mathscr{S}(r)|\mathbf{g}_{rr}|^{1/2} 4 \pi r^{2}dr+
\mathscr{G}\int_{0}^{R}\mathscr{S}(r)|\mathbf{g}_{rr}|^{3/2}\delta\mathcal{M}(r) 4\pi r dr\nonumber\\&
-\mathfrak{L}_{2}\int_{0}^{R}\delta \rho(r) r^{2}dr+\mathfrak{L}_{1}
\int_{0}^{R}\delta\mathscr{B}(r)|\mathbf{g}_{rr}|^{1/2}r^{2}dr
+\mathfrak{L}_{1}\mathscr{G}\int_{0}^{R}\mathscr{S}(r)|\mathbf{g}_{rr}|^{3/2} 4\pi r dr
\end{align}
Using the previous thermodynamic relations
\begin{align}
&\delta\mathscr{S}(r)=\frac{\delta(r)}{\Theta(r)}-\frac{\mathscr{C}(r)\delta\mathscr{B}(r)}{\Theta(r)}\\&
\mathscr{S}(r)=\frac{p(r)+\rho(r)-\mathscr{C}(r)}{\Theta(r)}
\end{align}
\begin{align}
&\int_{0}^{R}\frac{1}{\Theta(r)}\delta\rho(r)|g_{rr}|^{1/2} r^{2}dr
+\mathscr{G}\int_{0}^{R}\left|\frac{\rho(r)+p(r)}{\Theta(r)}\right|\delta \mathcal{M}(r)
|\mathbf{g}_{rr}|^{3/2} r dr\nonumber\\&-\mathfrak{L}_{2}\int_{0}^{R}\delta\rho(r) r^{2} dr +\int_{0}^{R}\left(\mathfrak{L}_{1}-\left|\frac{\mathscr{C}(r)}{\Theta(r)}\right|\right)\delta\mathscr{B}(r)
|\mathbf{g}_{rr}|^{1/2} r ^{2} dr\nonumber\\&+\mathscr{G}\int_{0}^{R}
\left(\mathfrak{L}_{1}-\left|\frac{\mathscr{C}(r)}{\Theta(r)}\right|\right)\mathscr{B}(r)\delta\mathcal{M}(r)
|g_{rr}|^{3/2} r dr=0
\end{align}
Since $\delta\mathcal{M}\sim\delta\rho(r)$ then
\begin{align}
\bm{\mathfrak{L}}_{1}=\frac{\mathscr{C}(r)}{\Theta(r)}
\end{align}
so the last two terms in (-) will vanish leaving
\begin{align}
&\int_{0}^{R}\frac{1}{\Theta(r)}\delta\rho(r)|g_{rr}|^{1/2} r^{2}dr
+\mathscr{G}\int_{0}^{R}\left|\frac{\rho(r)+p(r)}{\Theta(r)}\right|
|\mathbf{g}_{rr}|^{3/2} r \left(\int_{0}^{r}\delta\rho(\bar{r}) 4\pi \bar{r}^{2} d\bar{r}
\right)dr\nonumber\\&-\mathfrak{L}_{2}\int_{0}^{R}\delta\rho(r) r^{2} dr =0
\end{align}
which then implies that
\begin{align}
\int_{0}^{R}dr|\delta\rho(r)|r^{2}\left\lbrace \frac{1}{\Theta(r)}|\mathbf{g}_{rr}|^{1/2}
+4\pi \mathscr{G}\int_{0}^{R}\left|\frac{p(\bar{r})+\rho(\bar{r})}{\Theta(\bar{r})}
\right|\mathbf{g}_{rr}(\bar(r))|^{3/2}r dr-\mathfrak{L}_{2}
\right\rbrace=0
\end{align}
where the integration order has been changed. Solving for the Lagrange multiplier then gives
\begin{align}
\mathfrak{L}_{2}=\frac{1}{\Theta(r)}|\mathbf{g}_{rr}|^{1/2}+4\pi \mathscr{G}\int_{0}^{R}
\left|\frac{\rho(r)+p(r)}{\Theta(r)}\right||g_{rr}(r)|^{3/2}\bar{r} d\bar{r}
\end{align}
If $\Theta_{o}=1/\mathfrak{L}_{2}$ and $r=R$ then the Tolman temperature is recovered, and $\Theta_{o}$ is the surface temperature measured by an observer at infinity. Now since $\mathfrak{L}_{1}=\mathscr{C}(r)/\Theta(r)$ then $\mathscr{C}^{\prime}(r)=\mathfrak{L}_{1}\Theta(r)$. Then from the (differentiated) Gibbs-Duhem relation
\begin{align}
&p^{\prime}(r)=\mathscr{S}(r)\Theta(r)+\mathscr{B}(r)\mathscr{C}^{\prime}(r)\nonumber\\&
=\mathscr{S}(r)\Theta^{\prime}(r)+\mathscr{B}(r)\mathfrak{L}_{1}\Theta^{\prime}(r)
=\Theta^{\prime}(r)\mathscr{S}(r)+\mathfrak{L}_{1}\mathscr{B}(r))\nonumber\\&
=\Theta^{\prime}(r)\left(\left|\frac{|p(r)+\rho(r)|}{\Theta(r)}\right|
-\frac{\mathscr{C}(r)\mathscr{B}(r)}{\Theta}+\frac{\mathscr{C}(r)\mathscr{B}(r)}{\Theta(r)}\right)\\&
=\Theta^{\prime}(r)\left(\left|\frac{|p(r)+\rho(r)|}{\Theta(r)}\right|\right)
\end{align}
Hence
\begin{align}
\frac{\Theta^{\prime}(r)}{\Theta(r)}=\frac{p^{\prime}(r)}{|p(r)+\rho(r)|}
\end{align}
Now differentiating (10.55) and substituting (10.58)
\begin{align}
&\frac{d\mathfrak{L}_{2}}{dr}=\frac{1}{\Theta(r)}
\frac{d}{dr}|\mathbf{g}_{rr}|^{1/2}+|\mathbf{g}_{rr}|^{1/2}\frac{d}{dr}
\left|\frac{1}{\Theta(r)}\right|+
4\pi\mathscr{G}\frac{d}{dr}\left|\int_{0}^{R}\frac{p(r)+
\rho(r)}{\Theta(r)}|\mathbf{g}_{rr}|^{3/2}\bar{r}d\bar{r}\right|\nonumber\\&
=\frac{1}{\Theta(r)}\frac{d}{dr}|\mathbf{g}_{rr}|^{1/2}+|\mathbf{g}_{rr}|^{1/2}\frac{d}{dr}
\left|\frac{1}{\Theta(r)}\right|+4 \pi \mathscr{G}\left|\frac{p(r)+
\rho(r)}{\Theta(r)}|\mathbf{g}_{rr}|^{3/2}{r}\right|\nonumber\\&
=\frac{1}{\Theta(r)}\frac{d}{dr}|\mathbf{g}_{rr}|^{1/2}+|\mathbf{g}_{rr}|^{1/2}
(-|\Theta(r)|^{-2}\Theta^{\prime}(r))+\frac{1}{\Theta(r)}(4\pi\mathscr{G} pr+4\pi\mathscr{G} G \rho(r) r)|\mathbf{g}_{rr}|^{3/2}=0
\end{align}
which is
\begin{align}
&-\frac{\Theta^{\prime}(r)}{|\Theta(r)|^{2}}
\left(1-\frac{2\mathscr{G}\mathcal{M}(r)}{r}\right)^{-1/2}-\frac{1}{\Theta(r)}
\left(1-\frac{2\mathscr{G}\mathcal{M}(r)}{r}\right)^{-3/2}\left(-4\pi
\mathscr{G}\rho(r) r+\frac{G\mathcal{M}(r)}{r^{2}}\right)\nonumber\\&
-\frac{1}{\Theta(r)}(4\pi\mathscr{G} p(r) r+4\pi\mathscr{G} \rho(r) r)\left(1-\frac{2\mathscr{G}\mathcal{M}(r)}{r}=0
\right)^{-3/2}
\end{align}
Cancelling out the $\frac{1}{\Theta(r)}$ and $(1-2G\frac{\mathcal{M}(r)}{r})^{-1/2}$
terms then gives
\begin{align}
&-\frac{p^{\prime}(r)}{|p(r)+\rho(r)|}-
\left(1-\frac{2\mathscr{G}|\mathcal{M}(r)|}{r}\right)^{-1}\left(-4\pi \mathscr{G}\rho(r) r^{2}+\frac{\mathscr{G}|\mathcal{M}(r)|}{r}\right)\nonumber\\&-(4\pi\mathscr{G}p(r) r+4\pi \mathscr{G} \rho(r) r)\left(1-\frac{2\mathscr{G}|\mathcal{M}(r)|}{r}\right)^{-1}=0
\end{align}
This can then be arranged to give the Tolman-Oppenheimer-Volkoff equation
\begin{align}
\frac{dp(r)}{dr}=-\frac{\mathscr{G}\mathcal{M}(r)\rho(r)}{r^{2}}\left(1+\frac{p(r)}{\rho(r)}\right)
\left(1+\frac{4\pi r^{3} p(r)}{\mathcal{M}(r)}
\right)\left(1-\frac{2\mathscr{G}\mathcal{M}(r)}{r}\right)^{-1}
\end{align}
\end{proof}
\subsection{Integrability of the TOVE for incompressible stars and the sharp Buchdal bound}
The TOVE is only integrable for incompressible stars for which $\rho(r)=\rho=const.$ Such stars do not actually exist but this assumption enables a solution of the Einstein equations and also a sharp upper bound to be established for the quantity $\mathscr{G}M/R$. This bound, the Buchdal bound $\bm{[10]}$, also applies to all stars in the Universe and not just relativistic stars. It also sets an upper limit on the expected gravitational red shift of spectral lines from a star's surface.
\begin{thm}$(\underline{Buchdal Bound}$\newline
let $\bm{\Omega}=\mathbb{B}_{3}(R)\subset\mathbb{R}^{3}$ be a ball off radius $R$ which supports a perfect fluid/gas of constant density $\rho(r)=\rho$ and pressure $p(r)$ and total mass $M$. Then
\begin{align}
\mathcal{M}(r)=\int_{0}^{r}4\pi\rho\bar{r}^{2}d\bar{r}=\frac{4}{3}\pi r^{3}\rho
\end{align}
and $M=\mathcal{M}(R)=\tfrac{4}{3}\pi R^{3}\rho$. The TOVE is then integrable so that:
\begin{enumerate}
\item The pressure within the star is
\begin{align}
p(r)=\frac{3 M}{4\pi R^{3}}\frac{[1-(2\mathscr{G}M/R)]^{1/2}-[1-2\mathscr{G}M r^{2}/R^{3}]^{1/2}}{[1-(2\mathscr{G}M r^{2}/R^{3})]^{1/2}-3[1-2\mathscr{G}M/R]^{1/2}}
\end{align}
or equivalently
\begin{align}
p(x)=\rho\frac{\sqrt{1-x^{2}}-\sqrt{1-x_{o}^{2}}}
{3\sqrt{1-x_{o}^{2}}-\sqrt{1-x^{2}}}                                                                                     \end{align}
where $x=r\sqrt{\tfrac{8}{3}\pi\mathscr{G}\rho}$ and $x_{o}=R\sqrt{\frac{8}{3}\pi \mathscr{G}\rho}=\sqrt{2\mathscr{G}M/R^{2}}=r_{0}/R$, where $r_{0}$ is the square root of the event horizon radius $R_{o}=\sqrt{2\mathscr{G}M}$.
\item The central pressure for $r=0$ or $x=0$ is
\begin{align}
p_{c}=p(0)=\frac{3M}{4\pi R^{3}}\left( \frac{1-(2\mathscr{G}M/R))^{1/2}-1}{ 1-3[1-2\mathscr{G}M/R)^{1/2}}\right)
\end{align}
or equivalently
\begin{align}
p_{c}=p(0)=\frac{1-\sqrt{1-r_{0}/R}}{3\sqrt{1-r_{0}/R-1}}
\end{align}
\item The central pressure is infinite if $\mathscr{G}M/R=4/9$ so that there is a sharp upper bound
\begin{align}
\frac{\mathscr{G}M}{R}<\frac{4}{9}
\end{align}
\end{enumerate}
\end{thm}
\begin{proof}
Rewrite the TOVE with $\rho(r)=\rho$ so that
\begin{align}
&\frac{dp(r)}{dr}=-\frac{\mathscr{G}\mathcal{M}(r)\rho(r)}{r^{2}}\left(1+\frac{p(r)}{\rho(r)}\right)
\left(1+\frac{4\pi r^{3} p(r)}{\mathcal{M}(r)}
\right)\left(1-\frac{2\mathscr{G}\mathcal{M}(r)}{r}\right)^{-1}\nonumber\\&
=-\frac{\mathscr{G}}{r^{2}}[\frac{4}{3}\pi r^{3}\rho][p(r)+\rho]\left(1+ \frac{4\pi r^{3}p(r)}{\frac{4}{3}\pi r^{3}\rho}\right)\left(1-\frac{2\mathscr{G}}{r}\frac{4}{3}\pi r^{3}\rho\right)^{-1}\nonumber\\&
=4\pi\mathscr{G} r[\rho+p(r)][\frac{\rho}{3}+p(r)]\left[1-\frac{8}{3}\mathscr{G} r^{2}\rho\right]^{-1}
\end{align}
This expression is then integrated $\emph{inwards}$ from the surface at $r=R$ where $p(R)=0$ so that
\begin{align}
&-\int_{p(R)=0}^{p(r)}d\bar{p(r)}(p(r)+\rho)^{-1}(\frac{1}{3}\rho+\bar{p(r)})^{-1}
=\int_{p(r)}^{0}d\bar{p(r)}(p(r)+\rho)^{-1}
(\frac{1}{3}\rho+\bar{p(r)})^{-1}\\&
=\int_{R}^{r}4\pi \mathscr{G}\bar{r}\left[1-\frac{8}{3}\mathscr{G}\bar{r}^{2}\rho\right]d\bar{r}
\end{align}
Performing both integrals gives
\begin{align}
&\frac{3}{2\rho}\log\frac{\rho+p(r)}{3p(r)+\rho}=
-\frac{3}{4\rho}
\log\left|\frac{\frac{8}{3}\pi \mathscr{G} r^{2}-1}{\frac{8}{3}\pi
\mathscr{G} R^{2}-1}\right|\nonumber\\&=
\frac{3}{4\rho}
\log\left|\frac{\frac{8}{3}\pi \mathscr{G} R^{2}-1}{\frac{8}{3}\pi
\mathscr{G} r^{2}-1}\right|=
\frac{3}{4\rho}
\log\left|\frac{-[\frac{8}{3}\pi \mathscr{G} R^{2}-1]}{-[\frac{8}{3}\pi
\mathscr{G} r^{2}-1]}\right|\nonumber\\&=
\frac{3}{4\rho}
\log\left|\frac{[1-\frac{8}{3}\pi \mathscr{G} R^{2}]}{[1-\frac{8}{3}\pi
\mathscr{G} r^{2}]}\right|
\end{align}
Hence
\begin{align}
\log\left|\frac{\rho+p(r)}{3p(r)+\rho}\right|=\frac{1}{2}\log\left|\frac{[1-\frac{8}{3}\pi \mathscr{G} R^{2}]}{[1-\frac{8}{3}\pi
\mathscr{G} r^{2}]}\right|=\log\left|\frac{[1-\frac{8}{3}\pi \mathscr{G} R^{2}]}{[1-\frac{8}{3}\pi
\mathscr{G} r^{2}]}\right|^{1/2}
\end{align}
so that
\begin{align}
\left|\frac{\rho+p(r)}{3p(r)+\rho}\right|=\left|\frac{[1-\frac{8}{3}\pi \mathscr{G} R^{2}]}{[1-\frac{8}{3}\pi\mathscr{G}r^{2}]}\right|^{1/2}
\end{align}
Using $\rho=3M/4\pi R^{3}$ for $r<R$ and solving for $p(r)$ then gives formula (-) for the pressure. Setting $r=0$ or $x=0$ then gives the formulas (-) or (-) for the central pressure. The central pressure becomes infinite if
\begin{align}
3\sqrt{1-r_{o}/R}-1=0
\end{align}
which is $\mathscr{G}M/R=4/9$. This then establishes the Buchdal bound (10.68).
\end{proof}
\begin{rem}
It can be seen for any sufficiently massive star, regardless of how exotic its matter composition (Eg., 'quark stars') and regardless of it equation of state, can reach a density such that $GM/R=4/9$ at which point it would require an infinite pressure to support it against its own gravity. Gravity then totally dominates and the star cannot be physically supported. It will then collapses through its event horizon radius to a point, essentially becoming a black hole. Note also that in the Newtonian case with constant density, the central pressure can only become infinite when $R=0$ whereas in the relativistic case, the infinite central pressure will occur for a finite radius for which $R=\tfrac{9}{4}\mathscr{G}M$.
\end{rem}

\end{document}